\documentclass{article}
\usepackage{fullpage}

\usepackage[english]{babel}
\usepackage{blindtext}
\usepackage[T1]{fontenc}

\usepackage{amsmath}
\usepackage{amssymb}
\usepackage{amsthm}
\usepackage[hide]{chato-notes}
\usepackage{cite}
\usepackage{color}
\usepackage{graphicx}
\usepackage{xspace}
\usepackage{algorithmic}
\algsetup{linenosize=\small}
\usepackage[ruled,lined,linesnumbered,noend]{algorithm2e}

\SetCommentSty{mycommfont}
\usepackage{epsfig}
\usepackage{epstopdf}
\usepackage[font={bf}]{caption}
\usepackage{xcolor}
\usepackage{balance}
\usepackage{enumitem}
\usepackage{framed}

\usepackage{hyperref}
\hypersetup{
	colorlinks,breaklinks,
	urlcolor=[rgb]{0,0.25,0.75},
	linkcolor=[rgb]{0,0.25,0.75},
	citecolor=[rgb]{0,0.25,0.75}
}

\DeclareMathOperator*{\argmax}{arg\,max}

\newcommand{\blue}[1]{\textcolor{black}{#1}}

\newcommand{\spara}[1]{\smallskip\noindent\textbf{#1.}}

\newcommand{\comment}[1]{}
\newcommand{\eat}[1]{}

\newcommand{\InFullOnly}[1]{}
 
\newtheorem{theorem}{Theorem}

\newtheorem{claim}{Claim}

\newtheorem{definition}{Definition}

\newtheorem{lemma}{Lemma}
\newtheorem{example}{Example}

\newtheorem{problem}{Problem}

\newtheorem{proposition}{Proposition}

\def\noflash#1{\setbox0=\hbox{#1}\hbox to 1\wd0{\hfill}}
\newcommand{\PP}{P\xspace}

\newcommand{\calA}{\mathcal{A}}
\newcommand{\calB}{\mathcal{B}}

\newcommand{\qao}{q_{\calA|\emptyset}\xspace}
\newcommand{\qbo}{q_{\calB|\emptyset}\xspace}
\newcommand{\qab}{q_{\calA|\calB}\xspace}
\newcommand{\qba}{q_{\calB|\calA}\xspace}

\newcommand{\comic}{Com-IC\xspace}
\newcommand{\model}{Com-IC\xspace}

\newcommand{\timsim}{$\mathsf{RR}$-$\mathsf{SIM}$\xspace}
\newcommand{\timsimfast}{$\mathsf{RR}$-$\mathsf{SIM}$+\xspace}
\newcommand{\timcim}{$\mathsf{RR}$-$\mathsf{CIM}$\xspace}

\newcommand{\SIM}{\textsc{SelfInfMax}\xspace}
\newcommand{\CIM}{\textsc{CompInfMax}\xspace}

\newcommand{\EPT}{\mathit{EPT}\xspace}

\newcommand{\OPT}{\mathit{OPT}\xspace}
\newcommand{\rom}[1]{\uppercase\expandafter{\romannumeral #1\relax}}
\newcommand{\lrom}[1]{\lowercase\expandafter{\romannumeral #1\relax}}

\newcommand{\flix}{{Flixster}\xspace}
\newcommand{\douban}{{Douban}\xspace}

\newcommand{\dbBook}{{Douban-Book}\xspace}
\newcommand{\dbMovie}{{Douban-Movie}\xspace}
\newcommand{\lastfm}{{Last.fm}\xspace}
\newcommand{\bQ}{\mathbf{Q}}

\newcommand{\highdeg}{$\mathsf{HighDegree}$\xspace}
\newcommand{\pagerank}{$\mathsf{PageRank}$\xspace}
\newcommand{\rand}{$\mathsf{Random}$\xspace}
\newcommand{\generalTIM}{$\mathsf{GeneralTIM}$\xspace}
\newcommand{\greedy}{$\mathsf{Greedy}$\xspace}
\newcommand{\vanillaIC}{$\mathsf{VanillaIC}$\xspace}
\newcommand{\copying}{$\mathsf{Copying}$\xspace}

\newcommand{\squishlist}{
 \begin{list}{$\bullet$}
  {  \setlength{\itemsep}{0pt}
     \setlength{\parsep}{3pt}
     \setlength{\topsep}{3pt}
     \setlength{\partopsep}{0pt}
     \setlength{\leftmargin}{2em}
     \setlength{\labelwidth}{1.5em}
     \setlength{\labelsep}{0.5em}
} }
\newcommand{\squishlisttight}{
 \begin{list}{$\bullet$}
  { \setlength{\itemsep}{0pt}
    \setlength{\parsep}{0pt}
    \setlength{\topsep}{0pt}
    \setlength{\partopsep}{0pt}
    \setlength{\leftmargin}{2em}
    \setlength{\labelwidth}{1.5em}
    \setlength{\labelsep}{0.5em}
} }

\newcommand{\squishdesc}{
 \begin{list}{}
  {  \setlength{\itemsep}{0pt}
     \setlength{\parsep}{2pt}
     \setlength{\topsep}{2pt}
     \setlength{\partopsep}{0pt}
     \setlength{\leftmargin}{2em}
     \setlength{\labelwidth}{1.5em}
     \setlength{\labelsep}{0.5em}
} }

\newcommand{\squishdesctight}{
 \begin{list}{}
  {  \setlength{\itemsep}{0pt}
     \setlength{\parsep}{0pt}
     \setlength{\topsep}{0pt}
     \setlength{\partopsep}{0pt}
     \setlength{\leftmargin}{1em}
     \setlength{\labelwidth}{1.5em}
     \setlength{\labelsep}{0.5em}
} }

\newcommand{\squishnumlist} {
\newcounter{qcounter}
\begin{list}{\arabic{qcounter}.~}{\usecounter{qcounter}} 
{  \setlength{\itemsep}{0pt}
    \setlength{\parsep}{0pt}
    \setlength{\topsep}{0pt}
    \setlength{\partopsep}{0pt}
    \setlength{\leftmargin}{2em}
    \setlength{\labelwidth}{1.5em}
    \setlength{\labelsep}{0.5em}
}}

\newcommand{\squishend}{
  \end{list}
}

\begin{document}
\title{{\bf From Competition to Complementarity:\\Comparative Influence Diffusion and Maximization}\footnote{An abridged version of this article is to appear in the Proceedings of VLDB Endowment (PVLDB), Volume 9, No.~2, and the 42nd International Conference on Very Large Data Bases (VLDB 2016), New Delhi, India, September 5--9, 2016. The copyright of the PVLDB version is held by the VLDB Endowment.}
}
%\subtitle{[Full Technical Report]}
%\titlenote{HELLO WORLD}

\author{
\begin{tabular}{ccc}
Wei Lu$^\dag$ \hspace{5mm} & Wei Chen$^\ddag$ \hspace{5mm}  & Laks V.S. Lakshmanan$^\dag$ \\
\end{tabular}
\\$ $\\
\begin{tabular}{ccc}
$^\dag${University of British Columbia}  & $ \qquad $ & $^\ddag${Microsoft Research }\\
{Vancouver, B.C., Canada} & $ \qquad $ &{Beijing, China}\\
{\tt \{welu, laks\}@cs.ubc.ca} & $ \qquad $ & {\tt weic@microsoft.com}
\end{tabular}
}

%\numberofauthors{3}
%\author{
%\alignauthor
%Wei Lu\\
%        \affaddr{University of British Columbia}\\
%       \affaddr{Vancouver, B.C., Canada}\\
%       \email{welu@cs.ubc.ca}
%\alignauthor
%Wei Chen\\
%       \affaddr{Microsoft Research}\\
%       \affaddr{Beijing, China}\\
%       \email{weic@microsoft.com}
%\alignauthor Laks V.S. Lakshmanan\\
%       \affaddr{University of British Columbia}\\
%       \affaddr{Vancouver, B.C., Canada}\\
%       \email{laks@cs.ubc.ca}
%}

\maketitle

\begin{abstract}
Influence maximization is a well-studied problem that asks for a small set of influential users from a social network, such that by targeting them as early adopters, the expected total adoption through influence cascades over the network is maximized.
However, almost all prior work focuses on cascades of a single propagating entity or
	purely-competitive entities.
In this work, we propose the {\em Comparative Independent Cascade} (\comic) model
	that covers the full spectrum of entity interactions from competition to complementarity.
In \comic, users' adoption decisions depend not only on edge-level information propagation, but also on a node-level automaton whose behavior is governed by a set of model parameters,
	enabling our model to capture not only competition, but also complementarity, to {\sl any possible degree}.
We study two natural optimization problems, \emph{Self Influence Maximization} and \emph{Complementary Influence Maximization}, in a novel setting with complementary entities.
Both problems are NP-hard, and we devise efficient and effective approximation algorithms via non-trivial techniques based on reverse-reachable sets and a novel ``sandwich approximation'' strategy.
The applicability of both techniques extends beyond our model and problems.
Our experiments show that the proposed algorithms {consistently} 
outperform intuitive baselines on four real-world social networks, often by a significant margin.
In addition, we learn model parameters from real user action logs.
\end{abstract}
 
%However, almost all existing influence diffusion models focus either on
%	the propagation of a single entity or purely-competitive entities.

%In this work, we propose a new stochastic diffusion model called \emph{Comparative Independent Cascade} (\model), which captures a full range of possibilities in terms of relationships amongst different propagating entities.
%In this work, we extend existing diffusion models, which only capture 
%	diffusions of a single entity or purely competitive entities, to cover the full spectrum of
%	entity interactions from competition to complementarity, for which we call
%	{\em Comparative Independent Cascade} (\model) model.

\clearpage
\section{Introduction}\label{sec:intro}
Online social networks are ubiquitous and play an essential role in our daily life.
Fueled by popular applications such as viral marketing, there has been extensive research in influence and information propagation in social networks, from both theoretical and practical points of view.
A key computational problem in this field is {\em influence maximization}, which asks to identify a small set of $k$ influential users (also known as {\em seeds}) from a given social network, such that by targeting them as early adopters of a new technology, product, or opinion, the expected number of total adoptions triggered by social influence cascade (or, propagation) is maximized~\cite{kempe03, infbook}.
The dynamics of an influence cascade are typically governed by a 
{\em stochastic diffusion model}, which specifies
how adoptions propagate from one user to another in the network.

Most existing work focuses on two types of diffusion models --- {\em single-entity models} and {\em pure-competition models}.
A single-entity model has only one propagating entity for social network users to adopt:
the classic Independent Cascade (IC) and Linear Thresholds (LT) models~\cite{kempe03} belong to this category.
These models, however, ignore complex social interactions involving multiple propagating entities.
Considerable work has been done to extend IC and LT models to study competitive influence maximization, but almost all models assume that the propagating entities are in pure competition and users adopt at most one of them~\cite{ChenNegOpi11,HeSCJ12,BudakAA11,PathakBS10,borodin10,BharathiKS07,CarnesNWZ07, lu2013}.

In reality, the relationship between different propagating entities is certainly more general than pure competition.
In fact, consumer theories in economics have two well-known notions: {\em substitute goods} and {\em complementary goods}~\cite{snyder08}.
Substitute goods are similar ones that compete, and can be purchased, one in place of the other, e.g.,  smartphones of various brands. 
Complementary goods are those that tend to be purchased together, e.g,. iPhone and its accessories, computer hardware and software, etc.
There are also varying {\em degrees} of substitutability and complementarity: buying a product could lessen the probability of buying the other without necessarily eliminating it; similarly, buying a product could boost the probability of buying another to any degree.
Pure competition only corresponds to the special case of perfect substitute goods.

The limitation of pure-competition models can be exposed by the following example.
Consider a viral marketing campaign featuring  iPhone 6 and Apple Watch.
It is vital to recognize the fact that Apple Watch generally needs an iPhone to be usable, and iPhone's user experience can be greatly enhanced by a pairing Apple Watch  (see, e.g., \url{http://bit.ly/1GOqesc}).
Clearly none of the pure-competition models is suitable for this campaign because they do not even allow users to adopt both the phone and the watch!
This motivates us to design a more powerful, expressive, yet reasonably  tractable model that captures not only competition, but also complementarity, and to any possible degrees associated with these notions.

To this end, we propose the {\em Comparative Independent Cascade} model, or \model for short, which, unlike most existing diffusion models, consists of two critical components that work jointly to govern the dynamics of diffusions: edge-level information propagation and a {\em Node-Level Automaton} (NLA) that ultimately makes adoption decisions based on a set of model parameters, known as the {\em Global Adoption Probabilities} (GAPs). 
Of these, edge-level propagation is similar to the propagation captured by the classical IC and LT models, but only controls {\sl information awareness}. The NLA is a novel feature and is unique to our proposal. Indeed, the term ``comparative'' comes from the fact that once a user is aware, via edge-level propagation, of multiple products, intuitively she makes a comparison between them by ``running'' her NLA. Notice that ``comparative'' subsumes ``competitive'' and ``complementary'' as special cases. 
In theory, the \model model is able to accommodate any number of propagating entities (items) and cover the entire spectrum from competition to complementarity between pairs of items, reflected by the values of GAPs.

In this work, as the first step toward comparative influence diffusion and viral marketing, we focus on the case of two items. 
At any time, w.r.t.\ any item $\calA$, a user in the network is in one of the following four states: 
\emph{$\calA$-idle}, \emph{$\calA$-suspended}, \emph{$\calA$-rejected}, or \emph{$\calA$-adopted}.
The NLA sets out probabilistic transition rules between states, and different GAPs are applied based on \textsl{a given user's state w.r.t.\ the other item $\calB$ and the relationship between $\calA$ and $\calB$.}
Intuitively, competition (complementarity) is modeled as reduced probability 
	(resp., increased probability) of adopting the second item
	after the first item is already adopted.
After a user adopts an item, she propagates this information to her neighbors in the network, making them aware of the item.
The neighbor may adopt the item with a certain probability, as governed by her NLA.

We then define two novel optimization problems for two complementary items $\calA$ and $\calB$. 
Our first problem, {\em Self Influence Maximization (\SIM)}, asks for $k$ seeds for $\calA$ such that given a fixed set of $\calB$-seeds, the expected number of $\calA$-adopted nodes is maximized.
The second one, {\em Complementary Influence Maximization (\CIM)}, considers the flip side of \SIM:
given a fixed set of $\calA$-seeds, find a set of $k$ seeds for $\calB$ such that 
	the expected increase in $\calA$-adopted nodes thanks to $\calB$ is maximized.
%{\sl To the best of our knowledge, \CIM is the first problem to focus on maximizing complementary effects in influence maximization and viral marketing literature.}
{\sl To the best of our knowledge, we are the first to systematically study
	influence maximization for complementary items.}  %in the influence maximization and viral marketing literature.}
%\weic{I think both problems are novel. Laks seems to also make this comment.}

We show that both problems are NP-hard under \model.
Moreover, two important properties, 
		{\em submodularity} and {\em monotonicity} (see \textsection\ref{sec:related}),
		which would allow a greedy approximation algorithm frequently used for influence maximization,
		do not hold in unrestricted \model model.
Even when we restrict \model to
	mutual complementarity, submodularity still does not hold in general.

To circumvent the aforementioned difficulties, we first show that {submodularity} %and monotonicity 
holds for a subset of the complementary parameter space.
We then make a non-trivial extension to the {\em Reverse-Reachable Set} (RR-set) techniques~\cite{borgs14,tang14,tang15},
originally proposed for influence maximization with single-entity models,
%-- proposed by \cite{borgs14} and improved by \cite{tang14,tang15} for influence maximization with the IC and LT models -- 
to obtain effective and efficient approximation solutions to both \SIM and \CIM.
Next, we propose a novel {\em Sandwich Approximation} (SA) strategy which, for a given non-submodular set function, provides an upper bound function and/or a lower bound function, and uses them to obtain data-dependent approximation solutions w.r.t.\ the original function.
{We further note that both techniques are applicable to a larger context beyond the
	model and problems studied in this paper: 
	for RR-sets, we provide a new definition and general sufficient conditions
	not covered by~\cite{borgs14,tang14,tang15} that apply to a large family of 
	influence diffusion models, while SA applies to the maximization of any non-submodular functions that are upper- and/or lower-bounded by submodular functions.}

%
%\red{We emphasize that {\sl neither technique is exclusive for our own model and problems}:
%The RR-set generalization provides a general solution framework for a large family of influence diffusion  models,
% while SA is virtually applicable to the maximization of any non-submodular functions
% }
% \blue{that are upper- or lower-bounded by submodular functions.}

In experiments, we first learn GAPs from user action logs from two social networking sites -- Flixster.com and Douban.com.
We demonstrate that our approximation algorithms based on RR-sets and SA techniques %significantly 
{consistently}
outperform several intuitive baselines, {typically by a significant margin} on real-world networks.

To summarize, we make the following contributions:

\begin{itemize}
	\item We propose the \model model to characterize influence diffusion dynamics of products with arbitrary degree of competition or complementarity, and identify a subset of the parameter space under which submodularity and monotonicity of influence spread hold, paving the way for approximation algorithms (\textsection\ref{sec:model} and \textsection\ref{sec:submod}).
	\item We propose two novel problems -- Self Influence Maximization and Complementary Influence Maximization -- for complementary products under the \model model (\textsection\ref{sec:problem}).
	\item We show that both problems are NP-hard, and devise efficient and  effective approximation solutions
	{by non-trivial extensions to RR-set techniques and by proposing Sandwich Approximation}, both having applicability beyond this work (\textsection\ref{sec:rrset}). 
%\note[Laks]{what are we saying is non-trivial? the rr-sets technique proposed by others or our adaptation to that technique?} 
	\item We conduct empirical evaluations on four real-world social networks and demonstrate the superiority of our algorithms over intuitive baselines, and also propose a methodology for learning global adoption probabilities for the \model model from user action logs of social networking sites (\textsection\ref{sec:exp}).
\end{itemize}

For better readability, most of the proofs, as well as some additional theoretical and experimental results are presented in the appendix.

%The rest of the paper includes \textsection\ref{sec:related} that provides the background and related work, and \textsection\ref{sec:concl} that summarizes the paper and discusses future work.
%For lack of space, we omit some technical proofs and additional results; they are presented in the full version of this paper~\cite{comicTR}.

\eat{It is worth mentioning that although past work has suggested that the Linear Threshold (LT) model is better-suited to describe dynamics of product adoptions, while IC is more suitable for describing phenomena similar to the spread of epidemics or viruses, we first explore possibilities with the IC model for its simplicity and the fact that this research direction is almost completely uncharted (to the best of our knowledge).}
%Later, we will consider similar extensions to the LT model.

\section{Background \& Related Work}\label{sec:related}
Given a graph $G=(V,E,p)$ where $p: E \to [0,1]$ specifies pairwise influence probabilities (or weights) between nodes,  and $k\in \mathbb{Z}_+$, the {\em influence maximization} problem asks to find a set $S\subseteq V$ of $k$ seeds, activating which leads to the maximum 
	expected number of active nodes (denoted $\sigma(S)$)~\cite{kempe03}.
Under both IC and LT models, this problem is NP-hard;
Chen et al.~\cite{ChenWW10, ChenWW10b} showed computing $\sigma(S)$ exactly for any $S\subseteq V$ is \#P-hard. 
Fortunately, $\sigma(\cdot)$ is a {\em submodular} and {\em monotone} function of $S$ for both IC and LT, which allows a simple greedy algorithm with an approximation factor of $1-1/e-\epsilon$, for any $\epsilon > 0$~\cite{kempe03, submodular}.
A set function $f: 2^U \to \mathbb{R}_{\geq 0}$ is submodular if for any $S\subseteq T\subseteq U$ and any $x \in U \setminus T$, $f(S\cup\{x\}) - f(S) \geq f(T\cup\{x\}) - T(S)$, and monotone if $f(S) \leq f(T)$ whenever $S\subseteq T\subseteq U$.
Tang et al.~\cite{tang14,tang15} proposed new randomized approximation algorithms
which are orders of magnitude faster than the original greedy algorithms in \cite{kempe03}. 

In competitive influence maximization~\cite{ChenNegOpi11,HeSCJ12,BudakAA11,PathakBS10,borodin10,BharathiKS07,CarnesNWZ07, lu2013} (also surveyed in \cite{infbook}), 
	a common thread is the focus on pure competition, 
	which only allows users to adopt at most one product or opinion.
Most works are from the follower's perspective~\cite{BharathiKS07,CarnesNWZ07, HeSCJ12,BudakAA11}, i.e., given competitor's seeds, how to maximize one's own spread, or minimize the competitor's spread.
Lu et al.~\cite{lu2013} aims to maximize the total influence spread of all competitors while ensuring fair allocation.

For viral marketing with non-competing items,
Datta et al.~\cite{dattaMS10} studied influence maximization with  items whose propagations are independent.
Narayanam et al.~\cite{narayanam2012viral} studied a setting with two sets of products, where a product can be adopted by a node only when it has already adopted a corresponding product in the other set.
Their model extends LT. 
We depart by defining a significantly more powerful and expressive model in \model, compared to theirs which only covers the special case of {\em perfect complementarity}.
\InFullOnly{Our technical contributions for addressing the unique challenges posed by \model are substantially different from \cite{narayanam2012viral}, which follows the typical route as in~\cite{kempe03}.}

%The work in 
Myers and Leskovec  analyzed Twitter data to study the effects of different cascades on users and predicted the likelihood of a user adopting a piece of information given the cascades to which the user was previously exposed~\cite{myers12}.
McAuley et al. used logistic regression to learn substitute/complementary relationships between products from user reviews~\cite{mcauley15}.
Both studies focus on data analysis and behavior prediction and do not provide
	diffusion modeling for competing and complementary items, nor do they study the
	influence maximization problem in this context.

\section{Comparative Independent Cascade Model}\label{sec:model}
\subsubsection*{Review of Classical IC Model}
In the IC model~\cite{kempe03}, there is just one entity (e.g., idea or product) being propagated through the network. 
%We use ``product'' in our description of the model for concreteness and simplicity. 
An instance of the model has a directed graph $G=(V,E,p)$ where % with nodes $V$, edges $E\subset V\times V$, and function 
$p:E \to [0,1]$, %specifies pairwise influence probabilities between nodes,
and a seed set $S\subset V$.
For convenience, we use $p_{u,v}$ for $p(u,v)$.
At time step $0$, the seeds are {\em active} and all other nodes are {\em inactive}.
Propagation proceeds in discrete time steps.
At time $t$, every node $u$ that became active at $t-1$ makes one attempt to activate each of its inactive out-neighbors $v$.
This can be seen as node $u$ ``testing'' if the edge $(u,v)$ is ``live'' or ``blocked''.
The out-neighbor $v$ becomes active at $t$ iff the edge is live. 
The propagation ends when no new nodes become active.

\begin{figure}[t]
 \centering
     \includegraphics[width=0.66\textwidth]{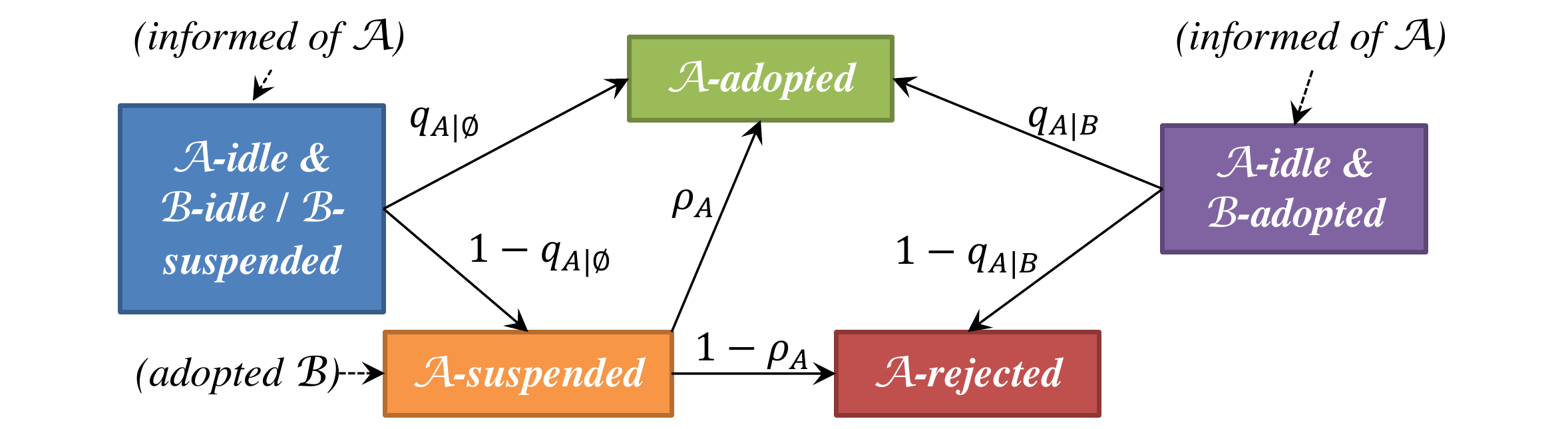}
%\vspace{-3mm}
 \caption{Com-IC model: Node-level automaton for product/item $\calA$.} \label{fig:nla}
%\vspace{-5mm}
\end{figure}

\subsubsection*{Key differences from IC model}
In the {\em Comparative IC model} (\comic), there are at least two products.
For ease of exposition, we focus on just two products $\calA$ and $\calB$ below.
Each node can be in any of the states \{\emph{idle}, \emph{suspended}, \emph{adopted}, \emph{rejected}\} w.r.t. {\sl each} of the products.
All nodes are initially in the joint state of ($\calA$-idle, $\calB$-idle).
One of the biggest differences between \comic and IC is the separation of information diffusion (edge-level) and the actual adoption decisions (node-level).
Edges only control the information that flows to a node: e.g., when $u$ adopts a product, its out-neighbor $v$ may be informed of this fact.
Once that happens, $v$ uses its own node level automaton (NLA) to decide which state to transit to.
This depends on $v$'s current state w.r.t. the two products as well as parameters corresponding to the state transition probabilities of the NLA, namely the Global Adoption Probabilities, defined below.

A concise representation of the NLA is in Figure~\ref{fig:nla}. 
Each state is indicated by the label.
The state diagram is self-explanatory. 
E.g., with probability $\qao$, a node transits from a state where it's $\calA$-idle to $\calA$-adopted, regardless of whether it was $\calB$-idle or $\calB$-suspended. %, or $\calB$-rejected  in the beginning. 

From the $\calA$-suspended state, it transits to $\calA$-adopted w.p. $\rho_\calA$ and to $\calA$-rejected w.p. $1-\rho_\calA$. The probability $\rho_\calA$, called {\em reconsideration probability}, as well as the reconsideration process will be explained below. 
Note that in a \comic diffusion process defined below, not all joint state is reachable from the initial ($\calA$-idle, $\calB$-idle) state, e.g., ($\calA$-idle, $\calB$-rejected).
Since all unreachable states are irrelevant to adoptions, they are negligible (details in Appendix).

\subsubsection*{Global Adoption Probability (GAP)}
The Global Adoption Probabilities (GAPs), consisting of four parameters
	$\bQ = (\qao, \qab, \qbo, \qba) \in [0,1]^4$ are important parameters of the NLA which decide the likelihood of adoptions after a user is informed of an item. 
%\eat{ 
%More specifically,
% $\qao$ (resp.\ $\qbo$) is the probability that a user adopts $\calA$ (resp.\ $\calB$)  given that she is neither $\calA$-adopted nor $\calB$-adopted, and
%$\qab$  (resp.\ $\qba$) is the probability that a user adopts $\calA$ (resp.\ $\calB$)  given that she is already $\calB$-adopted (resp.\ $\calA$-adopted). 
%} 
%Specifically,
 $\qao$ is the probability that a user adopts $\calA$ given that she is informed of $\calA$ but not $\calB$-adopted, and
$\qab$  is the probability that a user adopts $\calA$ given that she is already $\calB$-adopted. A similar interpretation applies to $\qbo$ and $\qba$.  
%\eat{We denote by $\bQ$ any set of GAPs: $\{\qao, \qbo, \qab, \qba\}$, and futhermore define $\bQ^+$ and $\bQ^-$ :} 

Intuitively, GAPs reflect the overall popularity of products and how they are perceived by the entire market.
They are considered aggregate estimates and hence are not user specific in our model.
We provide further justifications at the end of this section and describe a way to learn GAPs from user action log data in \textsection\ref{sec:exp}.

GAPs enable \model to model competition and complementarity, to arbitrary degrees. 
We say that $\calA$ {\em competes} with $\calB$ iff $\qba \leq \qbo$.
Similarly, $\calA$ {\em complements} $\calB$ iff $\qba \geq \qbo$. 
We include the special case of $\qba=\qbo$ in both cases above for
	convenience of stating our technical results, and it actually means 
	that the propagation of $\calB$ is completely independent of $\calA$
	({\em cf}. Lemma~\ref{lemma:indiff}).
%
%\red{Both relationships include $\qba=\qbo$ as a special case.}  
Competition and complementarity in the other direction are similar. 
The degree of competition and complementarity is determined by the difference between the two relevant GAPs, i.e., $|\qba - \qbo|$ and $|\qab - \qao|$.
%Both relationships include $\qab = \qao$ and $\qba = \qbo$ as special cases.
%We will see that if, e.g., $\qab = \qao$, the propagation of $\calA$ is completely independent of $\calB$
%% then irrespective of $\calB$-seeds, the influence spread of $\calA$ achieved by any given $\calA$-seed set remains the same 
%({\em cf}.\ Lemma~\ref{lemma:indiff}).
%In fact, it is the probability distribution over the sets of $\calA$-adopted nodes that will not change, and neither will the spread. We refer to this situation as $\calA$ is {\em indifferent} to $\calB$.
{For convenience, 
%we use $\bQ$ to refer to any set of GAPs $\in [0,1]^4$.
we use
$\bQ^+$  to refer to any set of GAPs representing {\em mutual complementarity}: $(\qao \le \qab) \wedge (\qbo \le \qba)$, and similarly, $\bQ^-$ for GAPs representing {\em mutual competition}: $(\qao \ge \qab) \wedge (\qbo \ge \qba)$.}

%\begin{definition}
%$\bQ^+$ denotes any set of GAPs representing {\em mutual complementarity}: $(\qao \le \qab) \wedge (\qbo \le \qba)$ and $\bQ^-$ denote for any set of GAPs representing {\em mutual competition}: $(\qao \ge \qab) \wedge (\qbo \ge \qba)$.
%\end{definition}

\subsubsection*{Diffusion Dynamics}
Let $G=(V,E,p)$ be a directed social graph with pairwise influence probabilities.
Let $S_\calA, S_\calB \subset V$ be the seed sets for $\calA$ and $\calB$.
Influence diffusion under \comic proceeds in discrete
	time steps.
Initially, every node is $\calA$-idle and $\calB$-idle.
At time step $0$, every $u\in S_\calA$ becomes $\calA$-adopted and
	every $u\in S_\calB$ becomes $\calB$-adopted\footnote{
No generality is lost in assuming seeds adopt an item without testing the NLA: for every  $v\in V$, we can create two dummy nodes $v_\calA, v_\calB$ 
	and edges $(v_\calA, v)$ and $(v_\calB, v)$ with $p_{v_\calA, v} = p_{v_\calB, v} = 1$.
	Requiring seeds to go through NLA is equivalent to
	constraining that $\calA$-seeds ($\calB$-seeds) be 
	selected from all $v_\calA$'s (resp.\ $v_\calB$'s).}.
If $u \in S_\calA \cap S_\calB$, we randomly decide the order of $u$ adopting
	$\calA$ and $\calB$ with a fair coin.
For ease of understanding, we describe the rest of the diffusion process in a
	modular way in Figure~\ref{fig:model}.
We use $N^+(v)$ and $N^-(v)$ to denote the set of out-neighbors and in-neighbors of
	$v$, respectively.

We draw special attention to \emph{tie-breaking} and \emph{reconsideration}. 
Tie-breaking is used when a node's in-neighbors adopt different products and try to inform the node at
	the same step. 
Node reconsideration concerns the situation that a node $v$ did not adopt $\calA$
	initially but later after adopting $\calB$ it may reconsider adopting $\calA$:
	when $\calB$ competes with $\calA$ ($\qao\ge \qab$), $v$ will not reconsider
	adopting $\calA$, but when $\calB$ complements $\calA$ (specifically, $\qao < \qab$), 
	$v$ will reconsider adopting $\calA$.
In the latter case, the probability of adopting $\calA$, $\rho_\calA$, is defined
	in such a way that the overall probability of adopting $\calA$ is equal to $\qab$
	(since $\qab = \qao+(1-\qao)\rho_\calA$).

%%%%%%%%%%%%%%%%%%%%%%%%%%
\begin{figure}[t!]
\vspace{-2mm}

\begin{framed}
%\begin{small}

%\begin{description}[style=unboxed,leftmargin=0pt]

\vspace{-3pt}
\spara{Global Iteration}
At every time step $t\ge 1$, for all nodes that became $\calA$- or $\calB$-adopted at $t-1$, their outgoing edges are tested for transition (1 below). 
After that, for each node $v$ that has at least one in-neighbor (with a live edge) becoming $\calA$- and/or $\calB$-adopted at $t-1$, $v$ is tested for possible state transition (2-4 below).

	\begin{description}[style=unboxed,leftmargin=8pt]

\item[1.]
\textbf{Edge transition.}
For an untested edge $(u,v)$, flip a biased coin independently: $(u,v)$ is {\em live} w.p.\ $p_{u,v}$ and {\em blocked} w.p.\ $1-p_{u,v}$. Each edge is tested at most once in the entire diffusion process.

\item[2.]
\textbf{Node tie-breaking.}
Consider a node $v$  to be tested at time $t$.
Generate a random permutation $\pi$ of $v$'s in-neighbors (with live edges)
	that adopted at least one product at $t-1$.
Then, test $v$ with each such in-neighbor $u$ and 
	$u$'s adopted item ($\calA$ and/or $\calB$)
	following $\pi$. 
If there is a $w\in N^-(v)$ adopting both $\calA$ and $\calB$,
	then test both products, following
	their order of adoption by $w$.

\item[3.]
\textbf{Node adoption.}
Consider the case of testing an  $\calA$-idle node $v$ for adopting $\calA$ (Figure~\ref{fig:nla}).
If $v$ is {\em not} $\calB$-adopted, then w.p.\ $\qao$, it becomes $\calA$-adopted
	and w.p.\ $1-\qao$  it becomes $\calA$-suspended.
If $v$ is $\calB$-adopted,  then w.p.\ $\qab$, it becomes $\calA$-adopted
	and w.p.\ $1-\qab$  it becomes $\calA$-rejected.
The case of adopting $\calB$ is symmetric.

\item[4.]
\textbf{Node reconsideration.}
Consider an $\calA$-suspended node $v$ that just adopts $\calB$ at time $t$.
Define $\rho_\calA = \max\{\qab - \qao, 0\} / (1 - \qao)$.
Then, $v$ {\em reconsiders} to become $\calA$-adopted w.p.\ $\rho_\calA$, or
	$\calA$-rejected w.p.\ $1-\rho_\calA$.
The case of reconsidering $\calB$ is symmetric.

	\vspace{-8pt}
	\end{description}
%\end{description}

%\end{small}
\end{framed}
%\vspace{-4mm}
\caption{\model model: diffusion dynamics}
\label{fig:model}
%\vspace{-4mm}
\end{figure}

\subsubsection*{Design Considerations}
The design of \comic not only draws on the essential elements from a classical diffusion model (IC) stemming from mathematical sociology, but also closes a gap between theory and practice, in which diffusions typically do not occur just for one product or with just one mode of pure competition.
With GAPs in the NLA, the model can characterize {\sl any} possible relationship between two propagating entities: competition, complementarity, and any degree associated with them.
GAPs are fully capable of handling asymmetric relationship between products.
%Besides, the pairwise relationship between products is not necessarily symmetric, which can also be handled by GAPs.
E.g., an Apple Watch ($\calA$) is complemented more by an iPhone ($\calB$) than the other way round: many functionalities of the watch are not usable without a pairing iPhone, but an iPhone is totally functional without a watch.
This asymmetric complementarity can be expressed by any GAPs satisfying $(\qab - \qao) > (\qba - \qbo) \geq 0$.
%\comic explicitly separates information propagation (edge-level) and adoption decision-making (node-level).
%On the one hand, information are propagated via friendship links in a social network, and on the other hand, each node makes its own adoption decisions pursuant to various node rules.
Furthermore, 
introducing NLA with GAPs and separating the propagation of product information from actual adoptions reflects Kalish's famous characterization of new product adoption~\cite{kalish85}: customers go through two stages -- {\em awareness} followed by {\em actual adoption}.
In Kalish's theory, product awareness is propagated through word-of-mouth effects; after an individual becomes aware, she would decide whether to adopt the item based on other considerations. % only when her valuation (the maximum amount of money she is willing to pay) is no less than the price.
%Furthermore, \comic offers additional novelty in its capability of handling two or more products.
%The model also specifies that once the information channel (edge) from one user to another becomes live, 
%	it remains live regardless of future products $u$ will inform $v$ of.
%The activation of information channels 
Edges in the network can be seen as information channels from
	one user to another.
Once the channel is open (live), it remains so.
{This modeling choice is reasonable as competitive goods are typically of the same kind and complementary goods tend to be adopted together.}
%is more natural than the alternative of attempting to open an information channel every time a user adopts a new product.

We remark that \comic encompasses previously-studied single-entity and pure-competition models as special cases.
When $\qao = \qbo = 1$ and $\qab = \qba = 0$, \model reduces to the (purely) Competitive Independent Cascade model~\cite{infbook}.
If, in addition, $\qbo$ is $0$, the model further reduces to the classic IC model.

\section{Formal Problem Statements}\label{sec:problem}
Many interesting optimization problems can be formulated under the expressive 
	\comic model.
In this work, we focus on influence maximization with {\sl complementary propagating entities}, since competitive viral marketing has been studied extensively (see \textsection\ref{sec:related}).
In what follows, we propose two problems.
The first one, {\em Self Influence Maximization} (\SIM), is a natural extension to the classical influence maximization problem~\cite{kempe03}.
The second one is the novel {\em Complementary Influence Maximization} (\CIM), where the objective is to maximize complementary effects (or ``boost'' the expected number of adoptions) by selecting the best seeds of a complementing good.

Given the seed sets $S_\calA, S_\calB$, we first define $\sigma_\calA(S_\calA, S_\calB)$ and $\sigma_\calB(S_\calA, S_\calB)$ to be the expected number of $\calA$-adopted and $\calB$-adopted nodes, respectively under the \comic model.
Clearly, both $\sigma_\calA$ and $\sigma_\calB$ are real-valued bi-set functions mapping $2^V \times 2^V$ to $[0, |V|]$, for any fixed $\bQ$. Unless otherwise specified,  GAPs are not considered as arguments to $\sigma_\calA$ and $\sigma_\calB$ as $\bQ$ is constant in a given instance of \comic.
Also, for simplicity, we may refer to $\sigma_\calA(\cdot,\cdot)$ as $\calA$-spread and $\sigma_\calB(\cdot,\cdot)$ as $\calB$-spread.
{The following two problems are defined in terms of  $\calA$-spread, without loss of generality.}

\begin{problem}[\SIM]\label{prob:sim}
Given a directed graph $G=(V,E,p)$ with pairwise influence probabilities, $\calB$-seed set $S_\calB \subset V$, a cardinality constraint $k$, and a set of GAPs $\bQ^+$, find an $\calA$-seed set $S_\calA^* \subset V$ of size $k$, such that the {\em expected number of $\calA$-adopted nodes} is maximized under \model:
$S_\calA^* \in \argmax_{T \subseteq V, |T| = k} \sigma_\calA(T, S_\calB).$
\end{problem}

\SIM is obviously NP-hard, as it subsumes {\sc InfMax} under the
	classic IC model when $S_\calB = \emptyset$ and $\qao = \qab = 1$.
By a similar argument, it is \#P-hard to compute the exact value of
	$\sigma_\calA(S_\calA, S_\calB)$ and $\sigma_\calB(S_\calA, S_\calB)$
	for any given $S_\calA$ and $S_\calB$.

\begin{problem}[\CIM]\label{prob:cim}
Given a directed graph $G=(V,E,p)$ with pairwise influence probabilities, $\calA$-seed set $S_\calA \subset V$, a cardinality constraint $k$, and a set of GAPs  $\bQ^+$, find a $\calB$-seed set $S_\calB^* \subseteq V$ of size $k$ such that the {\em expected increase (boost) in $\calA$-adopted nodes} is maximized under \model:
$S_\calB^* \in \argmax_{T\subseteq V, |T| = k} [ \sigma_\calA(S_\calA, T) - \sigma_\calA(S_\calA, \emptyset) ].$
\end{problem}

%}

%
%\weil{I commented out the old remark here which says CIM is equivalent to $\argmax_{S_\calB\subseteq V, |S_\calB| = k} \sigma_\calA(S_\calA, S_\calB)$.  IMO this gives little useful insight.  I wrote another insight that's potentially more valuable for the paper (e.g. in experiments, the Y-axis for CIM spread is much smaller than that for SIM spreads).}

\def\theoremCIMHard{
\CIM is NP-hard.
}

\begin{theorem}\label{thm:hard}
{\theoremCIMHard}
\end{theorem}

\begin{proof}

Let $\bf I$ be an instance of {\sc InfMax} with the IC model,
	defined by a directed graph $G=(V,E)$ and budget $b < |V|$.
We define an instance $\bf J$ of \CIM as follows.
For each $v\in V$, create a dummy copy $v'$ and a directed edge $(v',v)$ with influence probability $p_{v',v} = 1$.
Let $V'$ be the set of all dummy nodes.
Set $S_\calA = V'$, $\qao = 0$, $\qab = \qbo = \qba = 1$.
The budget $b$ of selecting $\calB$ seeds in instance 
	$\bf J$ is the same $b$ as in instance $\bf I$.
	
\begin{claim}\label{claim:cimhard1}
Consider the \CIM instance $\bf J$.
For all $\cal B$-seed sets $S_\calB$, there exists a set $T_\calB \subseteq V'$
	such that $\sigma_\calA(S_\calA, T_\calB) \geq \sigma_\calA(S_\calA, S_\calB)$.
\end{claim}

\begin{proof}[Proof of Claim~\ref{claim:cimhard1}]
Let $v \in S_\calB \cap V$ be any regular node.
Assume for now that $v'\not\in S_\calB$. 
Let $S'_\calB$ denote the $\calB$-seed set obtained by replacing all such $v$ by
	its corresponding dummy copy $v'$.
Since $S_\calA=V'$, all dummy nodes are $\calA$-adopted and all regular nodes are
	$\calA$-informed. 
Choosing $v'$ to be a $\calB$-seed makes $v'$ and $v$ both $\calB$-adopted,
	so by reconsideration $v$ becomes $\calA$-adopted as well.
Clearly, the sets of regular nodes that are $\calB$-adopted 
	under the configurations $(S_\calA, S_\calB)$ and $(S_\calA, S'_\calB)$ are the same.
Thus,  sets of regular nodes that are $\calA$-adopted under these configurations
	are also the same, which verifies the claim. 

Now suppose $v, v'\in S_\calB$.
We claim that there exists a dummy node outside $S_\calB$, which follows from the fact that
	the budget $b < |V|=|V'|$.
Now, replace $v$ with any dummy node $u'$ not in $S_\calB$ and call the resulting $\calB$-seed set $S'_\calB$.
By an argument similar to the above case, we can see that every regular node that is $\calA$-adopted under 
	configuration $(S_\calA, S_\calB)$ is also $\calA$-adopted under configuration $(S_\calA, S'_\calB)$. 
There may be additional $\calA$-adopted regular nodes
	thanks to the seed $u'\in S'_\calB \setminus S_\calB$. 
Again, the claim follows.
Finally, if $S_\calB \cap V = \emptyset$ (no regular node in $S_\calB$),
	then the claim trivially holds as we can simply let $T_\calB = S_\calB$.
\end{proof}

Now consider any $T\subseteq V$ and let $\sigma_{IC}(T)$ denote the expected spread of
	$T$ under the IC model in $G$.
Let $T' = \{v'\in V' : v\in T\}$ be the set of corresponding dummy copies.
Since $S_\calA = V'$, all regular nodes (in $V$) would become $\cal A$-informed
	when $S_\calB = \emptyset$. %, and moreover, $\sigma_\calA(S_\calA, \emptyset) = |V'|$.
Now set $S_\calB = T'$, which would make all nodes in $T$ (regular) become $\calB$-adopted,
	%(recall that $\qbo =1$ and $p_{v',v} =1$ for all $v\in V$), 
	and then $\calA$-adopted (by reconsideration), and further propagate both $\calA$
	and $\calB$ over the network.
Hence,
\begin{align}\label{eqn:nphard}
\sigma_{IC}(T) = \sigma_\calA(S_\calA, T') - n,
\end{align}
where $n =_{\mathrm{def}} |V| = |V'|$.

%\weic{We cannot directly claim that $S^*_{\bf J}$ is a solution of $\bf I$,
%	because it may contain nodes in $V'$, which is not in the instance of
%	$\bf I$. I made some changes below.}

Next, to prove the theorem, it suffices to show:
\begin{claim}\label{claim:cimhard2}
The subset $T_*\subset V$ is an optimal solution to the {\sc InfMax} (with IC model)
	instance ${\bf I}$ if and only if $T_*' \subset V'$ is an optimal solution to
	the \CIM instance $\bf J$ under the \comic model, where $T_*' = \{v'\in V' : v\in T_*\}$.
\end{claim}

\begin{proof}[Proof of Claim~\ref{claim:cimhard2}]
$(\Longrightarrow)$: Suppose for a contradiction that $T_*'$ is suboptimal to $\bf J$.
Let $X'$ be an optimal solution to $\bf J$ instead, which implies that
	$\sigma_\calA(S_\calA, X') > \sigma_\calA(S_\calA, T_*')$.

If $X'$ contains only dummy copies, let $X$ be the corresponding set of regular nodes.
Then, by Eq.~\eqref{eqn:nphard},
\[
\sigma_\calA(S_\calA, X') = \sigma_{IC}(X) + n > \sigma_\calA(S_\calA, T_*') = \sigma_{IC}(T_*) + n.
\]
Hence $\sigma_{IC}(X) > \sigma_{IC}(T_*)$, contradicting to the fact that $T_*$ is
	optimal to ${\bf I}$.

If $X'$ contains at least one regular node, then by repeated application of
	Claim~\ref{claim:cimhard1}, we can obtain a $\calB$-seed set with no regular nodes
	and furthermore dominates $X'$ w.r.t. $\calA$-spread.
Thus, w.l.o.g. we can ignore $\calB$-seed sets
	containing regular nodes. 
	
$(\Longleftarrow)$: Let $T_*'$ be an optimal solution to $\bf J$.
By the same domination argument, we can 
	assume w.l.o.g. that $T_*'$ consists of only dummy nodes.
Then, the optimality of $T_*$ for instance $\bf I$ can be argued
	using Eq.~\eqref{eqn:nphard} in an identical manner to the above. 
\end{proof}

This completes the proof that \CIM is NP-hard. 
\end{proof}

From the formulation of \CIM, we can intuitively see that the placement of $\calB$-seeds will
	be heavily dependent on the existing $\calA$-seeds.
For example, if $S_\calA$ and $S_\calB$ are in two different connected components of the graph,
	then evidently the boost is zero.
In contrast, if they are rather close and can influence roughly the same region in the graph,
	the boost is likely to be high.
Indeed, for the special case where $\qbo = 1$ and $k \geq |S_\calA|$, directly ``copying''
	$\calA$-seeds to be $\calB$-seeds will give the optimal boost.
However, this in no way diminishes the value of this problem, as 
	Theorem~\ref{thm:hard} shows the NP-hardness in the general case.
%\weic{I removed the mentioning of the case $\qbo < 1$, since we did not give an NP-hardness proof for this case.
%	Theorem 1 only proves the case when $k < |S_\calA|$.}

\def\theoremCIMOpt{
For \CIM, when $\qbo = 1$ and $k \geq |S_\calA|$, 
	we can solve the problem optimally by setting $S_\calB^*$ to be $S_\calA \cup X$, where
	$X$ is an arbitrary set in $V\setminus S_\calA$ with size $k - |S_\calA|$, that is
	$\sigma_\calA(S_\calA, S_\calA \cup X) = \max_{T\subseteq V, |T| = k} \sigma_\calA(S_\calA, T)$. 
}

\begin{theorem}\label{thm:cimOpt}
{\theoremCIMOpt}
\end{theorem}

The proof of this theorem relies on the possible world model, which is defined in the next section.
Hence, we defer to the proof to the appendix.

\section{Properties of \model}\label{sec:submod}

%In this subsection, we explore whether desired set function properties such as monotonicity and submodularity hold for the spread functions $\sigma_\calA$ and $\sigma_\calB$ under the Com-IC model.
%For simplicity and symmetry, we focus on $\sigma_\calA$ in our analysis, unless otherwise noted.
%Similar results apply to $\sigma_\calB$.

%For a set function $f: 2^V \rightarrow R$, $f$ is {\em monotonically increasing (resp. decreasing)} if for any subsets $S \subseteq T \subseteq V$,
%	$f(S) \le f(T)$ (resp. $f(S) \ge f(T)$); 
%	and $f$ is {\em submodular} if for any $S \subseteq T \subseteq V$ and any $u \in V \setminus T$, 
%	$f(S \cup \{u\}) - f(S) \ge f(T \cup \{u \}) - f(T)$.
Since neither of \SIM and \CIM can be solved in PTIME unless P = NP, we explore approximation algorithms by studying submodularity and monotonity for \model,
% it is natural to ask whether the spread functions $\sigma_\calA, \sigma_\calB$ satisfy the properties of submodularity and monotonicity since 
which may pave the way for designing approximation algorithms. % for \SIM and \CIM.
Note that $\sigma_\calA$ is a bi-set function taking arguments $S_\calA$ and $S_\calB$, so we analyze the properties w.r.t.\ each of the two arguments.
As appropriate, we refer to the properties of $\sigma_\calA$ w.r.t.\ $S_\calA$ ($S_\calB$) as {\em self-monotonicity} (resp., {\em cross-monotonicity}) and {\em self-submodularity} (resp., {\em cross-submodularity}).

\subsection{An Equivalent Possible World Model}\label{sec:pw}

To facilitate a better understanding of \model and our analysis on 
	submodularity, we define a {\em Possible World (PW) model} that provides an equivalent view of the \comic model.
%It will also help us study submodularity and monotonicity later in \textsection\ref{sec:submod}.
%Given an instance of the Com-IC model, a random possible world is generated as follows.
Given a graph $G=(V,E,p)$ and a diffusion model, a possible world consists of a {\em deterministic graph} sampled from a probability distribution over all subgraphs of $G$.
%For \model, $\Pr[(u,v) \in E'] = p_{u,v}$, independent of all other edges.
%The existence of each edge is independent and the probability of $(u,v)$ existing is $p_{u,v}$.
For % the \comic diffusion model, 
\comic, we also need some variables for each node to fix the outcomes of random events in relation to the NLA (adoption, tie-breaking, and reconsideration), so that {influence cascade is fully deterministic in a single possible world.}
%Let $\mathbf{W}$ be the set of all possible worlds and $W\in \mathbf{W}$ be any one of them.

\spara{Generative Rules}
%To generate a possible world $W$ with deterministic graph $G_W = (V,E_W)$, first, each edge $(u,v)$ is independently retained in $E_W$ w.p.\ $p_{u,v}$ -- a retained edge is a \emph{live} edge. 
Retain each edge $(u,v)\in E$ w.p. $p_{u,v}$ (\emph{live} edge) and drop it w.p. $1-p_{u,v}$ (\emph{blocked} edge).
This generates a possible world $W$ with $G_W=(V,E_W)$, $E_W$ being the set of live edges.
Next, for every node $v\in V$, we
\begin{enumerate}
\item choose ``thresholds'' $\alpha_\calA^{v,W}$ and $\alpha_\calB^{v,W}$ independently and uniformly at random from $[0,1]$, for comparison with GAPs in adoption decisions (when $W$ is clear from context, we write $\alpha_\calA^v$ and $\alpha_\calB^v$);

\item generate a random permutation $\pi_v^W$ of all in-neighbors $u \in N^-(v)$ (for tie-breaking);
\item sample a discrete value $\tau_v^W\in \{\calA, \calB\}$, where each value has a probability of $0.5$ (used for tie-breaking in case $v$ is a seed of both $\calA$ and $\calB$).
\end{enumerate}

\spara{Deterministic cascade in a PW}
At time step $0$, nodes in $S_\calA$ and $S_\calB$ first become
	$\calA$-adopted and $\calB$-adopted, respectively (ties,
	if any, are broken based on $\tau_v$).
Then, iteratively for each step $t\ge 1$, 
	a node $v$ becomes ``reachable'' by $\calA$ at time step $t$ 
	if $t$ is the length of a shortest path from any seed $u\in S_\calA$
	to $v$ consisting entirely of live edges and $\calA$-adopted nodes.
Node $v$ then becomes $\calA$-adopted at step $t$ 
	if $\alpha_\calA^v \leq x$, where $x = \qao$ if $v$ is not $\calB$-adopted, otherwise $x = \qab$.
	 %. (when $v$ is not yet $\calB$-adopted) or if $\alpha_\calA^v \leq \qab$ (when $v$ is already $\calB$-adopted).
For re-consideration, suppose $v$ just becomes $\calB$-adopted at step $t$, while being $\calA$-suspended (i.e., $v$ was reachable by $\calA$ before $t$ steps but $\alpha_\calA^v > \qao$).
Then, $v$ adopts $\calA$ if $\alpha_\calA^v \leq \qab$.
The reachability and reconsideration tests of $\calB$ are symmetric. 
For {\em tie-breaking}, if $v$ is reached by both $\calA$ and $\calB$
	at $t$, the permutation $\pi_v$ is used to determine
	the order in which $\calA$ and $\calB$ are considered.
In addition, if $v$ is reached by $\calA$ and $\calB$ from the same in-neighbor, e.g., $u$, then the informing order follows the order in which $u$ adopts $\calA$ and $\calB$.

The following lemma establishes the equivalence between this possible world model
	and \comic.
This allows us to analyze monotonicity and submodularity using the PW model only.
%All proofs can be found in	the full version~\cite{comicTR} of this paper.

\def\lemmaPW{
For any fixed $\calA$-seed set $S_\calA$ and $\calB$-seed set $S_\calB$, the joint distributions of the sets of $\calA$-adopted nodes and $\calB$-adopted nodes obtained  $(i)$ by running a \comic diffusion from $S_\calA$ and $S_\calB$ and $(ii)$  by randomly sampling a possible world $W$ and running a deterministic cascade from $S_\calA$ and $S_\calB$ in $W$,
	are the same.
}
\begin{lemma}\label{lemma:pw}
{\lemmaPW}
\end{lemma}

\spara{Equivalence Classes of Possible Worlds}
Notice that in a possible world, $\alpha_\calA^v$ and $\alpha_\calB^v$ are both
	real values in the interval $[0,1]$.
Thus in theory, the total number of possible worlds is infinite.
In what follows, we establish a finite number of {\em equivalence classes} of possible worlds.
This facilitates theoretical analysis of the model, e.g., the proof of Theorem~\ref{thm:cimOpt}
	makes use of this property (see appendix).

From the perspective of influence propagation dynamics under \comic,
	especially the final states of each node w.r.t.\ both products,
	\textsl{the exact values of $\alpha_\calA^v$ and $\alpha_\calB^v$ do not matter.
Instead, we only need to know the ranges that $\alpha_\calA^v$ and $\alpha_\calB^v$ fall in.
For $\alpha_\calA^v$, the ranges are
\begin{itemize}
\item $[0, \qao)$, $[\qao, \qab)$, $[\qab,1]$ when $\qao \le \qab$,
\item  $[0, \qab)$, $[\qab, \qao)$, $[\qao,1]$ when $\qab < \qao$.
\end{itemize}
Likewise for $\alpha_\calB^v$, they are:
\begin{itemize}
\item $[0, \qbo)$, $[\qbo, \qba)$, $[\qba,1]$ when $\qbo \le \qba$,  
\item $[0, \qba)$, $[\qba, \qbo)$, $[\qbo,1]$ when $\qba < \qbo$.
\end{itemize}
}

Given two possible worlds $W_1$ and $W_2$, we say that are {\em equivalent} iff each node $v\in V$ satisfies all of the following conditions:
\begin{enumerate}
\item $\alpha_\calA^{v, W_1}$ and $\alpha_\calA^{v, W_2}$ fall into the same range;
\item $\alpha_\calB^{v, W_1}$ and $\alpha_\calB^{v, W_2}$ also fall into the same range;
\item $\pi_v^{W_1} = \pi_v^{W_2}$ (same order of in-neighbours in both worlds);
\item $\tau_v^{W_1} = \tau_v^{W_2}$.
\end{enumerate}

Clearly, for any two possible worlds $W_1$ and $W_2$ in the same equivalence class, by model definition we have
\[
\sigma_\calA^{W_1} (S_\calA, S_\calB) = \sigma_\calA^{W_2} (S_\calA, S_\calB), \text{ and } 
\sigma_\calB^{W_1} (S_\calA, S_\calB) = \sigma_\calB^{W_2} (S_\calA, S_\calB).
\]
for any $\calA$-seed set $S_\calA$ and $\calB$-seed set $S_\calB$.

In view of the above, for any equivalence class ${\bf W}$ of possible worlds, by $\sigma_\calA^{{\bf W}}(S_\calA, S_\calB)$, without any ambiguity we mean $\sigma_\calA^{W}(S_\calA, S_\calB)$, where $W\in{\bf W}$ is any possible world in the equivalence class ${\bf W}$. 
It is straightforward to verify that the total number of equivalence classes is finite, as opposed to the total number of possible worlds, which is uncountable.

Let $\Pr[{\bf W}]$ denote the total probability mass of all the possible worlds that belong to $\bf W$.
%Equivalently, for an arbitrary possible world $W$, $\Pr[{\bf W}]$ is the probability that $W \in \bf W$.
In principle, this probability can be computed using integration over $\alpha_\calA^v$'s and $\alpha_\calB^v$'s,
	as there are still an uncountable number of possible worlds in any given equivalence class.
An alternative method that does not involve integration is to look at the ranges in which $\alpha_\calA^v$'s and $\alpha_\calB^v$'s
	fall.
For example, $\Pr[\qao \leq \alpha_\calA^v < \qab] = \qab - \qao$ (assuming complementarity).
This is correct as the $\alpha$-values are sampled uniformly at random from the interval $[0,1]$.
We can then express the expected spread function using a linear combination of all equivalence classes:
\begin{align}
\sigma_\calA(S_\calA, S_\calB) = \sum_{\bf W} \Pr[{\bf W}] \cdot \sigma_\calA^{\bf W} (S_\calA, S_\calB),
\end{align}
where $\sigma_\calA^{\bf W} (S_\calA, S_\calB)$ can be computed deterministically as we mentioned earlier
	in this subsection.

%Because each variable we only care about the three possible ranges, for every node  

\subsection{Monotonicity}
It turns out that when $\calA$ competes with $\calB$ while $\calB$ complements $A$, monotonicity does not hold in general (see Appendix for
	counter-examples).
But note that these cases are rather unnatural.
Hence, we focus on mutual competition
	($\bQ^-$) and mutural complementary cases ($\bQ^+$), for which we can show 
	self- and cross-monotonicity do hold.

\def\theoremMonotone{
For any fixed $\calB$-seed set $S_\calB$,  $\sigma_\calA(S_\calA, S_\calB)$ is monotonically increasing in $S_\calA$ for any set of GAPs in $\bQ^+$ and $\bQ^-$.
%if $\calA$ and $\calB$ either mutually competes with or mutually complements each other. 
Also, $\sigma_\calA(S_\calA, S_\calB)$ is monotonically increasing in $S_\calB$ for any GAPs in $\bQ^+$, and monotonically decreasing in $S_\calB$ for any $\bQ^-$.
}

\begin{theorem} \label{thm:monotone}
{\theoremMonotone}
\end{theorem}

\subsection{Submodularity in Complementary Setting}

Next, we analyze self-submodularity and cross-submodularity for mutual complementary cases ($\bQ^+$) that has direct impact on \SIM and \CIM.
The analysis for $\bQ^-$ is deferred to the appendix.

For self-submodularity, we show that it is satisfied in the case of {\em ``one-way complementarity''}, i.e.,  $\calB$ complements $\calA$ ($\qao \leq \qab$), but $\calA$ does not affect $\calB$ ($\qbo = \qba$), or vise versa (Theorem~\ref{thm:submod-complement}). %\footnote{One-way complementarity may also include the case of $\calB$ complements $\calA$ but $\calA$ competes with $\calB$. However, this setting is much less worthy for its unnaturalness.}.
%However, if mutual complementary is strict, i.e., there is no indifference, neither self-submodularity nor cross-submodularity is satisfied (counter-examples given).
%For the more general case (symmetric complementarity, where $\calB$ is also complemented by $\calA$), we conjecture that submodularity may not hold.
We will also show the $\sigma_\calA$ is cross-submodular in $S_\calB$ when $\qba = 1$ (Theorem~\ref{thm:cross-submod}).
However, both properties are not satisfied in general (see appendix for counter-examples).
We give two useful lemmas first, and thanks to
	Lemma~\ref{lemma:tie} below, we may assume w.l.o.g. that  tie-breaking always favors $\calA$
	in complementary cases.

\def\lemmaTie{
Consider any \comic instance with $\bQ^+$.
Given fixed $\calA$- and $\calB$-seed sets,
%the distribution of $\ell_\calA(v)$, for any $v\in V$, is independent of tie-breaking orders.
for all nodes $v\in V$, all permutations of $v$'s in-neighbors are equivalent in determining if
	$v$ becomes $\calA$-adopted and $\calB$-adopted, and thus the tie-breaking rule
	is not needed for mutual complementary case.
}

\begin{lemma}\label{lemma:tie}
{\lemmaTie}
\end{lemma}

\def\lemmaInDiff{
In the \comic model, if $\calB$ is indifferent to $\calA$ (i.e., $\qba = \qbo$), then for any fixed $\calB$ seed set $S_\calB$, the probability distribution over sets of $\calB$-adopted nodes is independent of $\calA$-seed set.
Symmetrically, the probability distribution over sets of $\calA$-adopted nodes is also independent of $\calB$-seed set if $\calA$ is indifferent to $\calB$.
}

\begin{lemma}\label{lemma:indiff}
{\lemmaInDiff}
\end{lemma}

%We now show that self-submodularity holds for one-way $\bQ^+$. 
%Assume w.lo.g.\ (thanks to Lemma~\ref{lemma:tie}) that tie-breaking always favors $\calA$.

\begin{theorem}\label{thm:submod-complement}
For any instance of Com-IC model with $\qao \leq \qab$ and $\qbo = \qba$,
$(i)$.\ $\sigma_\calA$ is {\em self-submodular} w.r.t.\ $\calA$ seed set $S_\calA$, for any fixed $\calB$-seed set $S_\calB$.
$(ii)$.\ $\sigma_\calB$ is {\em self-submodular} w.r.t.\ $\calB$ seed set $S_\calB$ and is independent of $\calA$-seed set $S_\calA$. 
\end{theorem}

\begin{proof}
First, $(ii)$ holds trivially.
By Lemma~\ref{lemma:indiff}, $\calA$ does not affect $\calB$'s diffusion in any sense.
Thus, $\sigma_\calB(S_\calA, S_\calB) = \sigma_\calB(\emptyset, S_\calB)$.
It can be shown that
	the function $\sigma_\calB(\emptyset, S_\calB)$ is both monotone and submodular w.r.t.\ $S_\calB$,
	for any $\qbo$, through a straightforward extension to the proof of Theorem 2.2 in  Kempe et al.~\cite{kempe03}. 
%\weil{Just to mention that the proof for submod of IC cannot be used per se, since we have $\qbo$ here.}

For $(i)$, first we fix a possible world $W$ and a $\calB$-seed set $S_\calB$.
Let $\Phi_\calA^W(S_\calA)$ be the set of $\calA$-adopted nodes in possible world $W$ with
	$\calA$-seed set $S_\calA$ ($S_\calB$ omitted when it is clear from the context).
%For self-submodularity, 
Consider two sets $S\subseteq T \subseteq V$, some node $u\in V\setminus T$, and finally a node $v \in \Phi_\calA^W(T \cup \{u\}) \setminus \Phi_\calA^W(T)$.
There must exist a live-edge path $\PP_\calA$ from $T \cup \{u\}$ consisting entirely of $\calA$-adopted nodes. %when $S_\calA= T \cup \{u\}$.
We denote by $w_0 \in T \cup \{u\}$ the origin of $\PP_\calA$.
%We next prove the following two claims:

%\begin{claim}\label{claim:submod-complement}
%$\PP_\calA$ remain $\calA$-adopted when $S_\calA = \{w_0\}$.
%\end{claim}
%\begin{proof} 
%\textsc{Proof of Claim~\ref{claim:submod-complement}}.
We first prove a key claim: $\PP_\calA$ remains $\calA$-adopted when $S_\calA = \{w_0\}$.
Consider any node $w_i \in \PP_\calA$.
In this possible world, if $\alpha_\calA^{w_i} \leq \qao$, then regardless of the diffusion of $\calB$, $w_i$ will adopt $\calA$ as long as its predecessor $w_{i-1}$ adopts $\calA$.
If $\qao < \alpha_\calA^{w_i} \leq \qab$, then there must also be a live-edge path $\PP_\calB$ from $S_\calB$ to $w_i$ that consists entirely of $\calB$-adopted nodes, and it boosts $w_i$ to adopt $\calA$.
Since $\qbo = \qba$, $\calA$ has no effect on $\calB$-propagation (Lemma~\ref{lemma:indiff}), and $\PP_\calB$ always exists and all nodes on $\PP_\calB$ would still be $\calB$-adopted through $S_\calB$ (fixed) irrespective of $\calA$-seeds. 
Thus, $\PP_\calB$ always boosts $w_i$ to adopt $\calA$ as long as $w_{i-1}$ is $\calA$-adopted.
Hence, the claim holds by a simple induction on $\PP_\calA$ starting from $w_0$. % \qed
%\end{proof}

Then, it is easy to see $w_0$ = $u$.
Suppose otherwise, then $w_0 \in T$ must be true.
By  the claim above %Claim~\ref{claim:submod-complement} 
and self-monotonicity of $\sigma_\calA$ (Theorem~\ref{thm:monotone}), $v \in \Phi_\calA^W(\{w_0\})$ implies $w\in \Phi_\calA^W(T)$, a contradiction. 
Therefore, we have $v \not\in \Phi_\calA^W(S)$ and $v \in \Phi_\calA^W(S \cup \{u\})$.
This by definition implies $|\Phi_\calA^W(\cdot)|$ is submodular for any $W$ and $S_\calB$, which
	is sufficient to show that $\sigma_\calA(S_\calA, S_\calB)$ is submodular in $S_\calA$.
\end{proof}
%Since $\sigma_\calA(S_\calA) = \sum_W \Pr[W] |\Phi_\calA^W(\cdot)|$, it is itself submodular.
%\weil{Need to replace $\sum$ with integral since the number of possible worlds are infinite due to $\alpha$}
%\qed 

\def\lemmaClaimBA{
On any $\calA$-path $\PP_\calA$, if some node $w$ adopts $\calB$ %(perhaps under some other $\calB$ seed set) 
and all nodes before $w$ on $\PP_\calA$ are $\calA$-ready, then every node following $w$ on $P_\calA$ adopts both $\calA$ and $\calB$, regardless of the actual $\calB$-seed set.
}

\def\claimBPath{
There is a $\calB$-path $\PP_\calB$ from some $\calB$-seed $x_0 \in T \cup\{u\}$ to $w$, such that
	even if $x_0$ is the only $\calB$-seed, $w$ still adopts $\calB$.
}

\begin{theorem} \label{thm:cross-submod}
In any instance of Com-IC with mutual complementarity $\bQ^+$, %, i.e., $\qao \leq \qab$ and $\qbo \le \qba$, 
	%if $\qba = 1$,  then
	$\sigma_\calA$ is {\em cross-submodular} w.r.t.\ $\calB$-seed set $S_\calB$, for any fixed
	$A$-seed set, as long as $\qba = 1$.
\end{theorem}

\begin{proof}
We first fix an $\calA$-seed set $S_\calA$.
Consider any possible world $W$.
Let $\Psi^W_\calA(S_\calB)$ be the set of $\calA$-adopted nodes in $W$ with $\calB$ seed-set $S_\calB$ (and $\calA$-seed
	set $S_\calA$).
Consider $\calB$-seed sets $S\subseteq T \subseteq V$ and another $\calB$-seed $u \in V \setminus T$.
It suffices to show that for any $v\in \Psi^W_\calA(T \cup \{u\}) \setminus \Psi^W_\calA(T)$,
	we have $v \in \Psi^W_\calA(S \cup \{u\}) \setminus \Psi^W_\calA(S)$.

Let an {\em $\calA$-path} be a live-edge path from some $\calA$-seed such that
	all nodes on the path adopt $\calA$, and $\calB$-path is defined symmetrically.
If a node $w$ has $\alpha_\calA^w \le \qao$, we say that $w$ is {\em $\calA$-ready}, meaning that $w$ is ready for $\calA$
	and will adopt $\calA$ if it is informed of $\calA$, regardless of its status on $\calB$.
We say a path from $S_\calA$ is an {\em $\calA$-ready path} if all nodes on the path (except the starting $\calA$-seed)
	are $\calA$-ready.
It is clear that all nodes on an $\calA$-ready path would always adopt $\calA$ regardless of $\calB$-seeds.
We define $\calB$-ready nodes and paths symmetrically.
We can show the following claim.

\begin{claim} \label{clm:bjoina}
{\lemmaClaimBA}
\end{claim}

Now consider the case of $S_\calB = T \cup \{u\}$ first.
Since $v\in \Psi^W_\calA(T \cup \{u\})$, there must be an $\calA$-path $\PP_\calA$
	from some node $w_0 \in S_\calA$ to $v$.
If path $\PP_\calA$ is $\calA$-ready, then regardless of $\calB$ seeds,
	all nodes on $\PP_\calA$ would always be $\calA$-adopted,
	but this contradicts the assumption that $v \not \in \Psi^W_\calA(T)$.
Therefore, there exists some node $w$ that is not $\calA$-ready, i.e., $\qao < \alpha^w_\calA \le \qab$.
Let $w$ be the first non-$\calA$-ready node on path $\PP_\calA$.
Then $w$ must have adopted $\calB$ to help it adopt $\calA$, and $\alpha^w_\calB \le \qbo$.
We can show the following key claim.

\begin{claim} \label{clm:bpathcrosssub}
{\claimBPath}
\end{claim}

With the key Claim~\ref{clm:bpathcrosssub}, the rest of the proof follows the standard argument as in the
	other proofs.
In particular, since even when $x_0$ is the only $\calB$-seed, $w$ can still be $\calB$-adopted, then
	by Claim~\ref{clm:bjoina}, $v$ would be $\calA$-adopted in this case.
Thus we know that $x_0$ must be $u$, because otherwise it contradicts our assumption that
	$v \not\in \Psi^W_\calA(T)$ (also relying on the cross-monotonicity proof made for Theorem~\ref{thm:monotone}).
Then again by the cross-monotonicity, we know that
	$v \in \Psi^W_\calA(S \cup \{u\})$, but $v \not \in \Psi^W_\calA(S)$.
This completes our proof.
\end{proof}

\section{Approximation Algorithms}\label{sec:rrset}
%We now present our approximation algorithms for \SIM and \CIM.
We first review the state-of-the-art in influence maximization and then derive a general framework (\textsection\ref{sec:generalTIM}) to obtain approximation algorithms for \SIM  (\textsection\ref{sec:rr-sim}) and \CIM  (\textsection\ref{sec:rr-cim}).

\spara{TIM algorithm}
For influence maximization,  %under the IC and LT models, 
Tang et al.~\cite{tang14} proposed the {\em Two-phase Influence Maximization (TIM)} algorithm that produces a $(1-1/e-\epsilon)$-approximation with at least $1-|V|^{-\ell}$ probability in $O((k+\ell)(|E|+|V|)\log|V|/\epsilon^2)$ expected running time.
%It combines theoretical approximation guarantee and efficiency, and is consider one of the state-of-the-art algorithms.
It is based on the concept of  {\em Reverse-Reachable sets (RR-sets)}~\cite{borgs14}, and applies to the Triggering model~\cite{kempe03} that generalizes both IC and LT.
TIM is {\sl orders of magnitude faster than greedy algorithm with Monte Carlo
	simulations~\cite{kempe03}, while still giving approximation solutions with high probability}.
Recently they propose a new improvement~\cite{tang15}, which significantly 
	reduces the number of RR-sets generated using martingale analysis.
To tackle \SIM and \CIM,
%Our extension of the RR-set approach focuses 
we primarily focus on the challenging task of correctly generating RR-sets
	in \model and other more general models, which is orthogonal 
         to the contributions of~\cite{tang15}. 
Hereafter we focus on the framework of \cite{tang14}.  % in this paper.
Due to the much more complex dynamics involved in \model, 
	adapting TIM to solve \SIM and \CIM is far from trivial, as we shall show.

\spara{Reverse-Reachable Sets}
%The notion of {\em Reverse-Reachable sets} (or {\em RR-sets}) was first proposed in \cite{borgs14}.
In a deterministic (directed) graph $G' = (V', E')$, for a fixed $v \in V'$, all nodes that can reach $v$ form an {\em RR-set} rooted at $v$~\cite{borgs14}, denoted $R(v)$.
A {\em random} RR set encapsulates two levels of randomness:
($i$) a ``root'' node $v$ is randomly chosen from the graph, and
($ii$) a deterministic graph is sampled according to a certain probabilistic rule that retains a subset of edges from the graph.
E.g., for the IC model, each edge $(u,v)\in E$ is removed w.p.\ $(1-p_{u,v})$, independently.
TIM first computes a lower bound on the optimal solution value and uses this bound to derive the number of random RR-sets to be sampled, denoted $\theta$.
To guarantee approximation solutions, $\theta$ must satisfy:
%, so that approximation solutions are guaranteed. Specifically,
%a parameter $\theta$.
%Then, it samples $\theta$ RR-sets and selects seeds greedily: starting from $S = \emptyset$, it adds to $S$ the node that covers the most additional RR sets, until $|S| = k$.
%For accurate estimation of influence spread, $\theta$ must satisfy:
\begin{align}\label{eqn:opt}
\theta \geq \epsilon^{-2} (8+2\epsilon)|V| \cdot \frac{\ell \log|V| + \log\binom{|V|}{k} + \log 2} {\OPT_k},
\end{align}
where $\OPT_k$ is the optimal influence spread achievable amongst all size-$k$ sets, and $\epsilon$ represents the trade-off between efficiency and quality: a smaller $\epsilon$ implies more RR-sets (longer running time), but gives a better approximation factor.
The approximation guarantee of TIM relies on a key result from \cite{borgs14}, re-stated here:

\begin{proposition}[Lemma 9 in \cite{tang14}]
\label{prop:tang}
Fix a set $S\subseteq V$ and a node $v \in V$.
Under the Triggering model, let $\rho_1$ be the probability that $S$ activates $v$ in a cascade, and $\rho_2$ be the probability that $S$ overlaps with a random RR-set $R(v)$ rooted at $v$.
Then, $\rho_1 = \rho_2$.
\end{proposition}

\subsection{A General Solution Framework}\label{sec:generalTIM}

%To incorporate the TIM algorithm for our new Com-IC model, we need to generalize the 
%	definition of RR-set and the TIM approach to a general diffusion model.
%We take a top-down approach to adapt the nice solution framework of TIM and use the adaptation to tackle \SIM and \CIM problems under the \model model.
We use Possible World (PW) models to generalize the theory in \cite{borgs14, tang14}.
For a generic stochastic diffusion model $M$, an \emph{equivalent} PW
	model $M'$ is a model that specifies a distribution over $\mathcal{W}$,  the set of all
	possible worlds, where influence diffusion in each possible world in $\cal W$ is
	deterministic. 
%, and in all $W\in \mathcal{W}$.
%and possible world $W$ is randomly drawn from a distribution.
%By equivalence we mean that 
Further, given a seed set (or two seed sets $S_\calA$ and $S_\calB$ as in
	\model), the distribution of the sets of active nodes
	(or $\calA$- and $\calB$-adopted nodes in \comic) in $M$ is the
	same as the corresponding distribution in $M'$.
Then, we define a generalized concept of {\em RR-set} through the PW model:

\begin{definition}[General  RR-Set] \label{def:generalRRset}
For each possible world $W \in \mathcal{W}$ and a given node $v$ (a.k.a.\ root), the {\em reverse reachable set
	(RR-set) of $v$ in $W$}, denoted by $R_W(v)$, 
	consists of all nodes $u$ such that the singleton
	set $\{u\}$ would activate $v$ in $W$.
A {\em random} {RR-set of $v$} is a set $R_W(v)$ where $W$ is randomly sampled from $\mathcal{W}$ using the probability distribution given in $M'$.
\end{definition}

It is easy to see that Definition~\ref{def:generalRRset}  encompasses
	the RR-set definition in~\cite{borgs14,tang14} for IC, LT, and Triggering models as special cases.
For the entire solution framework to work, 
	the key property that RR-sets need to satisfy is the following:

\begin{definition}[Activation Equivalence Property] \label{def:RRsetProperty}
%Given a graph $G=(V,E,p)$
%	with an influence diffusion model $M$ and an equivalent possible world model $M'$.
Let $M$ be a stochastic diffusion model and $M'$ be its equivalent possible world model.
Let $G=(V,E,p)$ be a graph.
Then, RR-sets have the {\em Activation Equivalence Property} if
	for any fixed $S\subseteq V$ and any fixed $v \in V$, 
	the probability that $S$ activates $v$ according to %the diffusion model 
	$M$
	is the same as the probability that $S$ overlaps with a random RR-set generated
	from $v$ in a possible world in $M'$.
\end{definition}
\vspace*{-2mm}

As shown in~\cite{tang14}, the entire correctness and complexity analysis is based
on the above property, and in fact in their latest improvement~\cite{tang15}, they
	directly use this property as the definiton of general RR-sets.
Proposition~\ref{prop:tang} shows that the activation equivalence property holds for
	the triggering model.
We now provide {\sl a more general sufficient condition for the activation equivalence
	property to hold} (Lemma~\ref{lemma:generic-rr}), 
	which gives {\sl concrete conditions on when the RR-set based framework would work}.
More specifically, we show that for any diffusion model $M$, if there is an  equivalent PW model $M'$ of which all possible worlds satisfy the following two properties, then RR-sets
	have the activation equivalence property.
%By saying $M$ and $M'$ are equivalent, we mean $M'$ specifies a probability distribution over the space of all possible worlds, such that (1) in every possible world $W$, influence propagates deterministically and (2) for any given seed set, the distributions over sets of active nodes under $M$ and $M'$ are the same. 

%Let $M$ be an influence propagation model consisting of a graph $G$ with edge weights and a set of randomized diffusion rules.
%Let $\mathcal{W}$ be the set of all possible worlds, i.e., deterministic graphs on which influence propagates in a deterministic manner.
%We assume that the distribution over sets of active nodes under the original model $M$ is the same as that under an equivalent model that samples possible worlds from $\mathcal{W}$.
%Furthermore, we require the diffusion in any possible world $W \in \mathcal{W}$ to satisfy two properties:

%\vspace*{-1mm}
%\begin{property}\label{def:p1}
%Given two seed sets $S$ and $T$ such that $S \subseteq T$, if a node $v$ can be activated by $S$ in a possible world $W$, then $v$ shall also be activated by $T$ in $W$.
%\end{property}
%\vspace*{-4mm}
%\begin{property}\label{def:p2}
%If a node $v$ can be activated by $S$ in a possible world $W$, then there exists $u\in S$ such that the singleton seed set $\{u\}$ can also activate $v$ in $W$.
%\end{property}
\subsubsection*{Possible World Properties}
\begin{itemize}
\item {\it (P1):} Given two seed sets $S \subseteq T$, if a node $v$ can be activated by $S$ in a possible world $W$, then $v$ shall also be activated by $T$ in $W$.
\item {\it (P2):} If a node $v$ can be activated by $S$ in a possible world $W$, then there exists $u\in S$ such that the singleton seed set $\{u\}$ can also activate $v$ in $W$.
In fact, {\it (P1)} and {\it (P2)} are equivalent to monotonicity and submodularity, as we formally state below.
\end{itemize}

\def\lemmaGenericRRSubmod{
Let $W$ be a fixed possible world.
Let $f_{v,W}(S)$ be an indicator function that takes on $1$ if $S$ can activate $v$ in $W$, and $0$ otherwise.
Then, $f_{v,W}(\cdot)$ is {\em monotone} and {\em submodular} for all $v \in V$ {if and only if} both (P1) and (P2) are satisfied in $W$.
}

\begin{lemma}\label{lemma:generic-rr-2}
{\lemmaGenericRRSubmod}
\end{lemma}

\def\lemmaGenericRR{
Let $M$ be a stochastic diffusion model and $M'$ be its equivalent possible world model.
If $M'$ satisfies Properties {\it (P1)} and {\it (P2)}, then the RR-sets
	as defined in Definition~\ref{def:generalRRset} have the activation equivalence property
	as in Definition~\ref{def:RRsetProperty}.
}

\begin{lemma}\label{lemma:generic-rr}
{\lemmaGenericRR}
\end{lemma}

%\begin{proof}
%It is sufficient to prove that in every possible world $W \in \mathcal{W}$, $S$ activates $v$ if and only if $S$ intersects with $v$'s RR set in $W$, denoted by $R_W(v)$.
%
%
%Suppose $R_W(v) \cap S \neq \emptyset$.
%Without loss of generality, we assume a node $u$ is in the intersection.
%By the definition of RR set, set $\{u\}$ can activate $v$ in $W$.
%Per Property~\ref{def:p1}, $S$ can also activate $v$ in $W$.
%
%Now suppose $S$ activates $v$ in $W$.
%Per Property~\ref{def:p2}, there exists $u\in S$ such that $\{u\}$ can also activate $v$ in $W$.
%Then by the RR-set definition, $u \in R_W(v)$.
%Therefore, $S \cap R_W(v) \neq \emptyset$.
%\end{proof}

%Finally, as long as an influence diffusion model satisfies Lemma~\ref{lemma:generic-rr}, we can apply the TIM algorithm to achieve approximation guarantees for influence maximization.

Comparing with directly using the activation equivalence property as the RR-set 
	definition in~\cite{tang15},
	our RR-set definition provides a more concrete way of constructing RR-sets, and
	our Lemmas~\ref{lemma:generic-rr-2} and~\ref{lemma:generic-rr} provide general
	conditions under which such constructions can ensure algorithm correctness.
Algorithm~\ref{alg:generalTIM}, \generalTIM, outlines a general solution framework based on
	RR-sets and TIM.
It provides a probabilistic approximation guarantee for any diffusion
	models that satisfy {\it (P1)} and {\it (P2)}.
Note that the estimation of a lower bound $LB$ of $\OPT_k$ (line~\ref{line:g1}) is orthogonal to our
	contributions and we refer the reader to \cite{tang14} for details.
Finally, we have:

\begin{theorem}\label{thm:generalTIM}
Suppose for a stochastic diffusion model $M$ with
	an equivalent PW model $M'$, that for every possible world $W$ and every $v\in V$, 
	the indicator function $f_{v,W}$ is monotone \& submodular. Then
	for influence maximization under $M$ with graph $G=(V,E,p)$ and
	seed set size $k$, \generalTIM (Algorithm~\ref{alg:generalTIM}) applied on the general
	RR-sets (Definition~\ref{def:generalRRset}) returns a $(1-1/e-\epsilon)$-approximate
	solution with at least $1-|V|^{-\ell}$ probability.
%, \textcolor{red}{and it runs in $O((k+\ell)(|V|+|E|)\log|V|/\epsilon^2)$ expected time.}
\end{theorem}
\vspace*{-1mm}

Theorem~\ref{thm:generalTIM} follows from Lemmas~\ref{lemma:generic-rr-2} and \ref{lemma:generic-rr}, and the fact that all theoretical analysis of TIM relies only on the Chernoff bound and the activation equivalence property, ``without relying on any other results specific to the IC model''~\cite{tang14}.
Next, we describe how to generate RR-sets correctly for \SIM and \CIM under \comic (line~\ref{line:g3} of Algorithm~\ref{alg:generalTIM}), which is much more complicated than IC/LT models~\cite{tang14}.
We will first focus on submodular settings for \SIM % ($\qao \leq \qab$ and $\qbo = \qba$, 
(Theorem~\ref{thm:submod-complement}) and \CIM %($\qao \leq \qab$ and $\qbo \leq \qba = 1$, 
(Theorem~\ref{thm:cross-submod}). 
In \textsection\ref{sec:sandwich}, we propose {\em Sandwich Approximation} to handle general $\bQ^+$ where submodularity does \emph{not} hold.

\subsection{Generating RR-Sets for {\sc\large SelfInfMax}}\label{sec:rr-sim}

%\SetAlgoSkip{-10em}

\begin{algorithm}[t!]
%\begin{small}
\DontPrintSemicolon
\caption{$\mathsf{GeneralTIM}$ ($G=(V,E,p)$, $k$, $\epsilon$, $\ell$)} \label{alg:generalTIM}
$LB \gets $ lower bound of $\OPT_k$ estimated by method in~\cite{tang14}\; \label{line:g1}
compute $\theta$ using Eq.~\eqref{eqn:opt} with $LB$ replacing $\OPT_k$\;
$\mathcal{R} \gets$ generate $\theta$ random RR-sets according to Definition~\ref{def:generalRRset} \tcp{for \SIM, use \timsim or \timsimfast; for \CIM, use \timcim}  \label{line:g3}
\For {\em $i = 1$ to $k$} {
    $v_i \gets$ the node appearing in the most RR-sets in $\mathcal{R}$\;
   $S \gets S \cup \{v_i\}$ \tcp{$S$ was initialized as $\emptyset$}
    remove all RR-sets in which $v_i$ appears\;
}
\textbf{return} $S$ as the seed set\;
%\end{small}
\end{algorithm}

%We next describe our RR-set generation algorithm for \SIM under \model, which we call \timsim.
%In this section, we focus on ``one-way'' $\bQ^+$:
%	$\qao \leq \qab$ and $\qbo = \qba$, for which 
%	self-submodularity holds (Theorem~\ref{thm:submod-complement}).
%For this case 
We present two algorithms, \timsim and \timsimfast, for generating
	random RR-sets per Definition~\ref{def:generalRRset}. %for \SIM.
%We will handle mutually complementary products
%	later in Section~\ref{sec:approx}.
%When plugging \timsim to the general TIM framework, 
The overall algorithm for \SIM can be obtained by plugging \timsim or \timsimfast
	into \generalTIM (Algorithm~\ref{alg:generalTIM}).

According to Definition~\ref{def:generalRRset},
	for \SIM, the RR-set of a root $v$ in
	a possible world $W$, $R_W(v)$, is
	the set of nodes $u$ such that if $u$ is the only $\calA$-seed,
	$v$ would be $\calA$-adopted in $W$, given any fixed $\calB$-seed set $S_\calB$.
By Theorems~\ref{thm:monotone} 
	and~\ref{thm:submod-complement} (%Although these
		%theorems are stated for the influence spread functions,
		whose proofs indeed show that the indicator function
		$f_{v,W}(S)$ is monotone and submodular),
		along with Lemmas~\ref{lemma:generic-rr-2} and~\ref{lemma:generic-rr}, we know that 
		RR-sets following Definition~\ref{def:generalRRset} have the
		activation equivalence property.
We now focus on how to construct RR-sets following Definition~\ref{def:generalRRset}.
Recall that in \comic, adoption decisions for $\calA$ are based on a number
	of factors such as whether $v$ is reachable via a live-edge path from $S_\calA$
	and its state w.r.t.\ $\calB$ when reached by $\calA$.
Note that $\qbo = \qba$ implies that $\calB$-diffusion is
	independent of $\calA$ (Lemma~\ref{lemma:indiff}).
Our algorithms take advantage of this fact, by 
first revealing node states w.r.t. $\calB$, which gives a sound
	basis for generating RR-sets for $\calA$.

\subsubsection{The RR-SIM Algorithm}

Conceptually,
\timsim (Algorithm~\ref{alg:rrset-sim}) proceeds in three phases.
Phase \rom{1} samples a possible world according to 
	\textsection\ref{sec:pw} (omitted from the pseudo-code).
Phase \rom{2} is a {\em forward labeling} process from the
	input $\calB$-seed set $S_\calB$ (lines \ref{line:x0} to \ref{line:x2}): a node
	$v$ becomes $\calB$-adopted if $\alpha_\calB^{v,W} \leq \qbo$
	and $v$ is reachable from $S_\calB$ via a path consisting entirely
	of live edges  and $\calB$-adopted nodes.
In Phase \rom{3} (lines \ref{line:x3} to \ref{line:x5}), we randomly select a node
	$v$ and generate RR-set $R_W(v)$ by running a Breadth-First
	Search (BFS) backwards (following incoming edges).
Note that the RR-set generation for IC and LT models~\cite{tang14}
	is essentially a simpler version of Phase \rom{3}.

\spara{Backward BFS}
Given $W$, % (the PW generated i phases),
	an RR-set $R_W(v)$ includes all nodes explored in the following backward BFS procedure.
Initially, we enqueue $v$ into a FIFO queue $Q$.
We repeatedly dequeue a node $u$ from $Q$ for processing until the queue is empty.

\begin{itemize}
\item {\em Case 1: $u$ is $\calB$-adopted.} There are two sub-cases:
	$(i)$.\ If $\alpha_\calA^u \leq \qab$, then $u$ is able to transit from $\calA$-informed to $\calA$-adopted.
	Thus, we continue to examine $u$'s in-neighbors.
	For all unexplored $w\in N^{-}(u)$, if edge $(w,u)$ is live, then enqueue $w$;
	$(ii)$.\ If $\alpha_\calA^u > \qab$, then $u$ cannot transit from $\calA$-informed to $\calA$-adopted, and thus $u$ has to be an $\calA$ seed to become $\calA$-adopted. In this case, $u$'s in-neighbors will not be examined.

\item {\em Case 2: $u$ is not $\calB$-adopted.} Similarly, if $\alpha_\calA^u \leq \qao$, perform actions as in 1$(i)$; otherwise perform actions as in 1$(ii)$. 
\end{itemize}

\def\theoremRRSIMCorrect{
Under one-way complementarity ($\qao \leq \qab$ and $\qbo = \qba$),
	the RR-sets generated by
	the \timsim algorithm satisfy Definition~\ref{def:generalRRset} for
	the \SIM problem.
As a result, Theorem~\ref{thm:generalTIM} applies to \generalTIM with \timsim
	in this case.
}
	
%\vspace*{-2mm}
\begin{theorem}\label{thm:rrsim-correct}
{\theoremRRSIMCorrect}
\end{theorem}
%\vspace*{-2mm}

%\SetAlgoSkip{}
\begin{algorithm}[t!]
%\begin{small}
\DontPrintSemicolon
\caption{\timsim($G=(V,E)$, $v$, $S_\calB$)} \label{alg:rrset-sim}
create an empty FIFO queue $Q$ and empty set $R$\;
enqueue all nodes in $S_\calB$ into $Q$ \tcp{start forward labeling}  \label{line:x0}
\While{\em $Q$ is not empty} {   \label{line:x1}
	$u \gets Q.\mathsf{dequeue}()$ and mark $u$ as $B$-adopted \;
	\ForEach{\em $v \in N^+(u)$ such that $(u,v)$ is live } {
		\If{\em $\alpha_\calB^{v,W} \leq \qbo \; \wedge$ $v$ is not visited} 
		{$Q.\mathsf{enqueue}(v)$ \tcp{also mark $v$ as visited}  \label{line:x2}} 
    }
}
clear $Q$, and then enqueue $v$ \tcp{\rm\it start backward BFS}  \label{line:x3}
\While{\em $Q$ is not empty} {
	$u \gets Q.\mathsf{dequeue}()$\;
    $R \gets R \cup \{u\}$\;
	\If{\em ($u$ is $\calB$-adopted $\wedge \; \alpha_\calA^{u,W} \leq \qab$)  $\vee$ ($u$ is not $\calB$-adopted $\wedge \; \alpha_\calA^{u,W} \leq \qao$) } {
        \ForEach{\em $w \in N^-(u)$ such that $(w,u)$ is live} {
         	\If{\em $w$ is not visited}{$Q.\mathsf{enqueue}(w)$ \tcp{also mark $w$ visited} \label{line:x5}}
		}
	} 
}
\textbf{return} $R$ as the RR-set\;  \label{line:x4}
%\end{small}
\end{algorithm}

\noindent\textbf{Lazy sampling.}
For \timsim to work, it is {\em not} necessary to 
	sample all edge- and node-level variables (i.e., the
	entire possible world) up front, as the forward labeling
	and backward BFS are unlikely to reach the whole graph.
Hence, we can simply reveal edge and node states on demand
	(``lazy sampling''), based on the principle of deferred
	decisions.
%That is, we only need to determine the live/blocked status of all relevant edges, and similarly we only generate the $\alpha$-values of all relevant nodes.
%By ``relevant'', we mean the edges and nodes that are actually examined in the forward and backward cascades.
In light of this observation, the following improvements are made to \timsim.
First, the first phase is simply skipped.
Second, in Phase \rom{2}, edge states and $\alpha$-values are sampled as the forward labeling from $S_\calB$ goes on.
We record the outcomes, as it is possible to encounter certain edges and nodes again in phase (\rom{3}).
Next, for Phase \rom{3}, consider any node $u$ dequeued from $Q$.
We need to perform an additional check on every incoming edge $(w,u)$.
If $(w,u)$ has already been tested live in Phase \rom{2}, then we just enqueue $w$.
Otherwise, we first sample its live/blocked status, and enqueue $w$ if it is live,
Algorithm~\ref{alg:rrset-sim} provides the pseudo-code for 
	\timsim, where sampling is assumed to be done whenever we need to 
	check the status of an edge or the $\alpha$-values of a node.

%The above improvements rely on the principle of deferred decisions: instead of sampling the status of all edges and all nodes upfront, we only sample what we truly need (``lazy sampling'').

\spara{Expected time complexity}
For the entire seed selection (Algorithm~\ref{alg:generalTIM}
	with \timsim) to guarantee approximate solutions, we must
	estimate a lower bound $LB$ of $\OPT_k$ and use it to derive the
	minimum number of RR-sets required, defined as $\theta$ in
	Eq. \eqref{eqn:opt}.
In expectation, the algorithm runs in $O(\theta \cdot \EPT)$ time,
	where $\EPT$ is the expected number of edges explored in generating
	one RR-set.
Clearly, $\EPT = \EPT_F + \EPT_B$, where $\EPT_F$ ($\EPT_B$)
	is the expected number of edges examined in forward labeling
	(resp., backward BFS).
Thus, we have the following result.

\def\lemmaRRSIMRuntime{
 In expectation, \generalTIM with \timsim runs in
$O\left((k+\ell)(|V|+|E|)\log |V|\left(1+ {\EPT_F}/{\EPT_B}\right) \right)$ time.
}

\begin{lemma}\label{lemma:rr-sim-time}
{\lemmaRRSIMRuntime}
\end{lemma}

%The time complexity has a term $\EPT_F/\EPT_B$, where 
%The 
$\EPT_F$ 
%term in time complexity
	increases when the input $\calB$-seed set grows.
	%\blue{the influence capability of the the input $\calB$-seed set increases}
Intuitively, it is reasonable that a larger $\calB$-seed set may have more
	complementary effect and thus it may take longer time to find the best
	$\calA$-seed set.
However, it is possible to reduce $\EPT_F$ as described below.
%improve the time complexity, as shown next.

%The expected running time complexity of \timsim is $O(\theta \cdot \EPT)$ where $\theta$ is the total number of random RR-sets to generate and $\EPT$ denotes the expected running time of generating a single RR-set.
%Recall that to achieve the probabilistic approximation guarantees in Theorem~\ref{thm:generalTIM}, $\theta$ must satisfy \eqref{eqn:opt}, which includes $\OPT$ that is NP-hard to compute (Theorem~\ref{thm:sim-hard}).
%Thus, $\OPT$ ought to be replaced with a lower bound of itself such that \eqref{eqn:opt} holds.
%
%Now we establish a lower bound of $\OPT$.
%Let $\EPT_F$ and $\EPT_B$ be the expected number of edges processed in phases (\rom{2}) and (\rom{3}) respectively.
%Also, let $\EPT_B_{ic}$ be the expected number of edges that would have been processed in phase (3) if $\qao = \qab = 1$ (\model in this setting is equivalent to the classic IC model in terms of the propagation of $\calA$).
%Consider any RR-set $R$, and let $p_R$ be the probability that a randomly chosen edge from the graph points to a node in $R$.
%Then, we have
%$$
%\EPT_B \leq \EPT_B_{ic} = (|E| / |V|) \cdot \sigma_{ic}(\{v^*\}),
%$$
%where the second equality comes from Lemma 4 of \cite{tang14}.

\subsubsection{The RR-SIM+ Algorithm}

\begin{algorithm}[t]
%\begin{small}
\DontPrintSemicolon
\caption{\timsimfast($G=(V,E)$, $v$, $S_\calB$)} \label{alg:rrset-sim-fast}
create an FIFO queue $Q$ and empty sets $R$, $T_1$\;
 $Q.\mathsf{enqueue}(v)$; \tcp{\rm\it first backward BFS}       \label{line:y1}
\While{\em $Q$ is not empty} { 
       $u \gets Q.\mathsf{dequeue}()$\;
	$T_1 \gets T_1 \cup \{u\}$\;
       \ForEach{\em unvisited $w \in N^-(u)$ such that $(w,u)$ is live } {
               $Q.\mathsf{enqueue}(w)$ and mark $w$ visited\;     \label{line:y2}
        }
}
\If {$T_1 \cap S_\calB \neq \emptyset$  } {  \label{line:y3}
	\tcp{auxiliary forward pass to determine $\calB$ adoption}
      clear $Q$, enqueue all nodes of $T_1\cap S_\calB$ into $Q$, and execute line \ref{line:x1} to line \ref{line:x2} in Algorithm~\ref{alg:rrset-sim}\;   \label{line:y4}
}
execute line \ref{line:x3} to line \ref{line:x4} in  Algorithm~\ref{alg:rrset-sim} 
\tcp{second backward BFS} 
 \label{line:y5}
 %\end{small}
\end{algorithm}

The \timsim algorithm may incur efficiency loss because some of the
	work done in forward labeling (Phase \rom{2}) may not be used in backward BFS (Phase \rom{3}).
E.g., consider an extreme situation where all nodes explored in forward labeling are in a different connected component of the graph than the root $v$ of the RR-set.
In this case, forward labeling can be skipped safely and entirely!
To leverage this, we propose \timsimfast (pseudo-code presented as Algorithm~\ref{alg:rrset-sim-fast}), of which the key idea is to run {\em two} rounds of backward BFS from the random root $v$.
The first round % (lines \ref{line:y1} to \ref{line:y2}) 
determines the {\em necessary scope} of forward labeling,
%(lines \ref{line:y3} to \ref{line:y4}), 
while the second one generates the RR-set.
%  (line \ref{line:y5}).
%The intuition is that the size of $S_\calB$ could be large, and hence the forward cascade from $S_\calB$ may explore a much larger region of $G$ compared to running BFS from just a single node $v$.
%Intuitively, forward labeling from $S_\calB$ may be more expensive (i.e., visit more nodes and edges) compared to a backward BFS from a single node $v$.

\spara{First backward BFS}
As usual, we create a FIFO queue $Q$ and enqueue the random root $v$.
We also sample $\alpha_\calB^v$ uniformly at random from $[0,1]$.
%While doing so, $v$ is marked {\em positive} w.p.\ $\qbo$ and {\em negative} w.p.\ $(1-\qbo)$.
%This positive/negative status is used later for checking if a node explored in this round can become $\calB$-adopted. 
Then we repeatedly dequeue a node $u$ until $Q$ is empty: for each incoming edge $(w,u)$, we test its live/blocked status based on probability $p_{w,u}$, independently.
If $(w,u)$ is live and $w$ has not been visited before, enqueue $w$ and sample its $\alpha_\calB^w$.

Let $T_1$ be the set of all nodes explored.
If $T_1 \cap S_\calB = \emptyset$, then none of the $\calB$-seeds can reach the explored nodes,
	so that forward labeling can be completely skipped.
The above extreme example falls into this case.
Otherwise, we run a {\em residual} forward labeling only from $T_1 \cap S_\calB$
	along the explored nodes in $T_1$: if a node $u\in T_1 \setminus S_\calB$ is reachable by some $s \in T_1 \cap S_\calB$
	via a live-edge path with all $\calB$-adopted nodes, and $\alpha_\calB^{u,W} \leq \qbo$, $u$ becomes $\calB$-adopted.
Note that it is not guaranteed in theory that this always saves time compared to \timsim, since
	the worst case of \timsimfast is that $T_1 \cap S_\calB = S_\calB$, which means that 
	the first round is wasted.
However, our experimental results \textsection\ref{sec:exp} indeed show that \timsimfast
	is at least twice faster than \timsim on three of the four datasets.

\spara{Second backward BFS}
This round is largely the same as Phase \rom{3} in \timsim, but there is a subtle difference.
Suppose we just dequeued a node $u$.
It is possible that there exists an incoming edge $(w, u)$ whose status is not determined.
This is because we do not enqueue previously visited nodes in BFS.
Hence, if in the previous round, $w$ is already visited via an out-neighbor other than $u$, $(w, u)$ would not be tested.
Thus, in the current round we shall test $(w,u)$, and decide if $w$ belongs to $R_W(v)$ accordingly.
To see \timsimfast is equivalent to \timsim, it suffices to show that for each node explored in the second backward BFS, its  adoption status w.r.t.\ $\calB$ is the same in both algorithms.

\def\lemmaRRSIMPlus{
Consider any possible world $W$ under the \comic model.
Let $v$ be a root for generating an RR-set.
For any $u\in V$ that is backward reachable from $v$ via live-edges in $W$, $u$ is determined as $\calB$-adopted in \timsim if and only if $u$ is determined as $\calB$-adopted in \timsimfast.
}

\begin{lemma}\label{lemma:rr-sim-plus}
{\lemmaRRSIMPlus}
\end{lemma}

\begin{proof}
We first prove the ``if'' part.
Suppose $u$ is determined as $\calB$-adopted in \timsimfast.
This means that there exists a node $s \in T_1 \cap S_\calB$, such that there is a path from $s$ to $u$ consisting entirely of live-edges and $\calB$-adopted nodes (every node $w$ on this path satisfies that $\alpha_\calB^w \leq \qbo$).
Therefore, in \timsim, where $W$ is generated upfront, this live-edge path must still exist.
Thus, $u$ must be also $\calB$-adopted in \timsim as well.

Next we prove the ``only if'' part.
By definition, if $u$ is determined as $\calB$-adopted in possible world $W$, then there exists a path $P$ from some $s \in S_\calB$ to $u$ such that the path consists entirely of live edges and all nodes $w$ on the path satisfy that $\alpha_W^w \leq \qbo$.
It suffices to show that if $u$ is reachable by $v$ backwards in $W$, then $P$ will be explored entirely by \timsimfast.

Suppose otherwise.
That is, there exists a node $z \in P$, such that $z$ cannot be explored by the first round backward BFS from $v$.
We have established that in the complete possible world $W$, there is a live-edge path from $z$ to $u$ and from $u$ to $v$ respectively.
Thus, connecting the two paths at node $u$ gives a single live-edge path $P_z$ from $z$ to $v$.
Now recall that the continuation of the first backward BFS phase in \timsimfast relies solely on edge status (as long as an edge $(w,u)$ is determined live, $w$ will be visited by the BFS).
This means that $z$ must have been explored in the first backward BFS allow the backward
	path from $v$ to $u$ and then along the path $P$, which is a contradiction.
\end{proof}

The analysis on expected time complexity is similar:
We can show that the expected running time of \timsimfast is 
	$O\left((k+\ell)(|V|+|E|)\log |V|\left(1+ {\EPT_{B1}}/{\EPT_{B2}}\right) \right)$,
	where $\EPT_{B1}$ ($\EPT_{B2}$) is the expected number of edges explored in the 
	first (resp., second) backward BFS.
Compared to \timsim, % we can see that 
$\EPT_{B2}$ is
	the same as $\EPT_{B}$ in \timsim, so \timsimfast will be faster than
	\timsim if $\EPT_{B1}< \EPT_F$, i.e., if the first backward BFS plus the residual forward labeling \textsl{explores fewer edges}, compared to the full orward labeling in \timsim. 

\eat{the number of edges explored in
	a backward search (with residual forward labeling) from a random node disregarding GAPs
	is smaller than that in a full forward labeling  %(in terms of explored edges) 
	from the fixed $\calB$-seed set in \timsim. }

%First, $\EPT = \EPT_1 + \EPT_2$.
%Clearly, by Lemma 4 in \cite{tang14},
%\[
%\EPT_1 = (m/n) \cdot \mathbb{E}[\sigma_{IC}(\{v^*\})],
%\]
%where $v^*$ is a node chosen randomly from a distribution over $V$ such that the probability mass of each node is in proportion to its in-degree, and the expectation is taken over $v^*$.
%Also, $\sigma_{IC}$ denotes the influence spread function of the classical IC model.
%
%And for $\EPT_2$, we show that
%\[
%\frac{n}{m}\EPT_2 \leq \mathbb{E}[\sigma_A(\{v^*\}, S_\calB)] \leq \OPT.
%\]
%
%Setting $\theta = \frac{m \lambda}{n \EPT_2}$, the total expected running time is
%\begin{align}
%&\quad O(\theta \cdot (\EPT_1 + \EPT_2)) \nn \\
%&= O \left( \frac{m \lambda}{n \EPT_2} \cdot \left( \frac{m \mathbb{E}[\sigma_{IC}(\{v^*\})]}{n} + EPT_2 \right) \right)  \nn \\
%&= O \left( \lambda \cdot \left( \frac{m^2 \mathbb{E}[\sigma_{IC}(\{v^*\})]}{n^2 EPT_2} + \frac{m}{n} \right) \right)
%\end{align}

%\vspace*{-1ex} 
\subsection{Generating RR-Sets for {\sc\large CompInfMax}}\label{sec:rr-cim}

\begin{algorithm}[t!]
%\begin{small}
\DontPrintSemicolon
\caption{\timcim($G=(V,E)$, $v$, $S_\calA$)}\label{alg:rrset-cim}

conduct forward labeling on $G$ from $S_\calA$, {\em cf}.\ Eq.~\eqref{eqn:labels}\;

\If{\em $v$ is neither $\calA$-suspended or $\calA$-potential} { \label{line:z98}
	\textbf{return} $\emptyset$ as the RR-set\; \label{line:z99} % \tcp{because $\calB$-seeds won't make any difference for $v$}
}

%$R \gets \emptyset$\; 
$Q.\mathsf{enqueue}(v)$ \tcp{\rm\it $Q$ initialized as an empty FIFO queue}

\While{\em $Q$ is not empty} 
{
	$u \gets Q.\mathsf{dequeue}()$\;
     
	\If{\em $u$ is $\calA$-suspended} 
	{
       		\If{\em $u$ is {$\calA\calB$-}diffusible} {   \label{line:z11}
			$R \gets R \cup \{u\}$ \tcp{$R$ was initialized as $\emptyset$}
			conduct a secondary backward BFS from $u$ {via $\calB$ diffusible nodes}, and add all explored nodes to $R$\;  \label{line:z12}
			%\tcp{\rm \red{We prove that the secondary search is both sound and complete, and thus, the primary search need not to add $u$'s in-neighbors to $Q$.}}
		} \lElse {   \label{line:z21}
			$R \gets R \cup \{u\}$ %\tcp{\rm\it no secondary search here} \label{line:z22} 
%			if the path from $u$ to $v$ contains no other non-diffusible $\calA$-suspended nodes\;   
			%\tcp{\rm \red{I think this IF-condition should be satisfied whenever we meet a non-diffusible $\calA$-suspended nodes in the primary search. Because once we add such a node to the RR-set, we do not add its in-neighbors to the FIFO queue. So, there won't be a single path that has $> 1$ non-diffusible $\calA$-suspended nodes in the primary search}}
		}
     } % ENDIF     
        \ElseIf{\em $u$ is $\calA$-potential} {
                  \If{\em $u$ is {$\calA\calB$-}diffusible} {    \label{line:z31}
                         \ForEach{\em unvisited $w\in N^-(u)$ s.t.  $(w,u)$ live} {
                               $Q.\mathsf{enqueue}(w)$; \tcp{\rm\it also mark it visited}    \label{line:z32}
                         }
                   } \Else{    \label{line:z41}
                           $S_{f} \gets$ nodes visited in a secondary forward  BFS\;
			$S_{b} \gets$ nodes visited in a secondary backward BFS\;
			$R \gets R \cup \{u\}$ if {$S_f \cap S_b$ contains 
				an $\calA$-suspended node $u_0$}\;   \label{line:z42} 
		}
        } 
            
} 
\textbf{return} $R$ as the RR-set\;
%\end{small}
\end{algorithm}

%We now study the problem of \CIM
%	for the case of $\qao \leq \qab$ and $\qbo \leq \qba = 1$,
%	for which we saw that $\sigma_{\calA}$ is cross-submodular
%	w.r.t.\ $B$-seed sets  (Theorem~\ref{thm:cross-submod}).

In \CIM, by Definition~\ref{def:generalRRset}, 
	a node $u$ belongs to an RR-set $R_W(v)$ iff
	$v$ is not $\calA$-adopted without any $\calB$-seed, but turns
	$\calA$-adopted when $u$ is the only $\calB$-seed.
It turns out that constructing RR-sets for \CIM following the above definition
	is significantly more difficult than that for \SIM. % in the previous subsection.
This is because when $\qao \leq \qab$ and $\qbo \leq \qba = 1$, $\calA$ and
	$\calB$ complement each other, and thus a simple forward labeling from the fixed
	$\calA$-seed set, without knowing anything about $\calB$, will not be able to
	determine the $\calA$ adoption status
	of all nodes. 
This is in contrast to \SIM with one-way complementarity for which
	$\calB$-diffusion is fully independent of $\calA$.
%which was the case for \SIM in the one-way complementarity case of
%	last section.
Thus, when generating RR-sets for \CIM,
	we have to determine more complicated status in a forward labeling process
	from $\calA$-seeds, as shown below.

%There is also an additional challenge brought by the fact that in the
%	submodular setting ($\qao \leq \qab$ and $\qbo \leq \qba = 1$)
%	of \CIM, $\calA$ and $\calB$ mutually complement each other.
%This is unlike the one-way complementary setting in \SIM where
%	$\calB$-diffusion is independent of $\calA$.
%Hence, designing \timcim is a more challenging and non-trivial process.
%The algorithm proceeds in two phases: forward labeling and backward search.

\spara{Phase \rom{1}: forward labeling}
The nature of \CIM requires 
%us to correctly identify nodes that cannot
%	adopt $\calA$ merely due to the influence of $S_\calA$, but
us to identify nodes with
	the potential to be $\calA$-adopted with the help of $\calB$.
To this end, we first conduct a forward search from $S_\calA$ to label
	the nodes their status of $\calA$.
%It is similar to a BFS, but additional checks may be necessary as we shall see.
As in \timsim, we also employ lazy sampling. %: reveal edge status and $\alpha$-values on demand and record them for the next phase.
The algorithm first enqueues all $\calA$-seeds (and labels them $\calA$-adopted) into a FIFO queue $Q$.
Then we repeatedly dequeue a node $u$ for processing until $Q$ is empty. 
Let $v$ be an out-neighbor of $u$. 
Flip a coin with bias $p_{u,v}$ to determine if edge $(u,v)$ is live.
If yes, we determine the label of $v$ to be one of the following:

%\vspace{-3mm}
\begin{equation}\label{eqn:labels}
%\begin{small}
\begin{cases}
\text{$\calA$-adopted,} & \text{if $u$ is $\calA$-adopted  $\wedge\; \alpha_\calA^v \leq \qao$} \\
\text{$\calA$-rejected,} & \text{if $\alpha_\calA^v > \qab$, regardless of $u$'s status}\\
\text{$\calA$-suspended,} &\text{if $u$ is $\calA$-adopted $\wedge\; \alpha_\calA^v \in (\qao, \qab]$}\\
\text{$\calA$-potential,} & \text{if $u$ is $\calA$-suspended/potential $\wedge\; \alpha_\calA^v \leq \qab$} 
\end{cases}
%\end{small}
\end{equation}
%\vspace{-3mm}

%Clearly, both $(ii)$ and $(iii)$ represent nodes with the potential of becoming
%	 $\calA$-adopted with the help of certain $\calB$-seeds.
%The main difference is that an {\em $\calA$-suspended} node is
%	already informed of $\calA$, and will become $\calA$-adopted by reconsideration
%	as soon as it becomes $\calB$-adopted, while an {\em $\calA$-potential} node is
%	not informed of $\calA$ yet, but is possible to adopt $\calA$ if certain
%	$\calA$-suspended nodes adopt both $\calA$ and $\calB$ and reach it via a live-edge
%	path in $W$.

\smallskip
Here, $\calA$-potential is just a label used for bookkeeping and is not a state. Then node $v$ is added to $Q$ unless it is $\calA$-rejected.
Note that both $\calA$-suspended and $\calA$-potential nodes can turn into $\calA$-adopted
	with the complementary effect of $\calB$.
The main difference is an $\calA$-suspended node is informed of $\calA$,
	while an $\calA$-potential is not and the informing action
	must be triggered by $\calB$-propagation.
Also, unlike a typical BFS, the forward labeling may need to revisit a node:
	if $u$ is $\calA$-adopted (just dequeued) and $v$ is previously
	labeled $\calA$-potential, $v$ should be ``promoted'' to $\calA$-suspended. %\footnote{
This occurs when $v$ is first reached by a live-edge path through an $\calA$-suspended/potential in-neighbor, but later $v$ is reached by a longer path through an $\calA$-adopted in-neighbor.
%}
%	

%\blue{
To facilitate the second phase, we  define additional node labels 
	\emph{$\calA\calB$-diffusible} and \emph{$\calB$-diffusible}.
Node $v$ is $\calA\calB$-diffusible if $v$ can adopt both $\calA$ and $\calB$ when
	$v$ is informed about both $\calA$ and $\calB$; while $v$ is 
	$\calB$-diffusible if $v$ can adopt $\calB$ when it is informed about $\calB$.
Accordingly, the technical conditions for them are given below:

\begin{equation*}
\label{eqn:diffusible}
\begin{cases}
\text{$\calA\calB$-diffusible,} & \text{if $\alpha_\calA^v \le \qao \vee \; 
		((\qao < \alpha_\calA^v \le \qab) \wedge (\alpha_\calB^v \le \qbo))$} \\
\text{$\calB$-diffusible,} & \text{if $\alpha_\calB^v \le \qbo \vee$ $v$ is $\calA$-adopted
				as labeled in~Eq.\eqref{eqn:labels}}
\end{cases}
\end{equation*}

Note that these diffusible labels are only based on a node's local state, and they
	are not limited to the nodes explored in the first phase --- some nodes may only
	be explored in the second phase and they also need to be checked for these diffusible
	labels.
%The purpose of the above labels will be made clear in the second phase of the algorithm.
%}

%A diffusible node can adopt $\calB$ (and then $\calA$) as long as it is $\calB$-informed, while a non-diffusible node can only do so by becoming a $\calB$-seed itself.
%
%\begin{definition}[Diffusible Nodes]
%A node $v$ is {\em diffusible} in possible world $W$ if $\alpha_\calB^{v,W} \leq \qbo$ and $\alpha_\calA^{v,W} \leq \qab$.
%\end{definition}

\spara{Phase \rom{2}: RR-set generation}
The second phase features a {\em primary backward search} from a random root $v$.
Also, a number of secondary searches (from certain nodes explored in the primary search)
	may be necessary to find all nodes qualified for the RR-set.
{Intuitively, the primary backward search is to locate $\calA$-suspended
	nodes $u$ via $\calA\calB$-diffusible and $\calA$-potential nodes, since once
	such a node $u$ adopts $\calB$, it will adopt $\calA$ and then through those
	$\calA\calB$-diffusible and $\calA$-potential nodes, the root $v$ will adopt $\calA$
	and $\calB$.
Thus such a node $u$ can be put into the RR-set of $v$.
In addition, if such node $u$ is also $\calA\calB$-diffusible, then any $\calB$ seed $w$
	that can activate $u$ to adopt $\calB$ via $\calB$-diffusible nodes can also be
	put into the RR-set of $v$, and we find such nodes $w$ using a secondary backward search
	from $u$ via $\calB$-diffusible nodes.
However, some additional complication may arise during the search process, and we cover
	all cases in Algorithm~\ref{alg:rrset-cim} and explain them below.
}

%After forward labeling, 
We first sample a root $v$ randomly from $V$.
In case $v$ is labeled $\calA$-adopted or $\calA$-rejected, % by the previous phase, 
we simply return $R$ as $\emptyset$ because no $\calB$-seed set can change $v$'s adoption
	status of $\calA$ (lines~\ref{line:z98} to \ref{line:z99}).
%The primary search is captured in Algorithm~\ref{alg:rrset-cim}.
The primary search then starts.
It  first enqueues $v$ into a FIFO queue $Q$.
Now consider a node $u$ dequeued from $Q$.
Four cases arise.

\begin{itemize}
\item {\it Case 1: $u$ is $\calA$-suspended and {$\calA\calB$-}diffusible} (lines~\ref{line:z11} to \ref{line:z12}).
We add $u$ to $R$.
Moreover, any node $w$ that can propagate $\calB$ to $u$
	by itself should also be added to $R$.
To find all such $w$'s, we launch a secondary backward BFS from $u$ 
	via $\calB$-diffusible nodes.
In particular, we conduct a reverse BFS from $u$ to explore all nodes that could
	reach $u$ via $\calB$-diffusible nodes, and put all of them into set $R$.
If this secondary search touches a node $w$ that is not $\calB$-diffusible, we 
	put $w$ in $R$ but do not further explore the in-neighbors of $w$.
	
%It explores all nodes $x$ satisfying either $(i)$ $\alpha_\calB^x \leq \qbo$
%	or $(ii)$ $x$ is $\calA$-adopted (since $\qba=1$,
%	$\calA$-adopted nodes adopt $\calB$ w.p.\ $1$ after being
%	reached by $\calB$).
%In both cases, $x$ can propagate $\calB$.
%When a node violates both conditions, the BFS will not
%	explore its in-neighbors.
%It can be verified from Lemma~\ref{thm:rrcim} that all nodes
%	visited in this BFS should be added to $R$.

\item
{\it Case 2: $u$ is $\calA$-suspended but not {$\calA\calB$-}diffusible}  (line~\ref{line:z21}).
Add $u$ to $R$, {but do not initiate a secondary search, because $u$ cannot
	adopt $\calA$ or $\calB$ even if it is informed of both $\calA$ and $\calB$, and thus
	the only way to make it adopt $\calB$ is to make it a $\calB$ seed.}
%
% as long as the live-edge path from $u$ to $v$ does not contain
%	another non-diffusible $\calA$-suspended.
%There is no need to further explore $u$'s in-neighbors for it is
%	non-diffusible.
\item
{\it Case 3: $u$ is $\calA$-potential and {$\calA\calB$-}diffusible} (lines~\ref{line:z31} to \ref{line:z32}).
We enqueue all $w \in N^-(u)$ such that $(w,u)$ is live in $W$
	and continue without adding $u$ to $R$, as $u$ cannot even
	be informed of $\calA$ by $S_\calA$ alone, and hence 
	cannot propagate $\calA$ to $v$ by itself.
%That is, if an edge $(w,u)$ was not revealed live or blocked in forward
%	labeling, there is no need to test it now
%Note that $u$ is not added to the RR-set, as it obviously does not satisfy
%	Theorem~\ref{thm:rrcim}.

\item
{\it Case 4: $u$ is $\calA$-potential but not {$\calA\calB$-}diffusible} (lines~\ref{line:z41} to \ref{line:z42}).
This is the most complicated case that needs a special treatment.
In general, we should stop the primary backward search at $u$ and try other branches,
	because $u$ is not yet $\calA$-informed and $u$ cannot help in diffusing $\calA$
	and $\calB$ even when informed of $\calA$ and $\calB$.
However, there is a special case in which we can still put $u$ in $R$ (making $u$ a
	$\calB$-seed):
	$u$ can reach an $\calA$-suspended and $\calA\calB$-diffusible node $u_0$
	via a $\calB$-diffusible path such that $u$ can activate $u_0$ in adopting $\calB$
	through this path, and then $u_0$ can reach back $u$ via an
	$\calA\calB$-diffusible path, such that $u_0$ can activate
	$u$ in adopting $\calA$.
E.g., consider Figure~\ref{fig:rrcim-tricky} (all edges are live): 
$a$ is an $\calA$-seed,
	$u$ is $\calA$-potential but not {$\calA\calB$-}diffusible, and 
	$u_0$ is {$\calA\calB$-}diffusible and $\calA$-suspended.
%Thus, $u$ belongs to $R$.

\end{itemize}

%In this case, $u$ is qualified for entry into $R$, 
%	if there is a live-edge path propagating $\calB$ from $u$ to a 
%	diffusible $\calA$-suspended  node $u_0$ so that
%	$u_0$ adopts both $\calA$ and $\calB$, and then propagate
%	the adoption to $v$ through $u$ again (see Example~\ref{ex:rrcim-tricky}).
To identify such $u$, we start two secondary BFS from $u$, 
	one traveling forwards, one backwards.
The forward search explores all $\calB$-diffusible nodes reachable from $u$ and
	puts them in a set $S_f$, and
	stops at a node $w$ when $w$ is not $\calB$-diffusible, but also puts  $w$
	in set $S_f$.
The backward search explores all $\calA\calB$-diffusible and 
	$\calA$-potential/suspended/adopted nodes that can reach $u$ and puts  them in
	set $S_b$.
If there is a node $u_0\in S_f \cap S_b$ that is $\calA$-suspended, then
	we can put $u$ into $R$.
After this special treatment, 
	we stop exploring the in-neighbors of $u$ in the primary search and
	continue the primary search elsewhere. We have:

%
%continues until hitting a node that cannot propagate $\calB$.
%They both traverse along live edges only.
%Let $S_f$ and $S_b$ be the set of nodes explored in the forward and backward
%	BFS, respectively.
%Their intersection contains the node $u_0$ that qualifies $u$.
%Hence, if $S_f \cap S_b \neq \emptyset$, we add $u$ to the RR-set.
%Otherwise, skip $u$'s in-neighbors and continue the primary search.
%The following example illustrates this case.

%
%\begin{example}\label{ex:rrcim-tricky}
%{\em
%
%\qed
%}
%\end{example}
%
%

\def\theoremRRCIMCorrect{
Suppose that $\qao \leq \qab$ and $\qbo \le \qba = 1$.
The RR-sets generated by
	the \timcim algorithm satisfies Definition~\ref{def:generalRRset} for
	the \CIM problem.
As a result, Theorem~\ref{thm:generalTIM} applies to \generalTIM with \timcim
	in this case.
}

\vspace*{-1mm}
\begin{theorem}\label{thm:rrcim-correct}
{\theoremRRCIMCorrect}
\end{theorem}
\vspace*{-1mm}

\spara{Expected Time Complexity}
Both phases of \timcim require more computations compared to \timsim.
First, the number of edges explored in Phase~\rom{1}, namely $\EPT_F$,
	is larger in \timcim, as the forward labeling here needs to 
	continue beyond just $\calA$-adopted nodes.
For Phase~\rom{2}, let $\EPT_{BS}$ be the expected number of
	edges pointing to nodes in $R$
	%that are $\calA$-suspended,
	and $\EPT_{BO}$ be the  expected number of all other
	edges examined in this phase (including both primary
	and secondary searches).
Thus, we have:

\def\lemmaRRCIMTime{
In expectation, \generalTIM with \timcim runs in
$O\left((k+\ell)(|V|+|E|)\log |V|\left(1+ \frac{\EPT_F+\EPT_{BO}}{\EPT_{BS}}\right) \right)$ 
time.
}
\begin{lemma}\label{lemma:rrcim-time}
{\lemmaRRCIMTime}
\end{lemma}

\begin{figure}
 \centering
   \includegraphics[width=0.45\textwidth]{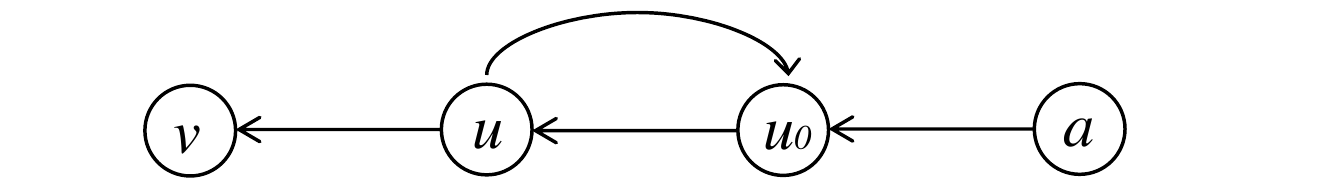}
 \caption{Case 4 of primary backward search in \timcim}\label{fig:rrcim-tricky}
\end{figure}

\subsection{The Sandwich Approximation Strategy}\label{sec:sandwich}
%
%First of all, based on Theorem~\ref{thm:submod-complement} and the seminal results in Nemhauser et al.~\cite{submodular}, the special case of SIM with one-way complementarity is approximable in polynomial time by a simple greedy hill-climbing algorithm (pseudo-code in Algorithm~\ref{alg:greedy-simple}) with a factor of $1-1/e-\epsilon$, for any $\epsilon > 0$.
%Note that $\epsilon$ stems from 
%	approximating $\sigma_A(S_\calA, S_\calB)$ and $\sigma_B(S_\calA, S_\calB)$ due to 
%	the $\#P$-hardness of their exact computation.
% 
%More generally, SIM with mutual complementarity does not enjoy the above result due to lack of self-submodularity, and neither does CIM for lack of cross-submodularity.
%However, the situation is not hopeless.
%In fact, we are able to devise an intelligent strategy that lead to a data-dependent approximation ratio.
%As we shall see shortly, our solution strategy not only applies to CIM and SIM problems studied here, \textsl{but is also generally applicable to any set function maximization problem that does not enjoy submodularity.}

We present the {\em Sandwich Approximation} (SA) strategy that leads to  algorithms with data-dependent approximation factors for \SIM and \CIM in the general mutual complement case of \comic ($\qao \leq \qab$ and $\qbo \leq \qba$) when submodularity may \emph{not} hold.
In fact, SA can be seen as a general strategy, applicable to any non-submodular maximization problems for which we can find submodular upper or lower bound
	functions.

Let $\sigma: 2^V \to \mathbb{R}_{\geq 0}$ be non-submodular.
Let $\mu$ and $\nu$ be submodular and defined on the same ground set $V$ such that $\mu(S) \leq \sigma(S) \leq  \nu(S)$ for all $S\subseteq V$.
That is, $\mu$ ($\nu$) is a lower (resp., upper) bound on $\sigma$ everywhere.
Consider the problem of maximizing $\sigma$ subject to a cardinality
	constraint $k$.
Notice that if the objective function were $\mu$ or $\nu$, the problem would be
	approximable within $1-1/e$ (e.g., max-$k$-cover) or $1-1/e-\epsilon$ (e.g.,  influence maximization)
	by the greedy algorithm~\cite{submodular, kempe03}.  
A natural question is: \textsl{Can we leverage the fact that $\mu$ and $\nu$ ``sandwich'' $\sigma$
	to derive an approximation algorithm for maximizing $\sigma$?}
The answer is ``yes''.

\spara{Sandwich Approximation}
First, run the greedy algorithm on all three functions. It produces an approximate solution for $\mu$ and $\nu$. 
Let $S_\mu$, $S_\sigma$, $S_\nu$ be the solution obtained for $\mu$, $\sigma$, and $\nu$ respectively.
Then, select the final solution to $\sigma$ to be
\begin{align}\label{eqn:sand}
S_{\mathit{sand}} = \argmax_{S \in \{S_\mu, S_\sigma, S_\nu\}} \sigma(S).  
\end{align}

%Consider the problem of maximizing $f(\cdot)$ subject to cardinality constraint $k$, i.e., finding a set $S^* \in \argmax_{S\subseteq U, |S| = k} f(S)$.
%%We run the following greedy-based algorithm on $f(\cdot)$, which as we show (Theorem~\ref{thm:sandwich} below) is in fact an approximation algorithm.
%Since $f(\cdot)$ is not submodular, Algorithm~\ref{alg:greedy-simple} {\em per se} does not provide approximation guarantee.
%However, we could run it on all of the three functions: $f(\cdot)$ itself, upper bound $h(\cdot)$, and lower bound $g(\cdot)$, and return the best solution set.
%Since $f(\cdot)$ is ``sandwiched'' by the other two functions, we call this algorithm ``Sandwich-Greedy'' (Algorithm~\ref{alg:greedy-sandwich}).
%Importantly, \textsl{Sandwich-Greedy produces approximation guarantees}.

%\begin{algorithm}[t!]
%\caption{\textsc{Sandwich-Greedy}}\label{alg:greedy-sandwich}
%Run Algorithm~\ref{alg:greedy-simple} with input $U$, $h$, $k$: output $S_\nu$\;
%Run Algorithm~\ref{alg:greedy-simple} with input $U$, $g$, $k$: output $S_g$\;
%Run Algorithm~\ref{alg:greedy-simple} with input $U$, $f$, $k$: output $S_\sigma$\;
%Output $T = \argmax_{S\in \{S_\sigma, S_\nu, S_g\}} f(S)$
%\end{algorithm}

%\def\theoremSA{
%Sandwich Approximation solution gives:
%\begin{align} \label{eqn:sandapprox}
%\sigma(S_{\mathit{sand}}) \ge \max \Big\{\frac{\sigma(S_\nu)}{\nu(S_\nu)}, \frac{\mu(S_\sigma^*)}{\sigma(S_\sigma^*)} \Big\}
%  \cdot (1-1/e) \cdot  \sigma(S_\sigma^*),
%\end{align}
%where $S_\sigma^*$ is the optimal solution maximizing $\sigma$ (subject to  cardinality constraint $k$).
%}

\begin{theorem}\label{thm:sandwich}
Sandwich Approximation solution gives:
\begin{align} \label{eqn:sandapprox}
\sigma(S_{\mathit{sand}}) \ge \max \Big\{\frac{\sigma(S_\nu)}{\nu(S_\nu)}, \frac{\mu(S_\sigma^*)}{\sigma(S_\sigma^*)} \Big\}
  \cdot (1-1/e) \cdot  \sigma(S_\sigma^*),
\end{align}
where $S_\sigma^*$ is the optimal solution maximizing $\sigma$ (subject to  cardinality constraint $k$).
\end{theorem}

\noindent\textbf{Remarks.}
While the factor in Eq. \eqref{eqn:sandapprox} involves $S_\sigma^*$, generally not computable in polynomial time, the first term inside max\{.,.\}, involves $S_\mu$ can be computed efficiently and can be of practical value (see Table~\ref{tab:SA-factor} in \textsection\ref{sec:exp}). 
%can still be meaningful, especially when $\sigma(S_\mu) > \max\{\sigma(S_\sigma), \sigma(S_\nu)\}$.
We emphasize that SA is much more general, not restricted to cardinality constraints.
E.g., for a general matroid constraint, simply replace $1-1/e$ with $1/2$ in \eqref{eqn:sandapprox}, as the greedy algorithm is a $1/2$-approximation in this case~\cite{submodular}.
Furthermore, monotonicity is not important, as maximizing general submodular functions can be approximated within a factor of $1/2$~\cite{buchbinder12}, and thus SA applies regardless of monotonicity.
On the other hand, the true effectiveness of SA depends on how close $\nu$ and $\mu$ are to $\sigma$: e.g., a constant function can be a trivial submodular upper bound function but would only yield trivial data-dependent approximation factors. Thus, an interesting question is how to derive $\nu$ and $\mu$ that are as close to $\sigma$ as possible while maintaining submodularity.
%On the other hand, we also caution that trivial data-dependent bound can be obtained through SA, e.g., a constant function is also submodular, and needless to say, using such bounds are meaningless.
%Thus, the true effectiveness of SA depends on the tightness of the bound, i.e., how to derive $\nu$ and $\mu$ that are as close to $\sigma$ as possible, while maintaining submodularity.
%This is certainly an interesting direction for future research.

Next, we apply SA to both \SIM and \CIM in the general mutual complementarity case ($\bQ^+$).
%The following theorem lays the foundation.

\spara{\SIM}
$\mathsf{GeneralTIM}$ with \timsim %(Algorithm~\ref{alg:rrset-sim}) 
	or \timsimfast %(Algorithm~\ref{alg:rrset-sim-fast}) 
	provides a $(1-1/e-\epsilon)$-approximate
	solution with high probability, when $\qao \leq \qab$ and $\qbo = \qba$.
When $\qbo < \qba$, function $\nu$ (upper bound) can be obtained by
	increasing $\qbo$ to $\qba$, while $\mu$ (lower bound) can be obtained
	by decreasing $\qba$ to $\qbo$.

\spara{\CIM}
$\mathsf{GeneralTIM}$ with \timcim provides a $(1-1/e-\epsilon)$-approximate
	solution with high probability, when $\qao \leq \qab$ and $\qbo \leq \qba = 1$.
When $\qba$ is not necessarily $1$, we obtain an upper bound
	function by increasing $\qba$ to $1$.
%There will be no lower bound function here, but SA still works, due to Eq. \eqref{eqn:Snu}.

%The correctness of the above approaches is ensured by the following theorem.
\def\theoremMonoQ{
Suppose $\qao \leq \qab$ and $\qbo \leq \qba$.
Then, under the \comic model, for any fixed $\calA$ and $\calB$ seed sets $S_\calA$ and
	$S_\calB$,
	$\sigma_\calA(S_\calA, S_\calB)$ is monotonically increasing w.r.t.\
	any one of $\{\qao, \qab, \qbo, \qba\}$ with other three GAPs fixed, as
	long as after the increase the parameters are still in $\bQ^+$.
}

%\vspace*{-2mm}
\begin{theorem}\label{thm:mono-q}
{\theoremMonoQ}
\end{theorem}
%\vspace*{-3mm}

\smallskip
Putting it all together, the final algorithm for \SIM is $\mathsf{GeneralTIM}$ with \timsim/\timsimfast and SA.
Similarly, the final algorithm for \CIM is Algorithm~\ref{alg:generalTIM} with \timcim and SA.
It is important to see how useful and effective SA is in practice. We address this question head on in \textsection\ref{sec:exp}, where we ``stress test'' the idea behind SA. Intuitively, if the GAPs are such that $\qbo$ and $\qba$ are close, the upper and lower bounds ($\nu$ and $\mu$) obtained for \SIM can be expected to be quite close to $\sigma$. Similarly, when $\qba$ is close to $1$, the corresponding upper bound for \CIM should be quite close to $\sigma$. We consider settings where $\qbo$ and $\qba$ are separated apart and similarly $\qba$ is not close to $1$ and measure the effectiveness of SA (see Table~\ref{tab:SA-factor}).
\eat{
\blue{Intuitively, the above SA schemes for \SIM and \CIM would be effective when
	(a) $\qao$ and $\qab$ are close to each other, since both problems are targeted at
	influence spread of $\calA$ and thus if $\calA$'s dependency on $\calB$ is small,
	then the upper and lower bounds $\nu$ and $\mu$  after changing $\qbo$ or
	$\qba$ should be close to the true influence spread $\sigma$; or
	(b) $\qba$ is close to $\qbo$ for \SIM or close to $1$ for \CIM, since it
	means the difference in $\calB$ adoption in upper and lower bound cases do not vary
	much; or
	(c) the edge probabilities are relatively small, and thus the difference of upper
	and lower bound cases will not be amplified through long chain of cascades.
In \textsection\ref{sec:realexp} we empirically demonstrate that SA is highly 
	effective for both problems in real networks with parameters learned from
	real trace data.
}
}

\section{Experiments}\label{sec:exp}

\begin{table*}[t!]
	%\scriptsize
	\centering
	\begin{tabular}{rcccc}
	 &  \dbBook & \dbMovie & \flix & \lastfm \\ \hline
	{ \# nodes}			& $23.3$K	& $34.9$K & $12.9$K	& $61$K  \\ 
	{ \# edges} 		&  $141$K	& $274$K & $192$K & $584$K   \\ 
	 { avg. out-degree} & $6.5$ & $7.9$	& $14.8$ & $9.6$  \\ 
	 { max. out-degree} & $1690$  & $545$ & $189$  &  $1073$   \\ \hline
	\end{tabular}
%\vspace{-2mm}
	\caption{Statistics of graph data (all directed)}
%\vspace{-5mm}
	 \label{tab:datasets}
\end{table*}

We perform extensive experiments on three real-world social
	networks.
	% to evaluate our algorithms: \flix, \douban, and \lastfm. 
% and also conduct scalability tests on larger synthetic graphs.
We first present results with synthetic GAPs (\textsection\ref{sec:synGAPexp});
	then we propose a method for learning GAPs using action log
	data (\textsection\ref{sec:learn}), and conduct experiments using
	learnt GAPs (\textsection\ref{sec:expReal}).

\spara{Datasets}
\flix is collected from a social movie site %\footnote{\url{http://www.cs.ubc.ca/~jamalim/datasets/}},
	and we extract a strongly connected component.
\douban is collected from a Chinese social network~\cite{yuan2013we}, where users
	rate  books, movies, music, etc.
We crawl all movie \& book ratings of the users in the
	graph, and derive two datasets
	from book and movie ratings: \dbBook and \dbMovie.
\lastfm is taken from the popular music website with social networking features.
Table~\ref{tab:datasets} presents basic stats of the datasets.
For all graphs, we learn influence probabilities on edges using the method proposed in~\cite{amit2010},
	which is widely adopted in prior work~\cite{infbook}.
Links in Flixster and Last.fm networks are undirected, and we direct them in both directions.
Links in Douban network are derived from follower-followee relationships and in our dataset,
	there is an edge from $u$ to $v$ if $v$ follows $u$ on Douban.

\subsection{Experiments with Synthetic GAPs}\label{sec:synGAPexp}

\begin{table*}
\centering
\begin{tabular}{rccc|ccc}
		\cline{2-7}
		\underline{\bf\SIM} & \multicolumn{3}{|c|}{\vanillaIC} & \multicolumn{3}{c|}{\copying} \\ \cline{2-7}
		$\qao$ & \multicolumn{1}{|c}{$0.1$} & $0.3$ & $0.5$ & $0.1$ & $0.3$ & \multicolumn{1}{c|}{$0.5$} \\ \cline{2-7}
		\dbBook & \multicolumn{1}{|c}{$5.89\%$} & $0.93\%$ & $0.50\%$ & $85.7\%$ & $207\%$ & \multicolumn{1}{c|}{$301\%$} \\
		\dbMovie & \multicolumn{1}{|c}{$24.7\%$} & $3.30\%$ & $1.72\%$ & $13.3\%$ & $68.8\%$ & \multicolumn{1}{c|}{$122\%$} \\
		\flix & \multicolumn{1}{|c}{$35.5\%$} & $11.3\%$ & $5.15\%$ & $16.7\%$ & $48.0\%$ & \multicolumn{1}{c|}{$84.8\%$} \\
		\lastfm & \multicolumn{1}{|c}{$31.5\%$} & $2.75\%$ & $0.70\%$ & $22.6\%$ & $88.5\%$ & \multicolumn{1}{c|}{$168\%$} \\	
		\cline{2-7}
		\cline{2-7}
		\underline{\bf\CIM}& \multicolumn{3}{|c|}{\vanillaIC} & \multicolumn{3}{c|}{\copying} \\ \cline{2-7}
		$\qbo$ & \multicolumn{1}{|c}{$0.1$} & $0.5$ & $0.8$ & $0.1$ & $0.5$ & \multicolumn{1}{c|}{$0.8$} \\ \cline{2-7}
		\dbBook & \multicolumn{1}{|c}{$13.4\%$} & $31.2\%$ & $25.6\%$ & $49.9\%$ & $32.9\%$ & \multicolumn{1}{c|}{$30.7\%$} \\
		\dbMovie & \multicolumn{1}{|c}{$135\%$} & $151\%$ & $101\%$ & $14.4\%$ & $7.46\%$ & \multicolumn{1}{c|}{$3.84\%$} \\
		\flix & \multicolumn{1}{|c}{$81.7\%$} & $58.5\%$ & $24.9\%$ & $13.6\%$ & $8.21\%$ & \multicolumn{1}{c|}{$10.7\%$} \\
		\lastfm & \multicolumn{1}{|c}{$140\%$} & $110\%$ & $48.3\%$ & $10.6\%$ & $9.12\%$ & \multicolumn{1}{c|}{$6.77\%$} \\
		\cline{2-7}
\end{tabular}
\caption{Percentage improvement of \generalTIM over \vanillaIC \& \copying, where the other seed set is chosen to be the 101st -- 200th from the \vanillaIC order}\label{tab:synGAP}
\end{table*}

%%%%%%%%%%%%%%%%%%%%%%%%%%%%%%%%%%%%%%%%%%%%%%%%%%
\begin{table*}[h!t!]
\centering
	\begin{tabular}{rccc|ccc}
		\cline{2-7}
		\underline{\bf\SIM} & \multicolumn{3}{|c|}{\vanillaIC} & \multicolumn{3}{c|}{\copying} \\ \cline{2-7}
		$\qao$ & \multicolumn{1}{|c}{$0.1$} & $0.3$ & $0.5$ & $0.1$ & $0.3$ & \multicolumn{1}{c|}{$0.5$} \\ \cline{2-7}
		\dbBook & \multicolumn{1}{|c}{$2.16\%$} & $1.12\%$ & $0.71\%$ & $133\%$ & $419\%$ & \multicolumn{1}{c|}{$676\%$} \\
		\dbMovie & \multicolumn{1}{|c}{$4.38\%$} & $1.49\%$ & $0.87\%$ & $236\%$ & $737\%$ & \multicolumn{1}{c|}{$1283\%$} \\
		\flix & \multicolumn{1}{|c}{$10.6\%$} & $0\%$ & $0\%$ & $134\%$ & $352\%$ & \multicolumn{1}{c|}{$641\%$} \\
		\lastfm & \multicolumn{1}{|c}{$3.76\%$} & $2.65\%$ & $1.65\%$ & $398\%$ & $1355\%$ & \multicolumn{1}{c|}{$2525\%$} \\	
		\cline{2-7}
		\cline{2-7}
		\underline{\bf\CIM}& \multicolumn{3}{|c|}{\vanillaIC} & \multicolumn{3}{c|}{\copying} \\ \cline{2-7}
		$\qbo$ & \multicolumn{1}{|c}{$0.1$} & $0.5$ & $0.8$ & $0.1$ & $0.5$ & \multicolumn{1}{c|}{$0.8$} \\ \cline{2-7}
		\dbBook & \multicolumn{1}{|c}{$548\%$} & $644\%$ & $717\%$ & $834\%$ & $238\%$ & \multicolumn{1}{c|}{$252\%$} \\
		\dbMovie & \multicolumn{1}{|c}{$1060\%$} & $994\%$ & $713\%$ & $637\%$ & $106\%$ & \multicolumn{1}{c|}{$89.2\%$} \\
		\flix & \multicolumn{1}{|c}{$956\%$} & $152\%$ & $102\%$ & $888\%$ & $119\%$ & \multicolumn{1}{c|}{$95\%$} \\
		\lastfm & \multicolumn{1}{|c}{$1361\%$} & $982\%$ & $710\%$ & $489\%$ & $65.3\%$ & \multicolumn{1}{c|}{$63.6\%$} \\
		\cline{2-7}
	\end{tabular}
\caption{Percentage improvement of \generalTIM over \vanillaIC \& \copying, where the other seed set is randomly chosen}
\label{tab:synGAP-random}
\end{table*}

%%%%%%%%%%%%%%%%%%%%%%%%%%%%%%%%%%%%%%%%%%%%%%%%%%
\begin{table*}[h!t!]
\centering
	\begin{tabular}{rccc|ccc}
		\cline{2-7}
		\underline{\bf\SIM} & \multicolumn{3}{|c|}{\vanillaIC} & \multicolumn{3}{c|}{\copying} \\ \cline{2-7}
		$\qao$ & \multicolumn{1}{|c}{$0.1$} & $0.3$ & $0.5$ & $0.1$ & $0.3$ & \multicolumn{1}{c|}{$0.5$} \\ \cline{2-7}
		\dbBook & \multicolumn{1}{|c}{$0.34\%$} & $0.34\%$ & $0\%$ & $0.34\%$ & $0.34\%$ & \multicolumn{1}{c|}{$0\%$} \\
		\dbMovie & \multicolumn{1}{|c}{$0.64\%$} & $0.54\%$ & $0.36\%$ & $0.64\%$ & $0.54\%$ & \multicolumn{1}{c|}{$0.36\%$} \\
		\flix & \multicolumn{1}{|c}{$0\%$} & $0\%$ & $0\%$ & $0\%$ & $0\%$ & \multicolumn{1}{c|}{$0\%$} \\
		\lastfm & \multicolumn{1}{|c}{$1.73\%$} & $1.38\%$ & $0.80\%$ & $1.73\%$ & $1.38\%$ & \multicolumn{1}{c|}{$0.80\%$} \\	
		\cline{2-7}
		\cline{2-7}
		\underline{\bf\CIM}& \multicolumn{3}{|c|}{\vanillaIC} & \multicolumn{3}{c|}{\copying} \\ \cline{2-7}
		$\qbo$ & \multicolumn{1}{|c}{$0.1$} & $0.5$ & $0.8$ & $0.1$ & $0.5$ & \multicolumn{1}{c|}{$0.8$} \\ \cline{2-7}
		\dbBook & \multicolumn{1}{|c}{$3.53\%$} & $1.19\%$ & $0.24\%$ & $1.03\%$ & $0.18\%$ & \multicolumn{1}{c|}{$0.05\%$} \\
		\dbMovie & \multicolumn{1}{|c}{$141\%$} & $22.2\%$ & $9.87\%$ & $138\%$ & $20.7\%$ & \multicolumn{1}{c|}{$8.67\%$} \\
		\flix & \multicolumn{1}{|c}{$3.93\%$} & $1.63\%$ & $-1.69\%$ & $3.93\%$ & $1.63\%$ & \multicolumn{1}{c|}{$-1.69\%$} \\
		\lastfm & \multicolumn{1}{|c}{$5.38\%$} & $2.19\%$ & $0.69\%$ & $2.72\%$ & $0.11\%$ & \multicolumn{1}{c|}{$-0.97\%$} \\
		\cline{2-7}
	\end{tabular}
\caption{Percentage improvement of \generalTIM over \vanillaIC \& \copying,
	where the other seed set is chosen to be the top-100 nodes by \vanillaIC}
\label{tab:synGAP-top100}
\end{table*}

%%%%%%%%%%%%%%%%%%%%%%%%%%%%%%%%%%%%%%%%%%%%%%%%%%

We first evaluate our proposed algorithms
	using synthetic GAPs.
We compare with two intuitive baselines:
$(i)$ \vanillaIC: It selects $k$ seeds using TIM algorithm~\cite{tang14} under the
	classic IC model, essentially ignoring the other product and the NLA in \comic model;
$(ii)$ \copying: For \SIM, it simply selects the top-$k$ $\calB$-seeds to be $\calA$-seeds
	%(assuming $|S_\calB| > |S_\calA|$) 
	and vice versa for \CIM.
	% (assuming $|S_\calA| > |S_\calB|$).

In \SIM, we set $\qab = \qba = 0.75$, $\qbo = 0.5$, and $\qao$ is set to 
$0.1$, $0.3$, $0.5$, which represent strong, moderate, and low complementarity.
%The input $\calB$-seeds are the 101st to 200th seeds in the greedy order as
%	selected by running TIM under the classic IC model.
In \CIM, we set $\qao = 0.1$, $\qab = \qba = 0.9$, such that the room for $\calB$
	to complement $\calA$ is sufficiently large to distinguish between algorithms.
We vary $\qbo$ to be $0.1$, $0.5$, and $0.8$.

%The input $\calA$-seeds are also the 101st to 200th seeds in the greedy order.
Lots of possibilities exist for setting the opposite seed set, i.e.,
	$\calB$-seeds for \SIM and $\calA$-seeds for \CIM.
We test three representative cases:
(1) randomly selecting 100 nodes -- this models our complete lack of knowledge;
(2) running \vanillaIC and selecting the top-100 nodes -- this models a situation where we assume the advertiser might use an advanced algorithm such as TIM to target highly influential users;
(3) running \vanillaIC and selecting the 101st to 200th nodes -- this models a situation where we assume those seeds are moderately influential.

Table~\ref{tab:synGAP} shows the percentage improvement of our algorithms over the two baselines, for the case of
	selecting the 101st to 200th nodes of \vanillaIC as the fixed opposite seed set.
As can be seen, \generalTIM performs consistently better than both baselines, and in many cases by a large margin.

Table~\ref{tab:synGAP-random} shows the percentage improvement of
	\generalTIM over \vanillaIC and \copying baselines when the other
	seed set ($\calB$-seeds for \SIM and $\calA$-seeds for \CIM)
	consists of 100 {\em random nodes}.
As can be seen, \generalTIM is significantly better except when
	comparing to \vanillaIC in \SIM.
This is not surprising, as when $\calB$-seeds are chosen randomly, they
	are unlikely to be very influential and hence it is rather
	safe to ignore them when selecting $\calA$-seeds, which is
	essentially how \vanillaIC operates.

Table~\ref{tab:synGAP-top100} shows the results when the other seed
	set is chosen to be the {\em top-100 nodes} from \vanillaIC.
These nodes represent the most influential ones under the IC model
	that can be found efficiently in polynomial time\footnote{Recall that influence maximization is NP-hard under the IC model, and thus the optimal top-100 most influential nodes are difficult to find. The seeds found by \vanillaIC can be regarded as a good proxy.}.
In this case, the advantage of \generalTIM over \vanillaIC and \copying
	is less significant.
On \flix dataset, the three algorithms achieve the same influence spread
	for \SIM, while for \CIM, when $\qbo = 0.8$, we even observe a slight
	disadvantage of \generalTIM in certain cases. 
The reason is that when $\qbo = 0.8$ (quite close to $1$), targeting
	the input $\calA$-seeds to be $\calB$-seeds is itself a good
	strategy (see Theorem~\ref{thm:cimOpt}).
For \SIM, we remark that \vanillaIC and \copying is equivalent.
%They are both very effective as the input $\calB$-seeds are already quite influential.

Overall, considering Tables \ref{tab:synGAP}, \ref{tab:synGAP-random}, and \ref{tab:synGAP-top100}
	all together, we can see that in the vast majority of all test cases,
	\generalTIM outperforms these two baselines, often by a large margin.
This demonstrates that \generalTIM is robust w.r.t. different selection methods of the opposite seed set.
Furthermore, in real-world scenarios, the opposite seed sets may simply consist of ``organic'' early adopters, i.e., users who
	adopt the product spontaneously. %, instead of being targeted by the advertiser.
The robustness of \generalTIM is thus highly desirable as it is often
	difficult to foresee which users would actually become organic early adopters in real life.

Also, in our model the influence probabilities on edges are assumed
	independent of the product; without this assumption \copying and \vanillaIC
	would perform even more poorly.
If we additionally assume that the GAPs are user-dependent, \vanillaIC would
	deteriorate further.
In contrast, our \generalTIM and RR-set generation algorithms can be easily
	adapted to both these scenarios.

\subsection{Learning GAPs from Real Data}\label{sec:learn}

\spara{Finding Signals from Data}
For \flix and \douban, we learn GAPs from timestamped rating data,
	which can be viewed as action logs.
Each entry is a quadruple
	$(u, i, a, t_{u,i,a})$, indicating user 
	$u$ performed action $a$ on item $i$ at time $t_{u,i,a}$. 
We count a rating quadruple as one {\em adoption action} and
	one {\em informing action}: if someone rated an item, she must
	have been informed of it first, as we assume only adopters rate items.
{A key challenge is how to find actions that can be mapped to 
	informing events that do not lead to adoptions.}
Fortunately, %in addition to numeric ratings,
	there are special ratings providing such signals % of informing events
         in \flix and \douban.
The former allows users to indicate if they ``want to see''  a movie,
	or are ``not interested'' in one.
We map both signals to the actions of a user being informed of a movie.
The latter allows users to put items into a wish list.
%\footnote{Once this user later rates this item, it is automatically removed from the wish list}.
Thus, if a book/movie is in a user's wish list, we treat it as an {\em informing action}.
For \douban, we separate actions on books and movies to derive two datasets: \dbBook and \dbMovie.

\spara{Learning Method}
Consider two items $\calA$ and $\calB$ in an action log.
%(and there are two types of actions on each item: adoption and informing).
Let $R_\calA$ and $I_\calA$ be the set of users who rated $\calA$ and who were informed of $\calA$, respectively.
Clearly, $R_\calA \subseteq I_\calA$. % and $R_\calB \subseteq I_\calB$. 
Thus, \[\qao = {|R_\calA \setminus R_{\calB \prec_{\mathit{rate}} \calA}|} \; / \; {|I_\calA \setminus R_{\calB \prec_{\mathit{inform}} \calA}|},\]
where $R_{\calB \prec_{\mathit{rate}} \calA}$ is the set of users who rated both items with $\calB$ rated first,
and $R_{\calB \prec_{\mathit{inform}} \calA}$ is the set of users who rated $\calB$ before being informed of $\calA$. 
Next, $\qab$ is computed as follows: \[\qab = {|R_{\calB \prec_{\mathit{rate}} \calA}|} \;  / \; {|R_{\calB \prec_{\mathit{inform}} \calA}|}.\]
Similarly, $\qbo$ and $\qba$ can be computed in a symmetric way.

%Table~\ref{tab:Q} presents the GAPs learned for a few pairs of popular movies in \flix dataset.
%More examples, including those in \dbBook and \dbMovie and $95\%$-confidence intervals
%	of those estimates, can be found in ~\cite{comicTR}.
	
Tables~\ref{tab:Q-flix} -- \ref{tab:Q-dbMovie} depict selected 
	GAPs learned from \flix, \dbBook, and \dbMovie datasets, using methods in \textsection\ref{sec:exp}.
Here we not only show the estimated value of $\bQ$, but also give 95\% confidence intervals.
By the definition of GAPs (\textsection\ref{sec:model}), we can treat each GAP as the parameter of
	a Bernoulli distribution.
Consider any GAP, denoted $q$, and let $\bar{q}$ be its estimated value from action log data.
The 95\% confidence interval\footnote{See any standard textbooks on probability theory and statistics} of $\bar{q}$ is given by
\[
\left[ \bar{q} - 1.96\sqrt{\bar{q}(1-\bar{q})/n_q}, \; \bar{q} + 1.96\sqrt{\bar{q}(1-\bar{q})/n_q} \right],
\]
where $n_q$ is the number of samples used for estimating $q$.

%\begin{table*}[t!]
%\centering
%	\begin{tabular}{cccccc}
%		$\calA$ & $\calB$ & $\qao$ & $\qab$ & $\qbo$ & $\qba$ \\  \hline
%		{\it Monster Inc.} & {\it Shrek} &  $.88$ & $.92$ & $.92$ & $.96$ \\ 
%		{\it Gone in 60 Seconds} & {\it Armageddon} &  $.63$ & $.77$ & $.67$ & $.82$ \\ 
%		{\it Prisoner of Azkaban} & {\it What a Girl Wants} & $.85$ & $.84$ & $.66$ & $.67$ \\ 
%		{\it Shrek} & {\it Fast and Furious}  & $.92$ & $.94$ & $.80$ & $.79$ \\ \hline
%	\end{tabular}
%	\caption{Selected GAPs learned for movies from \flix}
%	 \label{tab:Q}
%\end{table*}
%

\begin{table*}[t!]
	\centering
	\small
	\begin{tabular}{cccccc}
		$\calA$ & $\calB$ & $\qao$ & $\qab$ & $\qbo$ & $\qba$ \\  \hline
		{\it Monster Inc.} & {\it Shrek} &  $.88 \pm .01$ & $.92 \pm .01$ & $.92 \pm .01$ & $.96 \pm .01$ \\ 
		{\it Gone in 60 Seconds} & {\it Armageddon} &  $.63 \pm .02$ & $.77 \pm .02$ & $.67 \pm .02$ & $.82 \pm .02$ \\ 
		{\it Harry Porter: Prisoner of Azkaban} & {\it What a Girl Wants} & $.85\pm .01$ & $.84\pm .02$ & $.66\pm .02$ & $.67 \pm .02$ \\ 
		{\it Shrek} & {\it The Fast and The Furious}  & $.92\pm .02$ & $.94\pm .01$ & $.80\pm .02$ & $.79\pm .02$ \\ \hline
	\end{tabular}
	\caption{Selected GAPs learned for movies from \flix}
	 \label{tab:Q-flix}
\end{table*}

\begin{table*}[t!]
	\centering
	\scriptsize
	\begin{tabular}{cccccc}
		$\calA$ & $\calB$ & $\qao$ & $\qab$ & $\qbo$ & $\qba$ \\  \hline
		{\it The Unbearable Lightness of Being} & {\it Norwegian Wood (Japanese)} &  $.75 \pm .01$ & $.85\pm .02$ & $.92\pm .01$ & $.97\pm .01$ \\ 
		{\it Harry Potter and the Philosopher's Stone} & {\it Harry Potter and the Half-Blood Prince} &  $.99\pm 1$ & $1.0 \pm 0$ & $.97\pm .01$ & $.98\pm .01$ \\ 
		{\it Stories of Ming Dynasty III (Chinese)} & {\it Stories of Ming Dynasty VI (Chinese)} & $.94\pm .01$ & $1.0\pm 0$ & $.88\pm .03$ & $.98\pm .01$ \\ 
		{\it Fortress Besieged (Chinese)} & {\it Love Letter (Japanese)}  & $.89\pm .01$ & $.91\pm .03$ & $.82\pm .02$ & $.83 \pm .03$ \\ \hline
	\end{tabular}
	\caption{Selected GAPs learned for movies from \dbBook}
	 \label{tab:Q-dbBook}
\end{table*}
\begin{table*}[t!]
	\centering
	\begin{tabular}{cccccc}
		$\calA$ & $\calB$ & $\qao$ & $\qab$ & $\qbo$ & $\qba$ \\  \hline
		{\it Up } & {\it 3 Idiots} &  $.92\pm .01$ & $.94\pm .01$ & $.92 \pm .01$ & $.93 \pm .01$ \\ 
		{\it Pulp Fiction} & {\it Leon} &  $.81\pm .01$ & $.83\pm .01$ & $.95\pm .00$ & $.98\pm .01$ \\ 
		{\it The Silence of the Lambs} & {\it Inception} & $.90 \pm .01$ & $.86 \pm .01$ & $.92\pm .01$ & $.98\pm .01$ \\ 
		{\it Fight Club} & {\it Se7en}  & $.84\pm .01$ & $.89 \pm .01$ & $.89\pm .01$ & $.95\pm .01$ \\ \hline
	\end{tabular}
	\caption{Selected GAPs learned for movies from \dbMovie}
	 \label{tab:Q-dbMovie}
\end{table*}

\subsection{Experiments with Learned GAPs}\label{sec:expReal}

\spara{Baselines}
We compare \generalTIM with several baselines commonly used
	in the literature:
\highdeg: choose the $k$ highest out-degree nodes as seeds;
\pagerank: choose the $k$ nodes with highest PageRank score;
\rand: choose $k$ seeds uniformly at random.
We also include the \greedy algorithm~\cite{kempe03}
	with 10K iterations of MC simulations
	to compute influence spread for the \comic diffusion processes.
%Due to the simulations, it is prohibitively expensive: in one week,
%it only finishes on \flix (both \SIM and \CIM) and \dbBook (\SIM only) .
\vanillaIC and \copying are omitted as the results are
	similar to those in \textsection\ref{sec:synGAPexp} (when the GAPs are close to each other).

%\note[Wei]{There was confusion from Reviewer 1 regarding Greedy,
%	w/ Monte Carlo, and thus I gave a better description.}

\spara{Parameters}
The following pairs of items are tested:
\begin{itemize}
\item \dbBook: {\em The Unbearable Lightness of Being} as $\calA$
	and {\em Norwegian Wood} as $\calB$, and $\bQ = \{0.75, 0.85, 0.92, 0.97\}$.
\item \dbMovie: {\em Fight Club} as $\calA$ and {\em Se7en}
	as $\calB$, and $\bQ = \{0.84, 0.89, 0.89, 0.95\}$. 
\item \flix: {\em Monster Inc} as $\calA$ and {\em Shrek} as $\calB$. 
\item \lastfm: There is no signal in the data to
	indicate informing events, so the learning method in
	\textsection\ref{sec:learn} is not applicable.
As a result, we use synthetic $\bQ = \{0.5, 0.75, 0.5, 0.75\}$.
\end{itemize}

In all four datasets, $\calA$ and $\calB$ are mutually complementary,
	for which self/cross-submodularity does not hold (\textsection\ref{sec:submod}).
Hence, Sandwich Approximation (SA) are used by default for
	\generalTIM and \greedy~\cite{kempe03}. 
Unless otherwise stated, $k=50$.
For \generalTIM,  $\ell=1$ so that a success probability
	of $1-1/|V|$ is ensured~\cite{tang14}.
In \SIM (resp.\ \CIM), the input $\calB$-seeds (resp.\ $\calA$-seeds)
	are chosen to be the 101st to 200th seeds selected by \vanillaIC.
%For \SIM, $\calB$-seeds are the top-10 highest degree nodes so that
%	\greedy can at least finish on the smallest dataset (\flix), and
%	$\epsilon=0.5$ for \timsim and \timsimfast.
%For \CIM, $\calA$-seeds
%	are the same as described in \textsection\ref{sec:synGAPexp}, and
%	$\epsilon = 1$ for \timcim.
We set $\epsilon=0.5$, which is chosen to achieve
	a balance between efficiency (running time) and effectiveness (seed set quality).
In what follows we empirically validate that influence spread is {\sl almost
	completely unaffected} when $\epsilon$ varies from $0.1$
	to $1$.
	
%\note[Wei]{For lack of space, I move those $\epsilon$-related plots to arXiv version.}
	%, due to that
	%$\qao$ and $\qab$ are relatively close
	%in our setting, thus more $\calA$-seeds would allow greater
	%complementary effects (from $\calB$) to be mined.
	 
Algorithms are implemented in C++ and compiled using g++ O3
	optimization.
We run experiments on an openSUSE Linux server with 2.93GHz CPUs and 128GB RAM.

\begin{figure*}
\centering
\begin{tabular}{cccc}
\includegraphics[width=.4\textwidth]{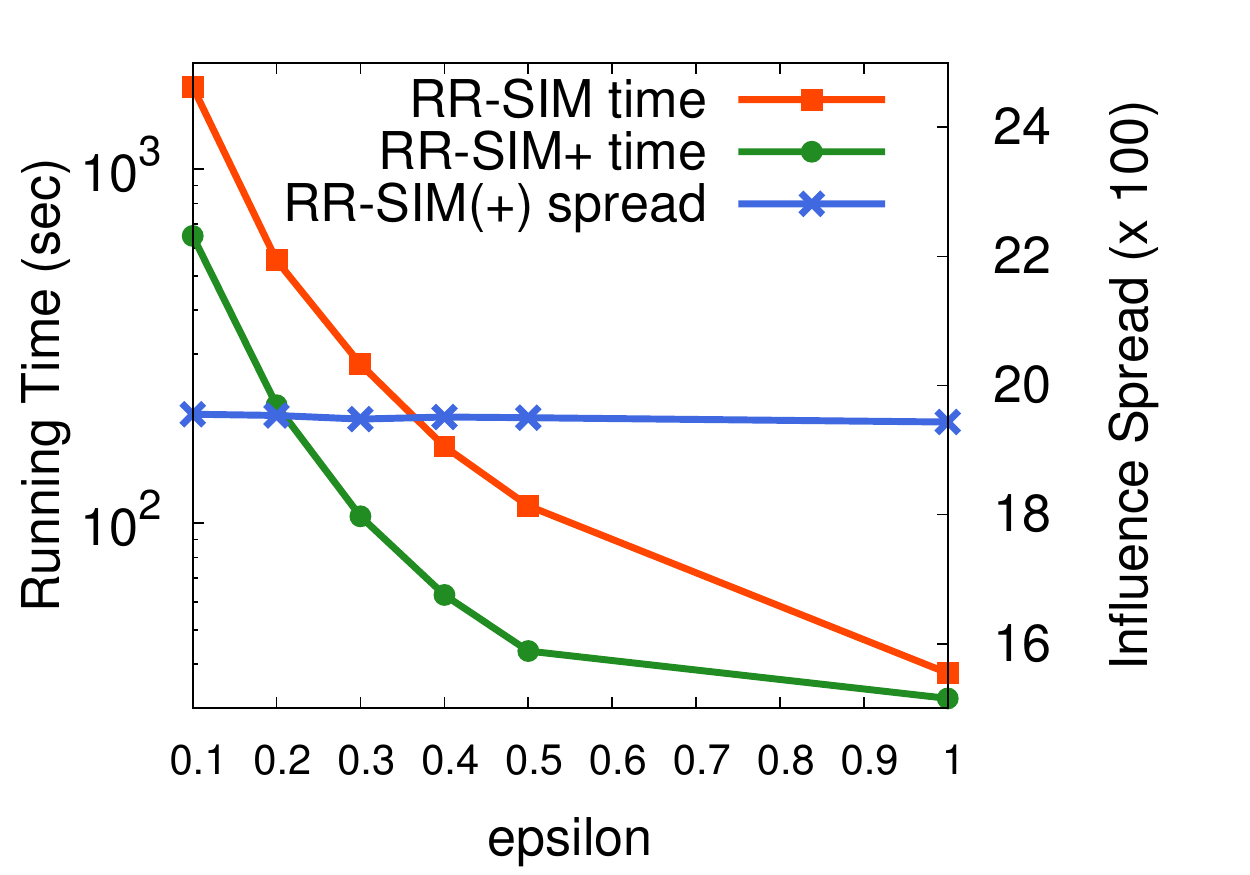}&
\includegraphics[width=.4\textwidth]{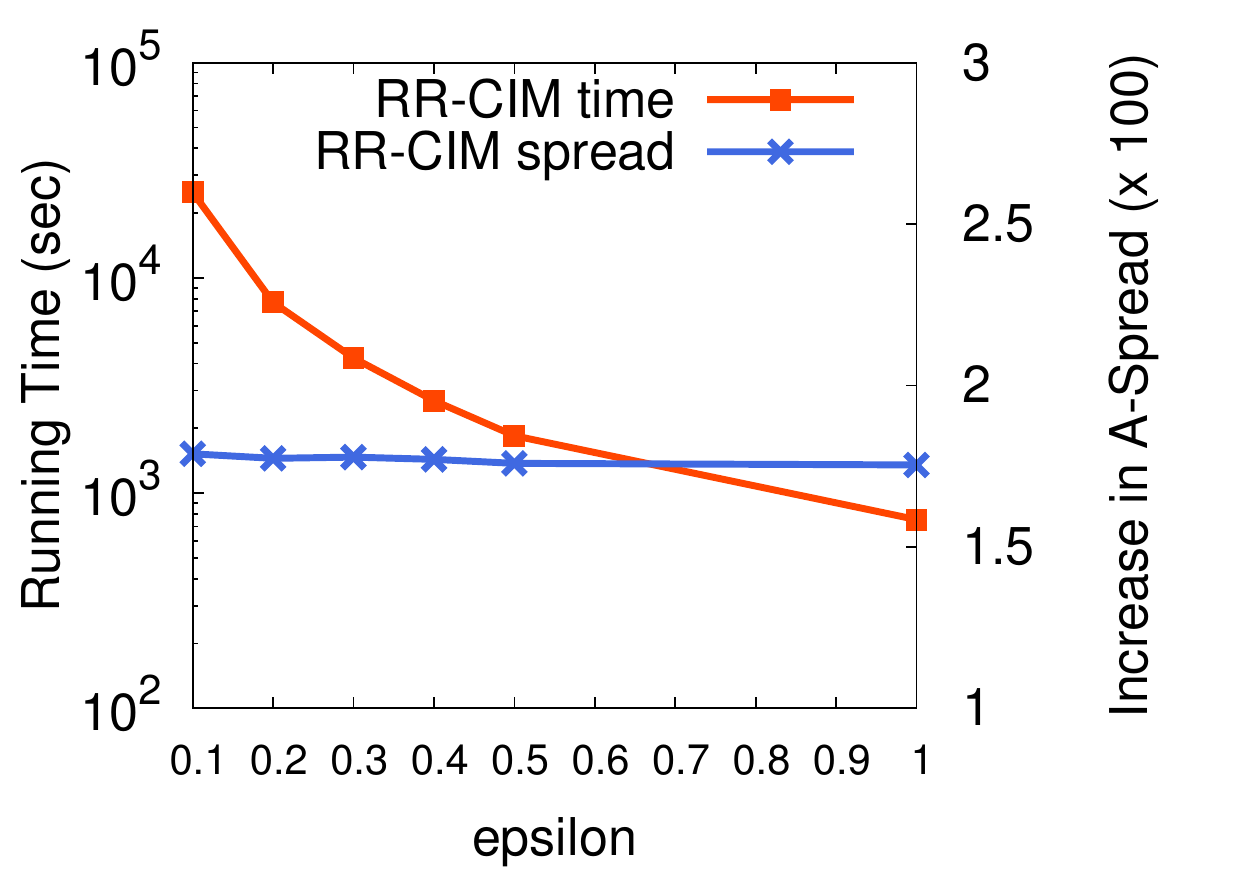}\\
(a) \SIM, \flix  & (b) \CIM, \flix   \\
\includegraphics[width=.4\textwidth]{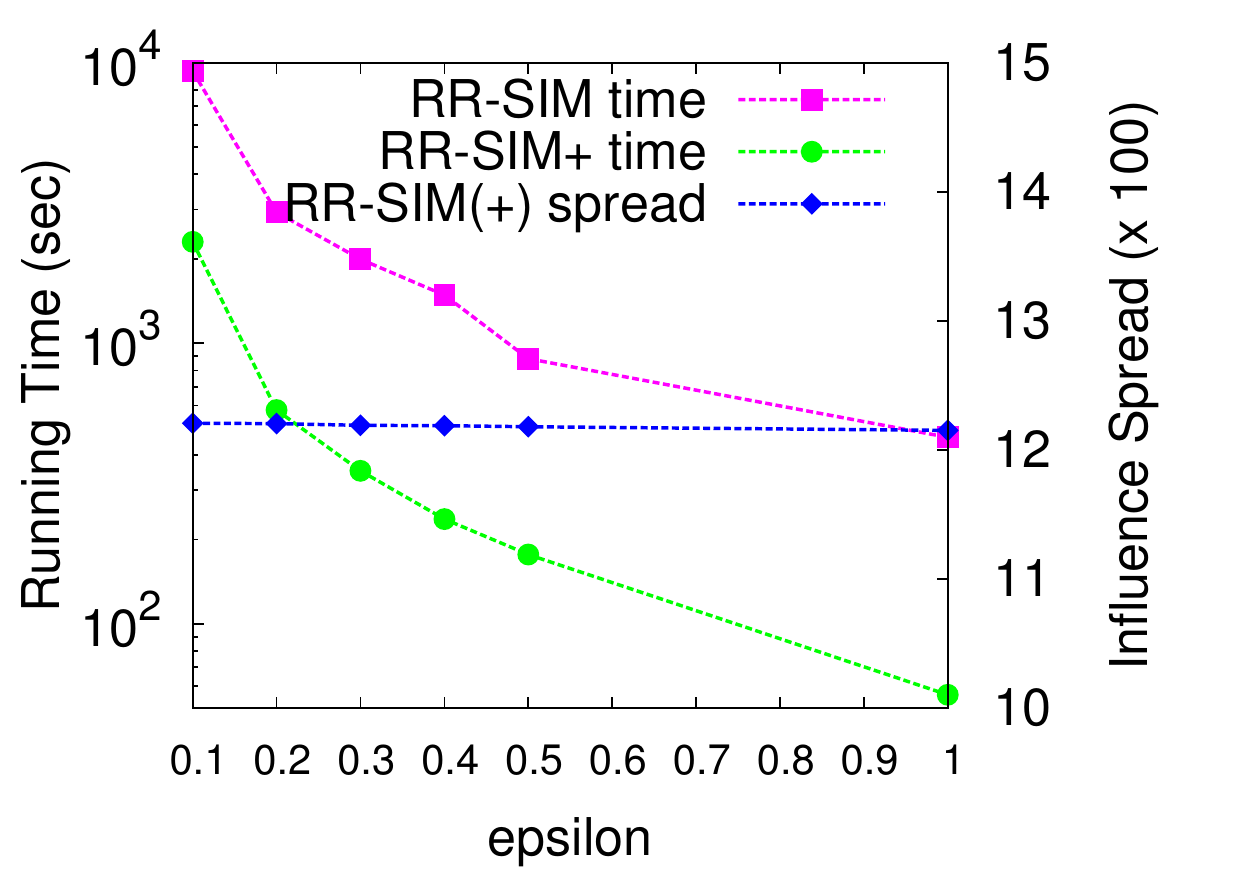}&
\includegraphics[width=.4\textwidth]{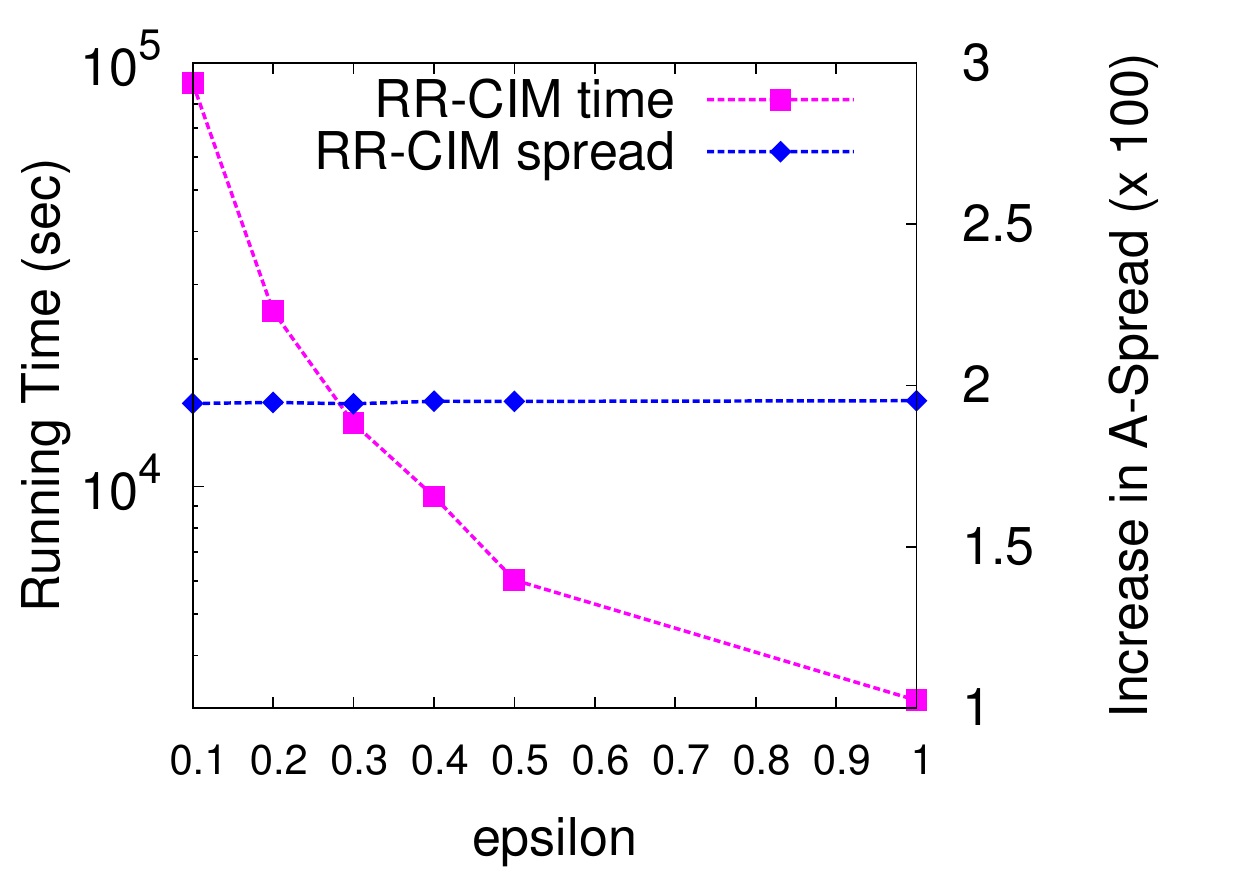}\\
  (c) \SIM, \dbBook & (d) \CIM, \dbBook
\end{tabular}
\caption{Effects of $\epsilon$ on RR-set algorithms (\flix and \dbBook)}
\label{fig:epsilon}
\end{figure*}

\begin{figure*}[h!t!]
\centering
\begin{tabular}{cc}
\includegraphics[width=.4\textwidth]{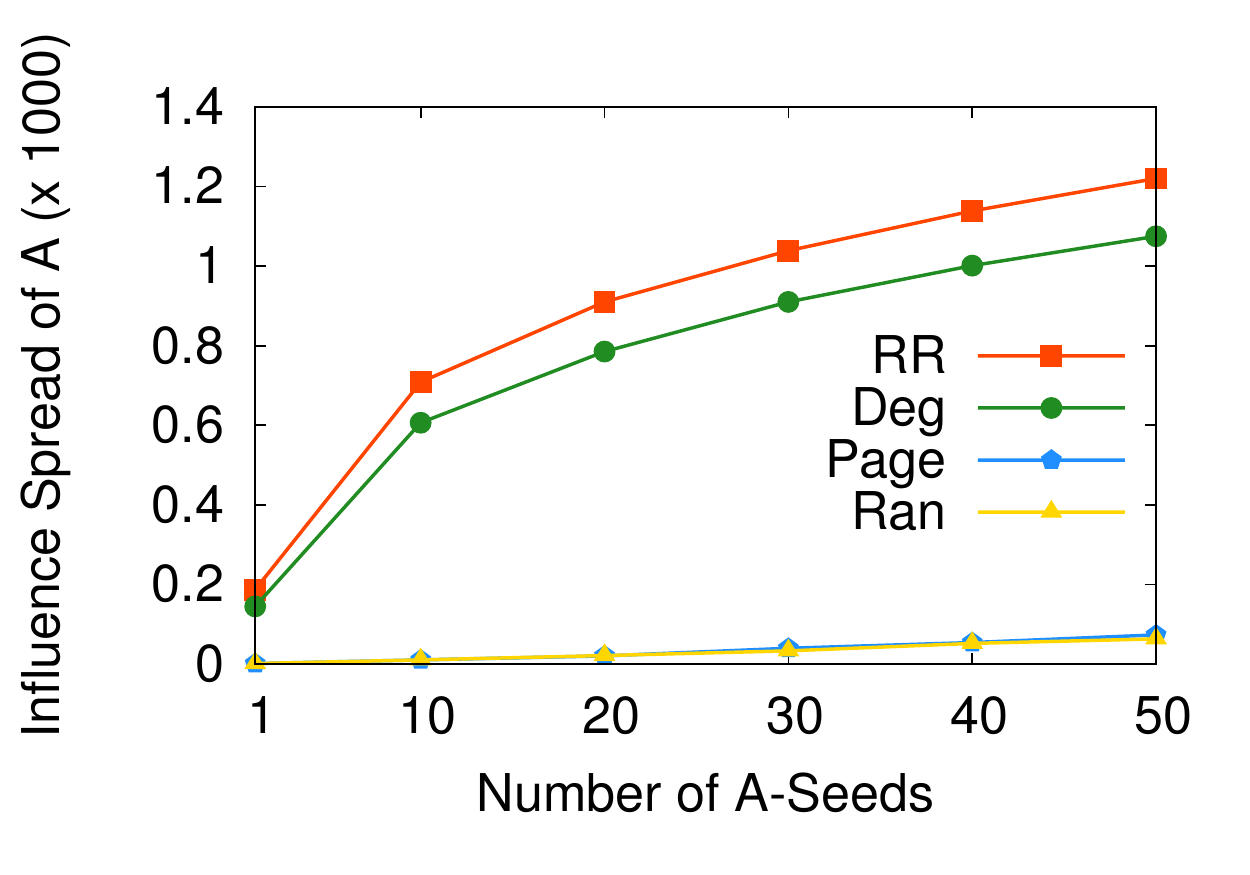} &
\includegraphics[width=.4\textwidth]{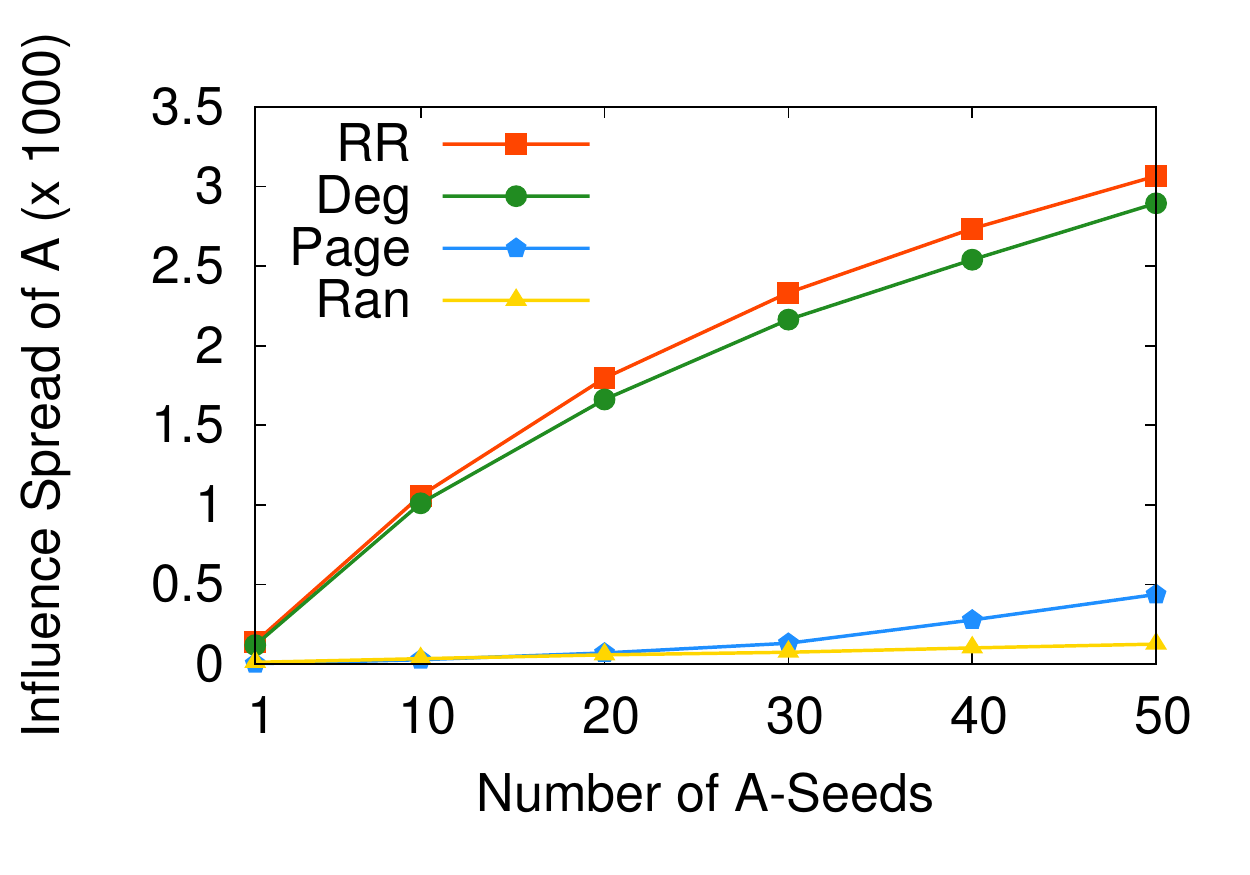} \\
(a) \dbBook & (b) \dbMovie \\ 
\includegraphics[width=.4\textwidth]{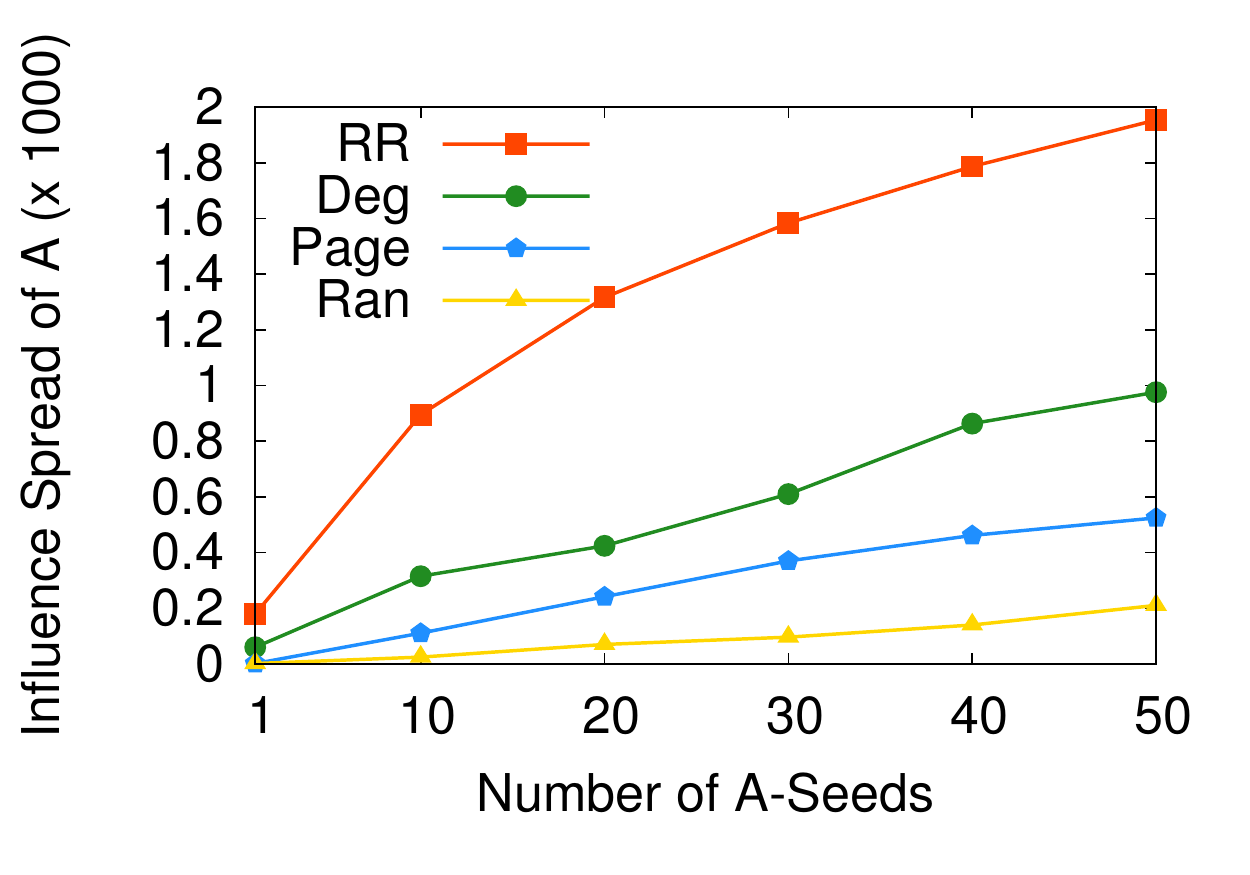} &
\includegraphics[width=.4\textwidth]{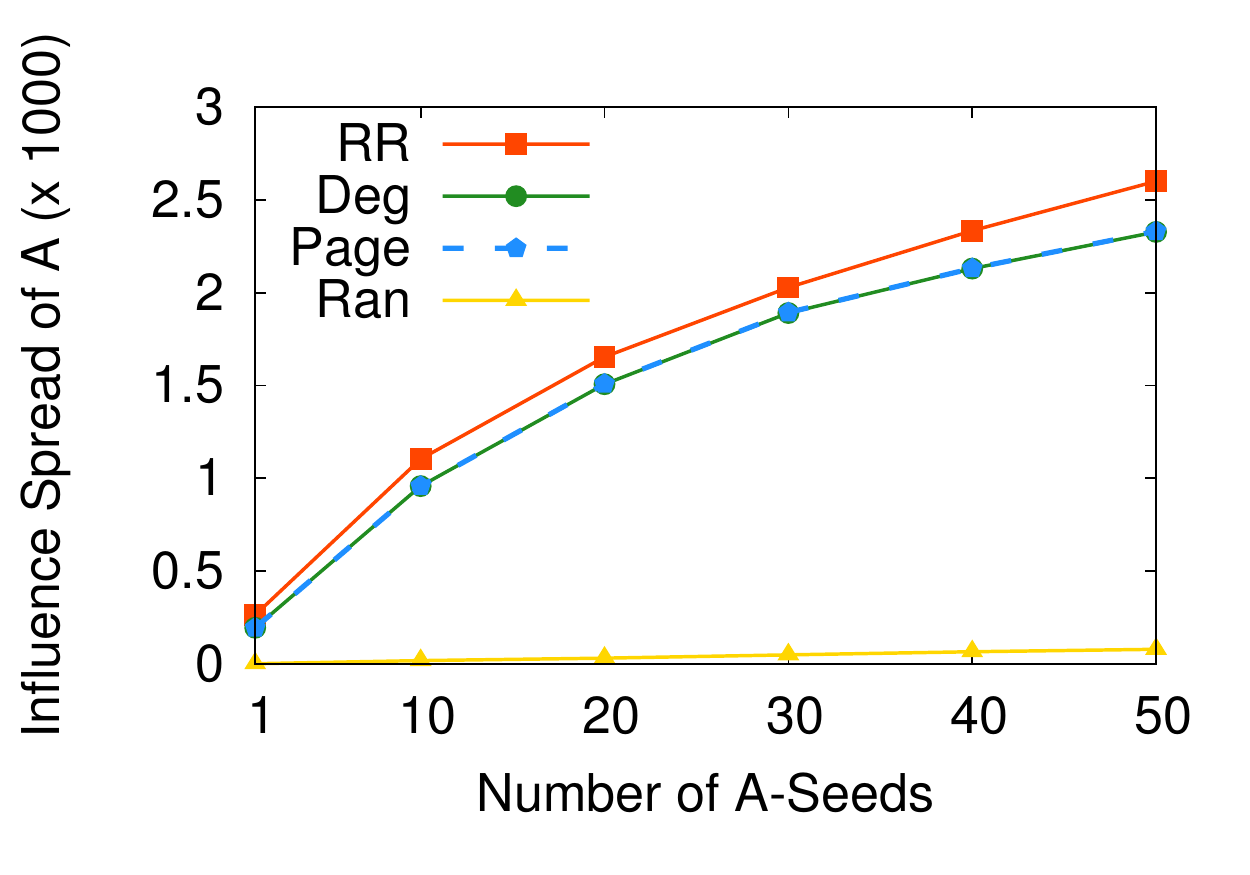} \\
(c) \flix &  (d)  \lastfm
\end{tabular}
\caption{$\calA$-Spread vs.\ $|S_\calA|$ for \SIM (RR -- \generalTIM with \timsimfast, Page -- PageRank; Deg -- High-Degree, Rand -- Random)} \label{fig:sim-spread}
\end{figure*}

\begin{figure*}[h!t!]
\centering
\begin{tabular}{cc}
 \includegraphics[width=.4\textwidth]{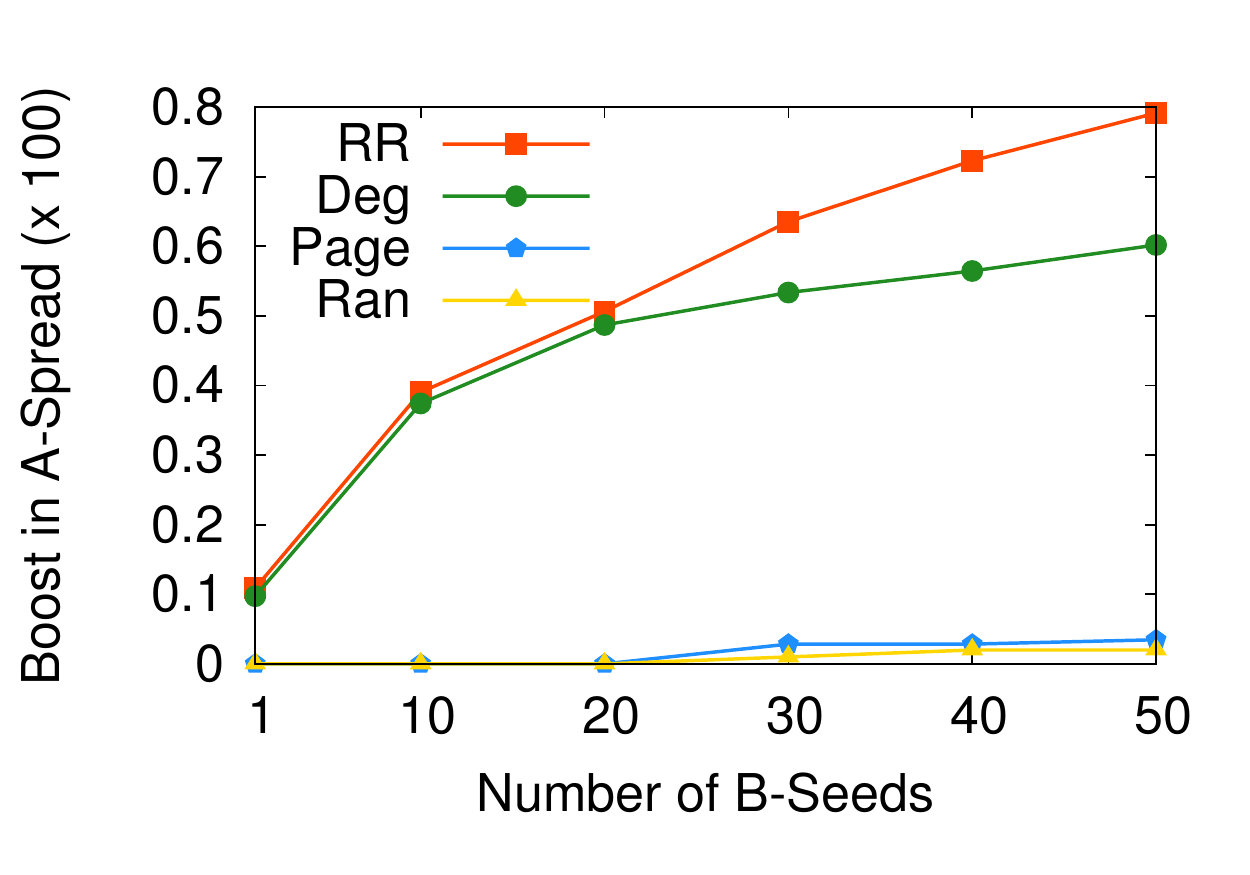} &
 \includegraphics[width=.4\textwidth]{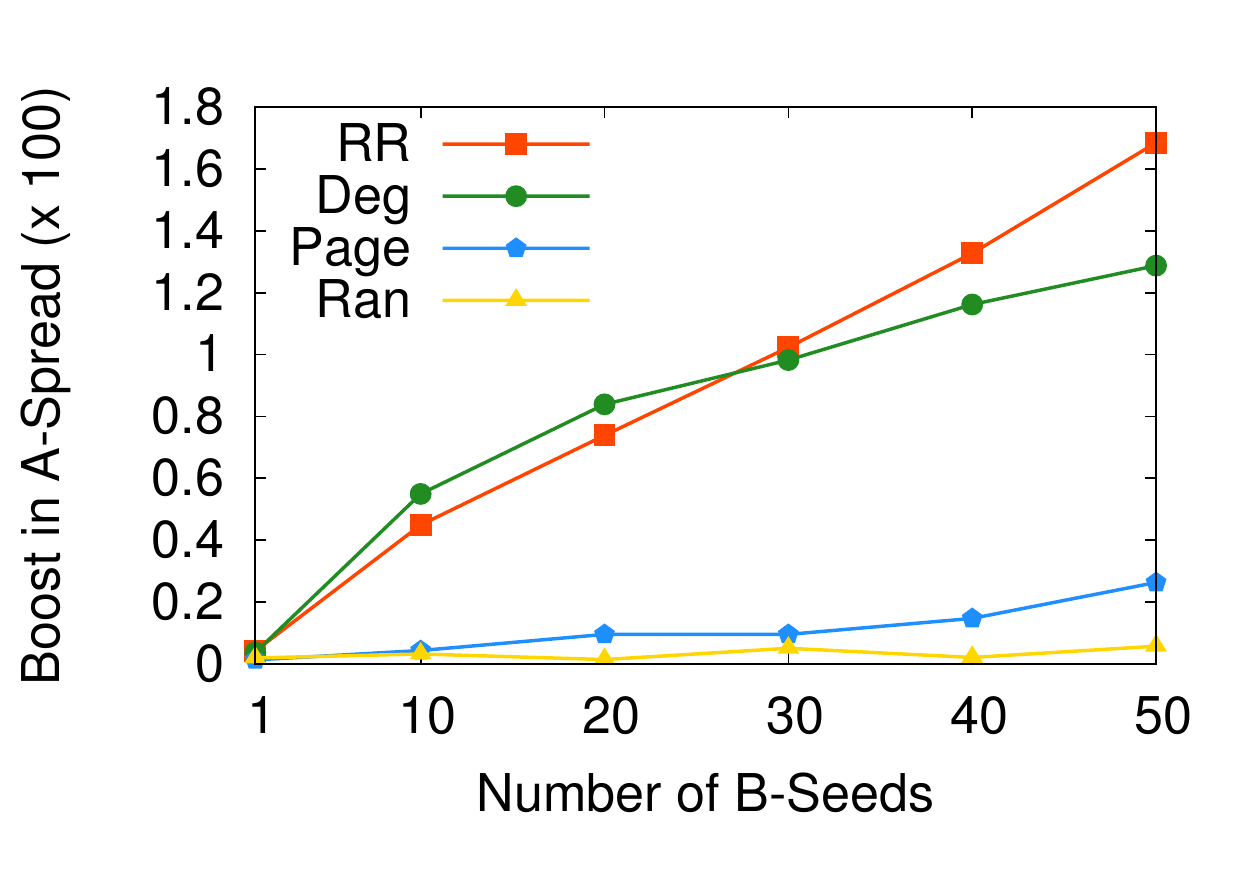} \\
 (a) \dbBook, $\sigma_\calA(S_\calA,\emptyset) = 612$ &  (b) \dbMovie, $\sigma_\calA(S_\calA,\emptyset) = 2357$ \\ 
 \includegraphics[width=.4\textwidth]{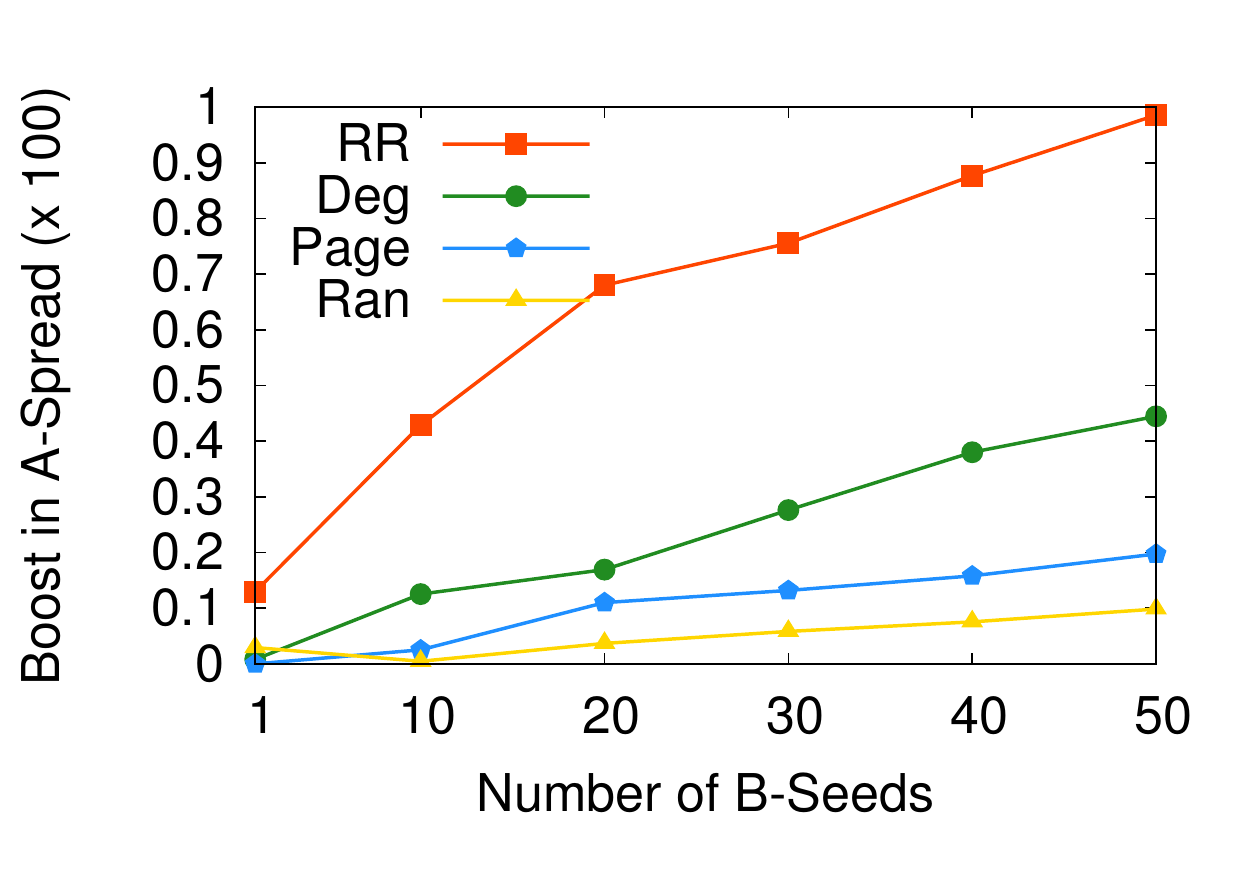} &
 \includegraphics[width=.4\textwidth]{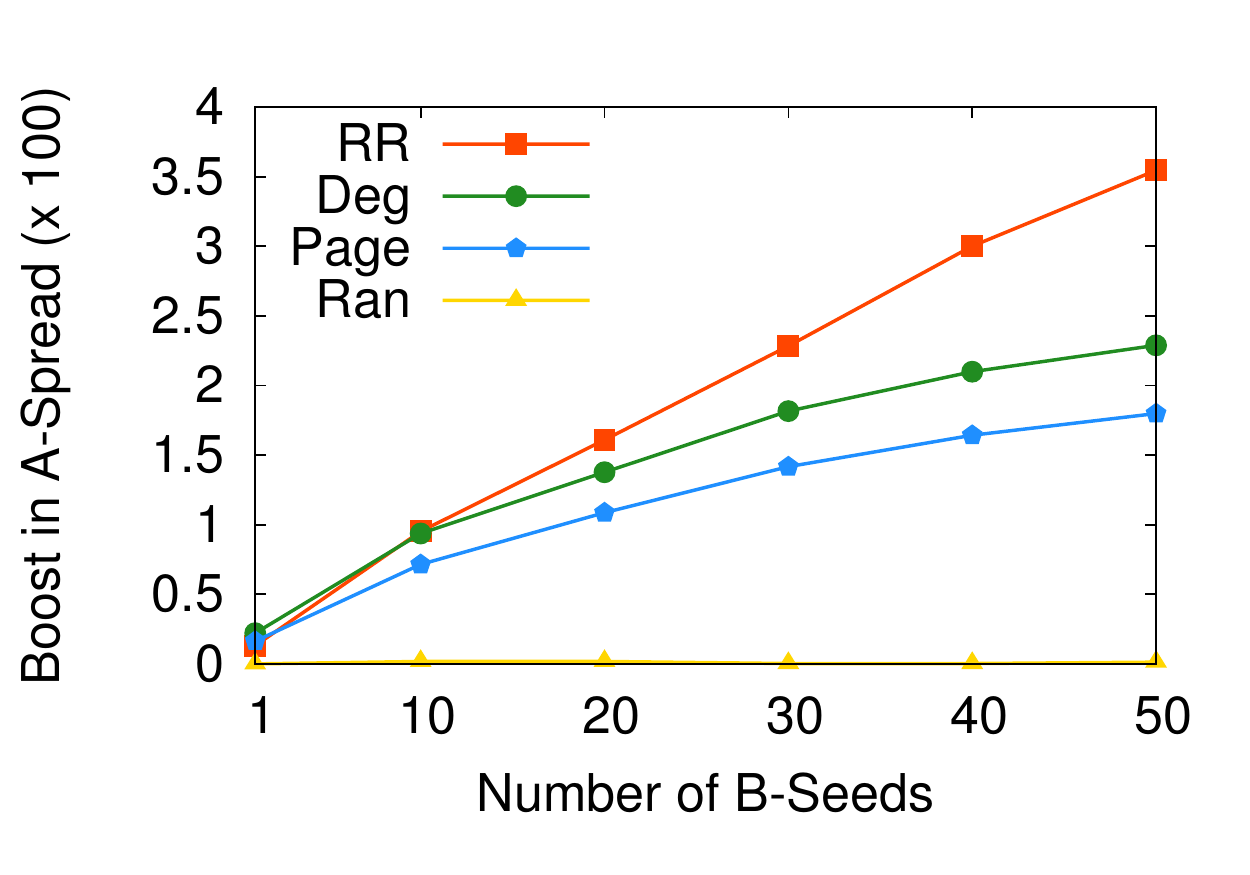} \\
  (c) \flix, $\sigma_\calA(S_\calA,\emptyset) = 1643$ &   (d)  \lastfm, $\sigma_\calA(S_\calA,\emptyset) = 1568$
\end{tabular}
\caption{Boost in $\calA$-Spread vs.\ $|S_\calB|$ for \CIM (RR -- \generalTIM with \timcim)} 
\label{fig:cim-spread}
\end{figure*}

\subsubsection*{Effect of {$\epsilon$} in {\generalTIM}}
%\noindent\textbf{Results \& Analysis: Effects of $\epsilon$.}
%We first evaluate the effects of $\epsilon$ on  \generalTIM.
As mentioned in \textsection\ref{sec:rrset}, $\epsilon$
	controls the trade-off between approximation ratio and efficiency.
Figure~\ref{fig:epsilon} plots influence spread and running time (log-scale)
	side-by-side, as a function of $\epsilon$, on \flix and \dbBook for
	both problems.
The results on other datasets are very similar and thus omitted.
We can see that as $\epsilon$ goes up from $0.1$ to $0.5$
	and $1$ (in fact, $\epsilon = 1 > (1-1/e)$ means theoretical approximation guarantees are lost),
	the running time of all versions of \generalTIM (\timsim, \timsimfast, \timcim)
	{\sl decreases dramatically}, by orders of magnitude.
	while, in practice, influence spread (\SIM) and boost (\CIM) are {\sl almost
	completely unaffected} (the largest difference among all test cases is only $0.45\%$).
%This allows us to set a reasonably large $\epsilon$: unless otherwise stated, we use 
%	$\epsilon=0.5$ for \timsim and \timsimfast, and $\epsilon = 1$ for \timcim.

\subsubsection*{Quality of Seeds}
The quality of seeds is measured by the influence spread or boost achieved.
%Since the ground-truth is \#P-hard to compute~\cite{ChenWW10}, 
We evaluate the spread of seed sets computed by all algorithms by MC simulations with 10K iterations for fair comparison.
As can be seen from Figures~\ref{fig:sim-spread} and \ref{fig:cim-spread}, our RR-set algorithms are consistently
	the best in almost all test cases, often leading by a significant margin.
The results of \greedy are omitted, since the spread it achieves is almost identical to \generalTIM, matching the observations in prior work \cite{tang14}.
\timsim results are identical to \timsimfast, and thus also omitted.

For \SIM, \generalTIM with \timsimfast is $13\%$, $2.7\%$, $100\%$, and $13\%$
	better than the next best algorithm on \dbBook, \dbMovie, \flix and \lastfm respectively, while for \CIM, \generalTIM with \timcim is $31\%$, $31\%$, $122\%$, and $51\%$ better. % (same order of datasets).
The boost in $\calA$-spread provided by $\calB$-seeds (\generalTIM with \timcim) is at least $6\%$ to $15\%$ of the original $\calA$-spread by $S_\calA$ only.
\highdeg performs well, especially in graphs with many nodes having large out-degrees (\dbMovie, \lastfm), while \pagerank produces good quality seeds only on \lastfm.
\rand is consistently the worst.
The performances of baselines are generally consistent with observations in prior works~\cite{ChenWW10, ChenWW10b,tang14} albeit for different diffusion models.

\begin{figure*}
\centering
\begin{tabular}{cc}
\includegraphics[width=.35\textwidth]{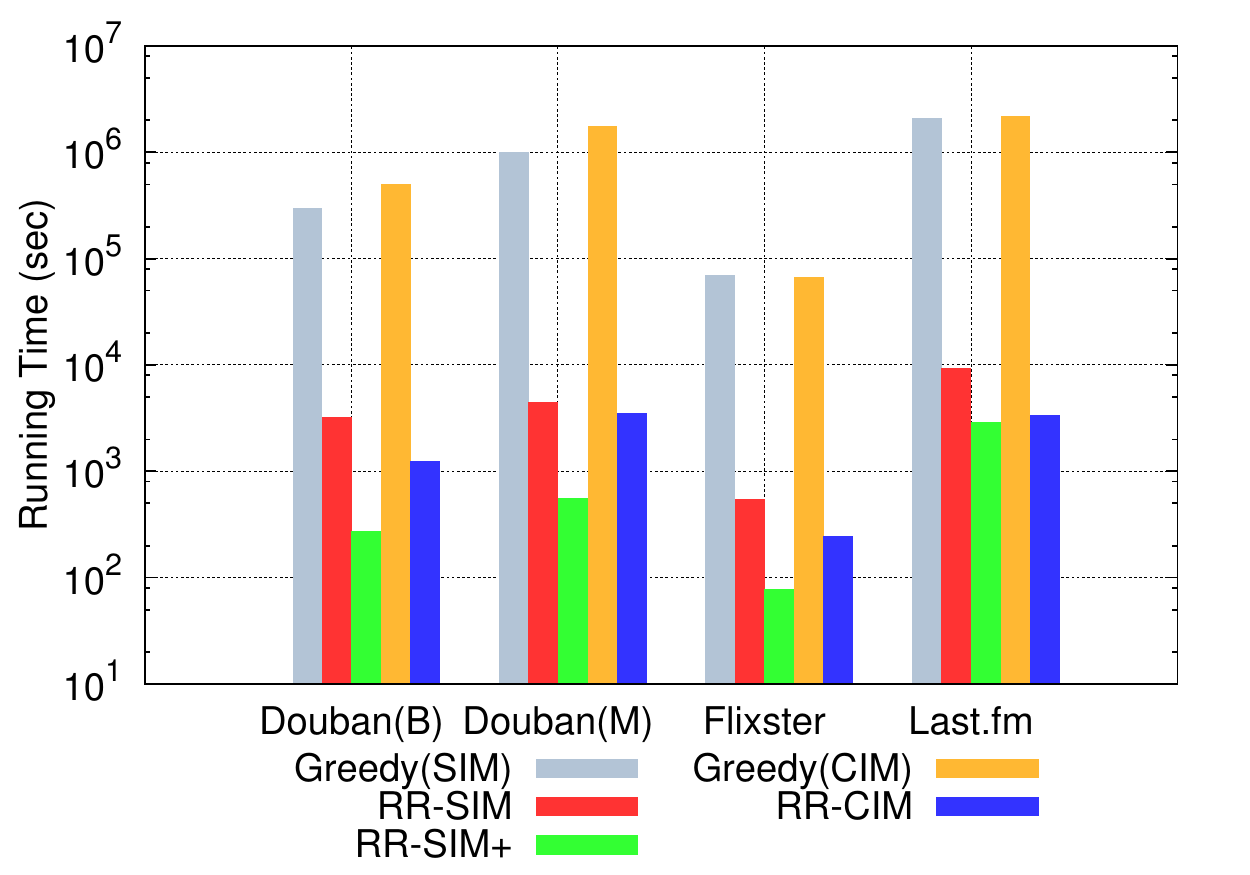}&
\includegraphics[width=.35\textwidth]{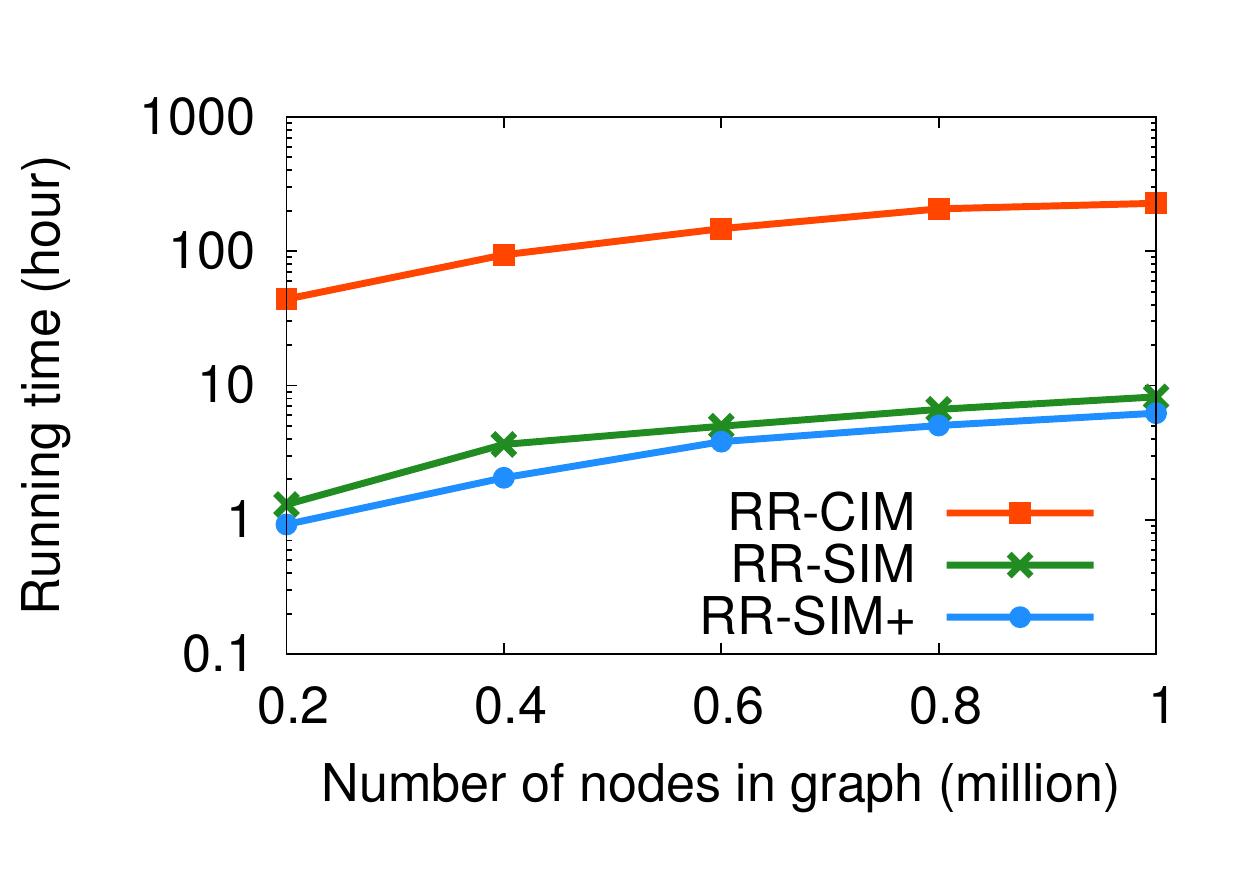}\\
%\hspace{0mm}\includegraphics[width=.285\textwidth]{graphics/sandwich}\\
 (a) Real networks  & (b) Synthetic graphs 
\end{tabular}
\caption{Running time}
\label{fig:time}
\end{figure*}

\begin{table*}
%\scriptsize
\centering
\begin{tabular}{lcccc}
						&\dbBook 	&\dbMovie	& \flix 	& \lastfm \\ \hline
$\mathsf{SIM_{learn}}$	& $0.996$	& $0.999$	& $0.996$	& $0.999$  \\ 
$\mathsf{SIM_{0.1}}$	& $0.652$	& $0.962$	& $0.492$ 	& $0.519$   \\ 
$\mathsf{SIM_{0.5}}$	& $0.770$ 	& $0.969$ 	& $0.633$	& $0.628$  \\ 
$\mathsf{SIM_{0.9}}$	& $0.946$	& $0.985$ 	& $0.926$   & $0.879$   \\ \hline

$\mathsf{CIM_{learn}}$	& $0.973$	& $0.918$	& $0.950$ 	& $0.825$   \\ 
$\mathsf{CIM_{0.1}}$	& $0.913$	& $0.832$	& $0.933$ 	& $0.772$ \\ 
$\mathsf{CIM_{0.5}}$	& $0.936$ 	& $0.885$	& $0.969$	& $0.857$ \\ 
$\mathsf{CIM_{0.9}}$	& $0.956$  	& $0.976$	& $0.993$ 	& $0.959$   \\ \hline
\end{tabular}
\caption{Sandwich approximation: $\sigma(S_\nu) / \nu(S_\nu)$}
\label{tab:SA-factor}
\end{table*}

\subsubsection*{Running Time and Scalability} 
We compare the running time of \generalTIM to \greedy, shown in Figure~\ref{fig:time}(a).
%Note that the top of $Y$-axis is set to be 604,800 seconds (one week), and touching it means
%	the algorithm runs beyond this time limit.
As can be seen, for \SIM, \generalTIM with \timsim, \timsimfast is about {\sl two to three orders of magnitude} faster than \greedy; % (shown as {\sf MC-SIM});
for \CIM, \generalTIM with \timcim is also about {\sl two orders of magnitude} faster than \greedy. % (shown as {\sf MC-CIM}).
%In addition, we observe that \timsimfast is 2.6, 5.5, and 2.2 times as fast as \timsim on \flix, \dbBook, and \lastfm respectively, while on \dbMovie, \timsim is 3.4 times faster.
%This is not surprising: As mentioned in \textsection\ref{sec:rrset}, though the adjustment made in \timsimfast
%	is intuitive and confirmed effective in three of the datasets, there is no theoretical guarantee
%	that \timsimfast is always faster than \timsim.
In addition, we observe that \timsimfast is 12, 8, 7, and 2 times as fast as \timsim on \dbBook, \dbMovie, \flix, and \lastfm respectively.
The running time of \highdeg, \pagerank, and \rand baselines are omitted since they are typically very efficient~\cite{ChenWW10, ChenWW10b, infbook}.

We then use larger synthetic networks to test the scalability of \generalTIM with our RR-set generation algorithms.
We generate power-law random graphs of 0.2, 0.4, ..., up to 1 million nodes with a power-law degree exponent of 2.16~\cite{ChenWW10}.
These graphs have an average degree of about 5.
%Settings and parameters are set per \textsection\ref{sec:exp-setting}, and 
We use the GPAs from \flix. %\weic{I do not understand this sentence. What does it mean we use $\bQ^+$?}
%The difference between \timsimfast and \timsim is minimal so we only show \timsimfast.
%We found that typically high-degree nodes in power-law graphs, chosen to be $\calB$-seeds, have much larger degrees compared to real graphs, and hence $\calB$-seeds will reach a significant portion of the graph, making the adjustment of \timsimfast less effective than on real graphs.
We can see that \generalTIM with \timsimfast within 6.2 hours for the 1-million node graph, and its running time grows linearly in graph size, which indicates great scalability.
\timcim is slower %than \timsimfast 
due to the inherent intricacy of \CIM, but it also scales linearly.
To put its running time measures in perspective, \greedy --- the only other known approximation algorithm for \CIM --- takes about 48 hours on \flix (12.9K nodes), while \generalTIM with \timcim is 4 hours faster on a graph 10 times as large.

\subsubsection*{Approximation Factors by Sandwich Approximation}
Recall from \textsection\ref{sec:sandwich} that the
	approximation factor yielded by SA is data-dependent:
$\sigma(S_{\mathit{sand}}) \ge \max \{\frac{\sigma(S_\nu)}{\nu(S_\nu)}, \frac{\mu(S_\sigma^*)}{\sigma(S_\sigma^*)} \} \cdot (1-1/e-\epsilon) \cdot  \sigma(S_\sigma^*).$
To see how good the SA approximation factor is in real-world graphs, we compute ${\sigma(S_\nu)}/{\nu(S_\nu)}$, as\InFullOnly{ although 
$S_\sigma^*$ is unknown due to NP-hardness (Theorem~\ref{thm:hard}),}
	SA is guaranteed to have an approximation factor
	of at least $(1-1/e-\epsilon) \cdot {\sigma(S_\nu)}/{\nu(S_\nu)}$.

In the GAPs learned from data,
	both $\qba-\qbo$ and $\qab-\qao$ are small and thus likely
	``friendly'' to SA, as we mentioned in \textsection\ref{sec:sandwich}.  
	Thus, we further ``stress test'' SA with more adversarial settings:
First, set $\qao = 0.3$ and $\qab = 0.8$;
Then, 
	for \SIM, fix $\qba = 1$ and vary $\qbo$ from $\{0.1, 0.5, 0.9\}$;
	for \CIM, fix $\qbo = 0.1$ and vary $\qba$ from $\{0.1, 0.5, 0.9\}$.
%\pink{
%We remind that to get the upper bound function $\nu(\cdot)$, 
%	for \SIM we increase $\qbo$ to $\qba$ and for \CIM we increase $\qba$ to $1$.}
%\weic{This sentence is optional, can be removed if we need space.}
%\InFullOnly{This ensures mutual complementarity.}

Table~\ref{tab:SA-factor} illustrates the results on all datasets with
	both learned GAPs and artificial GAPs.
We use shorthands $\mathsf{SIM}$ and $\mathsf{CIM}$ for \SIM and \CIM respectively.
Subscript $\mathsf{learn}$ means the GAPs are learned from data.
In stress-test cases (other six rows), e.g., for $\mathsf{SIM}$, subscript $\mathsf{0.5}$ means $\qbo = .5$,
	while for $\mathsf{CIM}$, it means $\qba = .5$.
As can be seen, with real GAPs, the ratio is extremely close to $1$, matching our intuition.
For artificial GAPs, the ratio is not as high, but most of them are still close to $1$.
E.g., in the case of  $\mathsf{SIM_{0.5}}$, $\sigma(S_\nu) / \nu(S_\nu)$
	ranges from $0.628$ (\lastfm) to $0.969$ (\dbMovie), which correspond to
	an approximation factor of $0.40$ and $0.61$ ($\epsilon$ omitted).
Even the smallest ratio ($0.492$ in $\mathsf{SIM_{0.1}}$, \flix) would still yield
	a decent factor at about $0.3$.
This shows that SA is fairly effective and robust for solving non-submodular cases
	of \SIM and \CIM.
%In the appendix, we provide further evidences to support our findings here.

%As discussed in \textsection\ref{sec:sandwich}, 
%SA is likely to be effective for \SIM and \CIM when the adoption probabilities are
%	close, which is indeed the case for the real world data we learn.
We also measure the effectiveness of SA by comparing the spread achieved
	by seed sets $S_\sigma, S_\nu, S_\mu$ obtained
	w.r.t.\ the original, upper bound, and lower bound functions respectively.
Such spread must be computed using the original function $\sigma$, and in our case,
	the unaltered GAPs.
More specifically, we calculate the relative error defined as follows (for \CIM, disregard $S_\mu$).
\[
\mathtt{SA\_error} =\frac{\max\{|\sigma(S_\sigma) - \sigma(S_\mu)|, |\sigma(S_\sigma) - \sigma(S_\nu)|\}} {\sigma(S_\sigma)}.
\]

\begin{figure}
\centering
\includegraphics[width=0.45\textwidth]{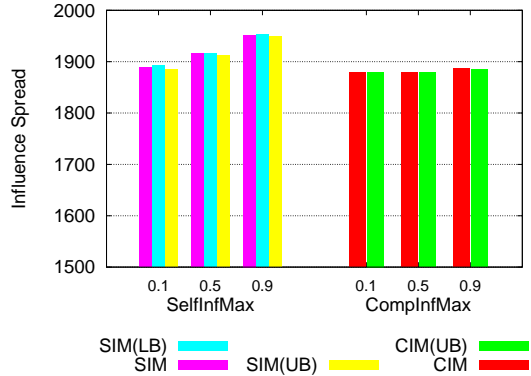}
\caption{Sandwich Approximation on \flix}
\label{fig:sandwich}
\end{figure}

In all four datasets, for both problems, the largest error is only $0.2\%$!
To see if this is due to that $\qbo$ and $\qba$ being close in the GAPs learned
	from action log data, we further ``stress test'' SA with a much more adversarial set-up.
For \SIM, vary $\qbo$ to be $\{.1, .5, .9\}$ and fix $\qba = .96$.
For \CIM, vary $\qba$ to be $\{.1, .5, .9\}$ and fix $\qbo = .1$ (to maintain complementarity).
No change is made to $\qao$ and $\qab$.

Figure~\ref{fig:sandwich} compares $\sigma(S_\sigma)$, $\sigma(S_\mu)$, and $\sigma(S_\nu)$
	on \flix.
As can be seen, even in this adversarial setting, SA is still highly effective for both \SIM and \CIM.
Amongst all test cases, the largest error is $0.4\%$.
The results on other datasets are very similar and hence omitted.
%We further stress test more by separating $\qao$ and $\qab$ to a wider margin, and still achieve
%	similar results (hence omitted).
%This means the SA strategy is highly robust and useful for our problems. 

\section{Conclusions \& Future Work}\label{sec:concl}
\balance

In this work, 
we propose the Comparative Independent Cascade (\comic) model that allows any degree of competition or complementarity between two different propagating items, and study the novel \SIM and \CIM problems for complementary products.
We %identify parameter subspaces of \comic under which \SIM and \CIM enjoy submodularity and monotonicity, and 
develop non-trivial extensions to the RR-set techniques to achieve approximation algorithms.
For non-submodular settings, we propose Sandwich Approximation to achieve data-dependent approximation factors.
Our experiments demonstrate the effectiveness and efficiency of proposed solutions.

%This work opens up a number of interesting avenues for future research.
For future work, one direction is to design more efficient algorithms or
	heuristics for \SIM and (especially) \CIM; 
	e.g., whether near-linear time algorithm is still available for these
	problems is still open.
Another direction is to fully characterize the entire GAP space $\bQ$ in terms of
	monotonicity and submodularity properties.
Moreover, an important direction is to 
	extend the model to multiple items.
Given the current framework, \model can be extended to accommodate $k$ items, 
	if we allow 
	$k\cdot 2^{k-1}$ GAP parameters --- for each item, we specify the probability of 
	adoption for every combination of other items that have been adopted.
However, how to simplify the model and make it tractable, how to reason about
	the complicated two-way or multi-way competition and complementarity, how
	to analyze  monotonicity and submodularity, and how to learn
	GAP parameters from real-world data all remain as
	interesting challenges.
Last, it is also interesting to consider an extended \comic model
	in which influence probabilities on edges are product-dependent.

\section*{Acknowledgments}
This research is supported in part by a Discovery grant and a Discovery Accelerator Supplements grant
	from the Natural Sciences and Engineering Research Council of Canada (NSERC).
We also thank Lewis Tzeng for some early discussions on modeling influence
propagations for partially competing and partially complementary items.

%\clearpage
\appendix
%\section{Remarks on Com-IC Model}
%
%\spara{Special Cases of Com-IC}
%We remark that \comic encompasses previously-studied single-entity and pure-competition models as special cases.
%When $\qao = \qbo = 1$ and $\qab = \qba = 0$, \model reduces to the (purely) competitive Influence Cascade model~\cite{infbook}.
%If, in addition, $\qbo$ is $0$, the model further reduces to the classic IC model.

%\note[Laks]{The appendix looks a bit messy. We should order the material in the same order as the section of the main paper it pertains to. E.g., proofs before additional experiments.} 

\clearpage

\section*{Appendix}

\section{Remarks on Com-IC Model}

\subsection{Unreachable States}
Recall that in the \comic model, before an influence diffusion starts,
	all nodes are in the initial joint state ($\calA$-idle, $\calB$-idle).
According to the diffusion dynamics defined in Figure~\ref{fig:model}, 
	there exist five unreachable joint states, {\sl which are not material
	to our analysis and problem-solving, since none of these
	is relevant to actual adoptions}, the objectives studied in \SIM and \CIM.
For completeness, we list these states here.
\begin{enumerate}
\item ($\calA$-idle, $\calB$-rejected)
\item ($\calA$-suspended, $\calB$-rejected)
\item ($\calA$-rejected, $\calB$-idle)
\item ($\calA$-rejected, $\calB$-suspended)
\item ($\calA$-rejected, $\calB$-rejected)
\end{enumerate}

\begin{lemma}
In any instance of the \comic model (no restriction on GAPs), no node can reach
	the state of ($\calA$-idle, $\calB$-rejected), from its initial state
	of ($\calA$-idle, $\calB$-idle).
\end{lemma}
\begin{proof}
Let $v$ be an arbitrary node from graph $G=(V,E,p)$.
Note that for $v$ to reject $\calB$, it must be first be informed of $\calB$ (otherwise it remains $\calB$-idle, regardless of its state w.r.t.\ $\calA$), and then becomes $\calB$-suspended (otherwise it will be $\calB$-adopted, a contradiction).
Now,  note that $v$ is never informed of $\calA$, and hence it will not be triggered  to reconsider $\calB$, the only route
	to the state of $\calB$-rejected, according to the model definition.
Thus, ($\calA$-idle, $\calB$-rejected) is unreachable.
\end{proof}

The argument for ($\calA$-rejected, $\calB$-idle) being unreachable is symmetric,
	and hence omitted.

\begin{lemma}\label{lemma:unreach}
In any instance of the \comic model (no restriction on GAPs), no node can reach
	the state of ($\calA$-suspended, $\calB$-rejected), from its initial state
	of ($\calA$-idle, $\calB$-idle).
\end{lemma}
\begin{proof}
Let $v$ be an arbitrary node from graph $G=(V,E,p)$.
Note that for $v$ to reject $\calB$, it must be first be informed of $\calB$ (otherwise it remains $\calB$-idle, regardless of its state w.r.t.\ $\calA$), and then becomes $\calB$-suspended (otherwise it will be $\calB$-adopted, a contradiction).
Now,  $v$ transits from $\calA$-idle to $\calA$-suspended, meaning that $v$ does not adopt $\calA$.
This will not further trigger reconsideration, and hence $v$ stays at $\calB$-suspended.
This completes the proof.
\end{proof}

The argument for ($\calA$-rejected, $\calA$-suspended) being unreachable is symmetric,
	and hence omitted.
Finally, it is evident from the proof of Lemma~\ref{lemma:unreach} that, the joint state
	of ($\calA$-suspended, $\calA$-suspended) is a {\em sunken} state, meaning the node
	will not get out it to adopt or reject any product.
This implies that  ($\calA$-rejected, $\calB$-rejected) is also unreachable.

\subsection{Counter-Examples for Submodularity and Monotonicity}

The first two counter-examples show that self-monotonicity and
	cross-monotonicity may not hold in general for the \comic model when there is
	no restriction on GAPs.

\begin{figure}[h!]
 \centering
   \includegraphics[width=0.45\textwidth]{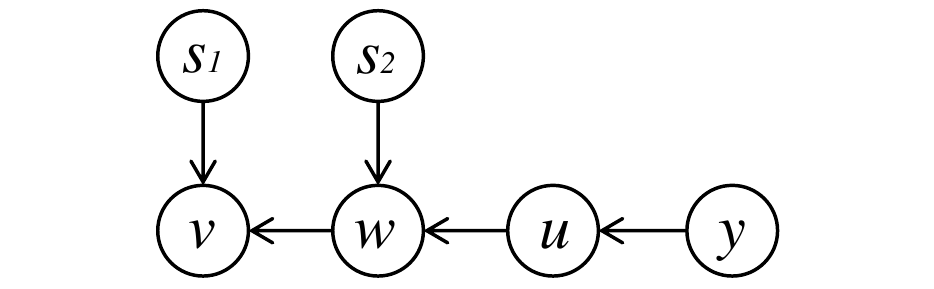}
 \caption{The graph for Example \ref{exp:non-mono}}
\label{fig:counter-mono}
\end{figure}

\begin{example}[Non-Self-Monotonicity] 
\label{exp:non-mono}
{\em
Consider Figure \ref{fig:counter-mono}.
All edges have probability $1$.
GAPs are $\qao = q \in (0,1)$, $\qab = \qbo = 1$, $\qba = 0$, which means that
	$\calA$ competes with $\calB$ but $\calB$ complements $\calA$.
Let $S_\calB = \{y\}$.
If $S_\calA$ is $S = \{s_1\}$, the probability that $v$ becomes $\calA$-adopted is $1$,
	because $v$ is informed of $\calA$ from $s_1$, and even if it does not adopt
	$\calA$ at the time, later it will surely adopt $\calB$ propagated from $y$,
	and then $v$ will reconsider $\calA$ and adopt $\calA$.
If it is $T = \{s_1, s_2\}$, that probability is $1-q+q^2 < 1$: $w$ gets $\calA$-adopted w.p. $q$ blocking $\calB$ and then $v$ gets $\calA$-adopted w.p. $q$; $w$ gets $\calB$-adopted w.p. $(1-q)$ and then $v$ surely gets $\calA$-adopted. 
Replicating sufficiently many $v$'s, all connected to $s_1$ and $w$, will lead to
	$\sigma_\calA(T,S_\calB) < \sigma_\calA(S,S_B)$. %, hence non-monotone. 
The intuition is that the additional $\calA$-seed $s_2$ ``blocks'' $\calB$-propagation
	as $\calA$ competes with $\calB$ ($\qba < \qbo$) but $\calB$ complements
	$\calA$ ($\qab > \qao$).
Clearly $\sigma_\calA$ is not monotonically decreasing in $S_\calA$ either (e.g., in a graph when all nodes are isolated).
Hence, $\sigma_\calA$ is not monotone in $S_\calA$.
\qed
}
\end{example}

\begin{figure}[h!]
 \centering
   \includegraphics[width=0.45\textwidth]{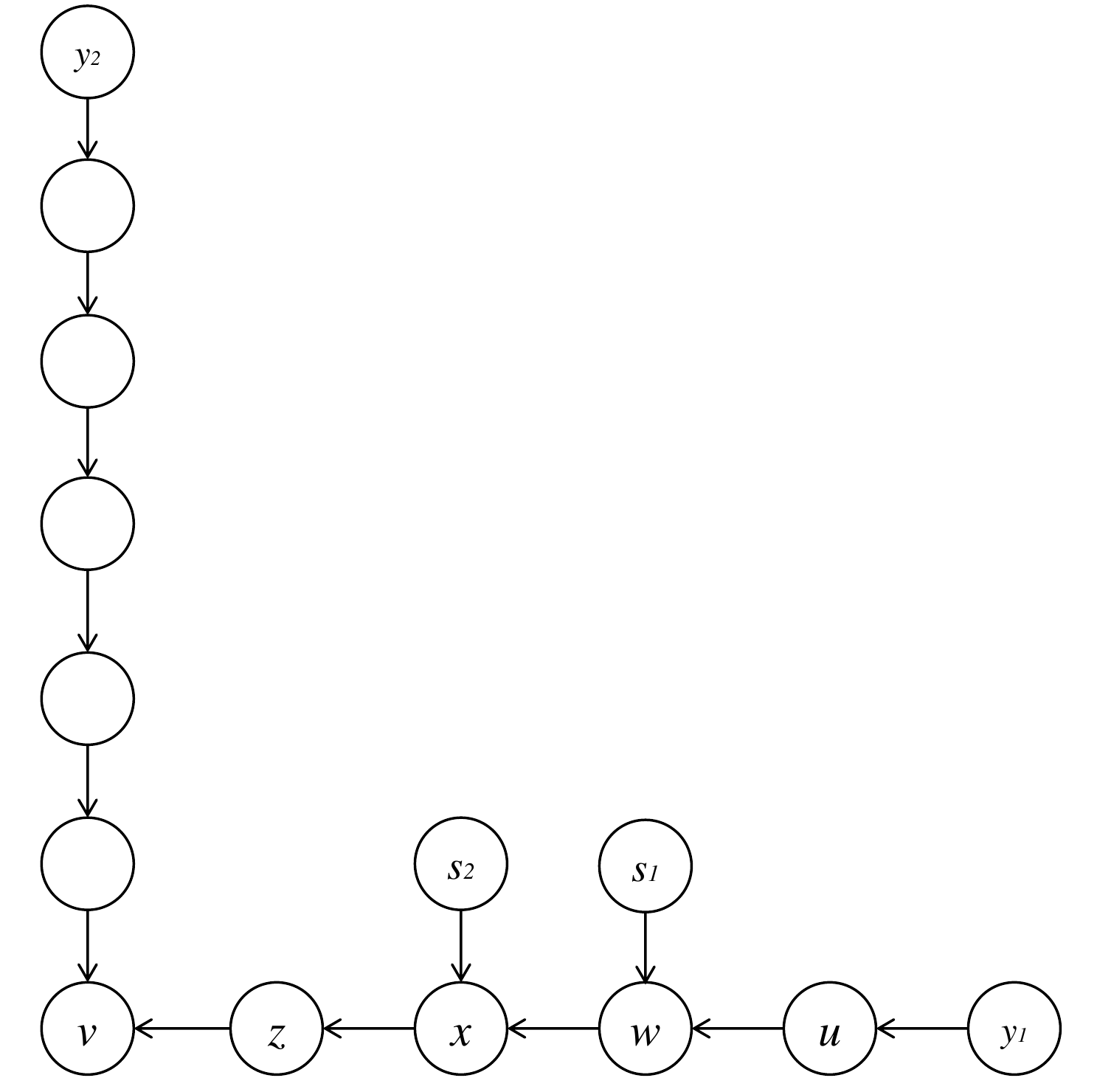}
 \caption{The graph for Example \ref{exp:non-cross-mono}}
\label{fig:counter-cross-mono}
\end{figure}

\begin{example}[Non-Cross-Monotonicity] 
\label{exp:non-cross-mono}
{\em
We use the example shown in Figure \ref{fig:counter-cross-mono}.
All edges have probability $1$.
Nodes $y_1$ and $y_2$ are the two fixed $\calA$ seeds, and we want to grow
	$\calB$ seed set from $S=\{s_1\}$ to $T=\{s_1, s_2\}$.
GAPs satisfy $0 < \qab < \qao < 1$ and $0 < \qbo < \qba < 1$, 
	which means that
	$\calA$ complements $\calB$ but $\calB$ competes with $\calA$.
Since $\calB$ competes with $\calA$, it is straightforward to have examples in which
	the $\sigma_\calA$ decreases when $\calB$ seed set grows.
We use Figure~\ref{fig:counter-cross-mono} to show a possible world
	in which the growth of $\calB$ seed sets leads to $v$ adopting $\calA$,
	indicating that $\calA$ spread may also increase.
Even though we do not have a direct example showing that 
	$\sigma_\calA$ increases when $\calB$ seed set grows, 
	we believe the possible world example is a good indication that
	$\sigma_\calA$ is not cross-monotone in $S_\calB$.

Figure~\ref{fig:counter-cross-mono} uses Figure~\ref{fig:counter-mono}
	as a gadget.
Intuitively, when $\calB$ seed set grows from $S$ to $T$, 
	the probability that $z$ adopts $B$ decreases as shown in Example~\ref{exp:non-mono}.
Then we utilize this and the fact that $\calB$ competes with $\calA$ 
	to show that when $S$ is the $\calB$ seed set, due to $\calB$'s competition from
	$x$ node $v$ will not adopt $\calA$, but when $T$ is the $\calB$ seed set,
	there is no longer $\calB$'s competition from $x$ and thus $v$ will adopt $\calA$.

The node thresholds in the possible world are as follows:
$\qab < \alpha^w_\calA \le \qao$, $\qbo < \alpha^x_\calB \le \qba$, $\qao < \alpha^z_\calA$,
	$\qab < \alpha^v_\calA \le \qao$, and all other non-specified $\alpha$ values
	take value $0$, meaning that they will not block diffusion.

Consider first that $S$ is the $\calB$ seed set.
Since $\qbo < \alpha^x_\calB$, $x$ is informed about $\calB$
	from $s_1$ but will not adopt $\calB$ directly from $s_1$.
From $y_1$, we can see that $\calA$ will pass through $u$, $w$ and reaches $x$.
After $x$ adopts $\calA$, it reconsiders $\calB$, and since $\alpha^x_\calB \le \qba$,
	$x$ adopts $\calB$.
Node $x$ then informs $z$ about $\calA$ and $\calB$, in this order.
However, since $\qao < \alpha^z_\calA$, $z$ does not adopt $\calA$, but it adopts
	$\calB$ ($\alpha^z_\calB = 0 \le \qbo$).
Next $z$ informs $v$ about $\calB$, which adopts $\calB$ since $\alpha^v_\calB=0\le \qbo$.
This happens one step earlier than $\calA$ reaches $v$ from $y_2$, but since
	$\qab < \alpha^v_\calA$, $v$ will not adopt $\calA$.

Now consider that $T$ is the $\calB$ seed set.
In this case, $w$ definitely adopts $\calB$ from $s_2$.
Since $\qab < \alpha^w_\calA$ and $w$ adopts $\calB$ first, $w$ will not adopt
	$\calA$ from $y_1$.
Because $\qbo < \alpha^x_\calB$ and $w$ blocks $\calA$ from reaching
	$x$, $x$ will not adopt $\calA$ or $\calB$.
Then $z$ will not adopt $\calA$ or $\calB$ either.
This allows $\calA$ to reach $v$ from $y_2$, and since $\qab < \alpha^v_\calA \le \qao$,
	$v$ adopts $\calA$.

We can certainly duplicate $v$ enough times so that when $\calB$ seed set grows from
	$S$ to $T$, the $\calA$-spread in this possible world also increases, even though
	$\calB$ is competing with $\calA$. \qed

}
\end{example}

\begin{figure}[h!]
 \centering
   \includegraphics[width=0.45\textwidth]{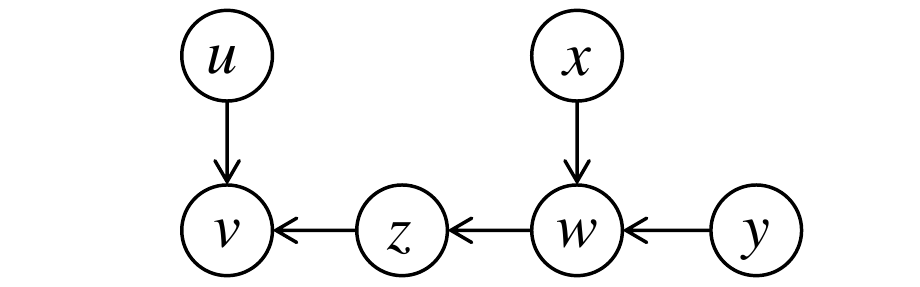}
 \caption{The graph for Example \ref{exp:complement1} and \ref{exp:complement2}}
\label{fig:counter-sub}
\end{figure}

The following two examples show that even when two items are mutually complementary,
	self-submodularity and cross-submodularity in general may not hold.

\begin{example}[Non-Self-Submodularity]
\label{exp:complement1}
{\em
Consider the possible world in Figure~\ref{fig:counter-sub}.
All edges are live.
The node thresholds are:
for $w$: $\alpha_\calA^w \leq \qao$, $\qbo < \alpha_\calB^w \leq \qba$;
for $z$: $\alpha_\calA^z > \qab$, $\alpha_\calB^z < \qbo$;
for $v$: $\qao < \alpha_\calA^v \leq \qab$, $\alpha_\calB^v \leq \qbo$;
Then fix $S_\calB = \{y\}$.
For $S_\calA$, let $S = \emptyset$, $T = \{x\}$, and $u$ is the additional seed.
It can be verified that only when $S_\calA = T \cup \{u\}$, $v$ becomes
	$\calA$-adopted, violating self-submodularity.

A concrete example of $\bQ$ for which submodularity does not hold is as follows: 
	$\qao = 0.078432$; $\qab = 0.24392$; $\qbo = 0.37556$; $\qba = 0.99545$.
Seed sets are the same as above.
We denote by $p_v(S_\calA)$ the probability that $v$ becomes $\calA$-adopted with A-seed set $S_\calA$.
It can be verified that:
$p_v(S) = 0$, $p_v(S \cup \{u\}) = 8.898 \cdot 10^{-5}$, $p_v(T) = 0.027254$, and $p_v(T \cup \{u\}) = 0.027383$.
Clearly, $p_v(T \cup \{u\}) - p_v(T) > p_v(S \cup \{u\}) - p_v(S)$.
Hence, replicating $v$ sufficiently many times will lead to
	$\sigma_\calA(T\cup\{u\}, S_\calB) - \sigma_\calA(T, S_\calB) \geq \sigma_\calA(S\cup\{u\} ,  S_\calB)- \sigma_\calA(S,  S_\calB)$, violating self-submodularity.
\qed
}
\end{example}

\begin{example}[Non-Cross-Submodularity]
\label{exp:complement2}
{\em
Consider the possible world in Figure~\ref{fig:counter-sub}.
The node thresholds are:
For $w$: $\qao < \alpha_\calA^w \leq \qab$, $\alpha_\calB^w \leq \qbo$;
for $z$: $\alpha_\calA^z \leq \qao$, $\alpha_\calB^z > \qba$;
for $v$: $\qao < \alpha_\calA^v \leq \qab$, $\alpha_\calB^v \leq \qbo$.
Fix $S_\calA = \{y\}$.
For $S_\calB$, let $S = \emptyset$, $T = \{x\}$, and $u$ is the additional seed.
It can be verified that only when $S_\calB = T \cup \{u\}$, $v$ becomes $\calA$-adopted,
	violating % both cross-monotonicity and 
	cross-submodularity. 
\qed
}
\end{example}

The above example applies even when $\qba = \qbo < 1$.
The key is that $\qba < 1$ and $\alpha_\calB^z > \qba$, which prevents $\calB$ to
	pass through $z$, and thus another $\calB$-seed $u$ is needed to inform $v$ of $\calB$.

%\section{Pseudo-code of {\large \timsimfast}}

%Here we give the pseudo-code to \timsimfast that is omitted from \textsection\ref{sec:rrset}.
%This is one of the RR-set generation
%	procedure we proposed for \SIM in \textsection\ref{sec:rrset}.

%\section{Additional Experimental Results}

%\subsubsection*{\large \generalTIM vs.\ \vanillaIC and \copying}

%We now report additional results for the experiments done with 
%	synthetic GAPs in \textsection\ref{sec:synGAPexp}.

%\subsubsection*{GAPs learned from Real Social Networks}

%\subsubsection*{Sandwich Approximation (SA)}

\section{Proofs and Additional Theoretical Results}\label{sec:proof}

\subsection{Proofs for Results in Section~\ref{sec:problem}}

\textbf{Theorem~\ref{thm:cimOpt}} (re-stated).
{\em \theoremCIMOpt}

%The proof will rely on the possible world model as defined in \textsection\ref{sec:submod}.

%\note[Laks]{I find that the overview plot discussion above is not tightly connected with what is presented in the formal proof below. Specifically, the connection to the equivalence classes idea is not made in the proof. I'd prefer the proof to be presented top-down. If we do that, we can cut the plot discussion above. First start by restating the claim in the theorem in terms of the equivalence classes idea. Then reduce that to showing that the claim holds in individual worlds. Then give the ``for each world'' proof. } 
%
%\note[Wei Lu]{Re: For better flow, I moved the discussions on equivalent class of PWs to Sec 5.1.
%In fact, I left another note there, which suggests that we slightly adjust the generation process of PWs, and make the total number finite.
%It might be simpler and easier to go that route.}

\note[Laks]{Some of the disconnect remains. Don't we have any use for equivalence classes in this proof? We actually said wheer equiv. classes are introduced that they are used in the proof of Theorem 2!} 

\note[Wei]{Re: I agree.  I have made changes to address this issue.}

\begin{proof}
We first show the theorem holds in an arbitrary possible world $W$.
%Consider an arbitrary PW $W$ where all edge status and $\alpha$-values of all nodes are
%	determined.
Following Eq.~\eqref{eqn:labels}, when $S_\calA$ acts alone (without any $\calB$-seeds),
	we can classify all nodes into four
	types: $\calA$-adopted, $\calA$-rejected, $\calA$-suspended, and $\calA$-potential.
Clearly, only $\calA$-suspended and $\calA$-potential may later adopt $\calA$ with the help
	of $\calB$, and the other two types are not relevant to \CIM.
	
Now, let all nodes in $S_\calA$ be $\calB$-seeds, and thus they
	become $\calB$-adopted, too.
Since $\qbo = 1$, all nodes that can be reached by the seeds via live-edge paths will
	eventually become $\calB$-adopted, including all $\calA$-suspended and $\calA$-potential
	nodes.
For $\calA$-suspended nodes, after adopting $\calB$, they will adopt $\calA$ through
	re-consideration right away.
For $\calA$-potential nodes, the propagation of $\calA$ will reach them after $\calA$-suspended nodes
	upstream adopt $\calA$ (together with $\calB$), and after being reached, they adopt both items. 

Therefore, in this possible world, making $S_\calA$ to be $\calB$-seeds will ``convert''
	all nodes that can possibly become $\calA$-adopted, namely all $\calA$-suspended and
	$\calA$-potential nodes.
Hence, any additional $\calB$-seed will not further increase the spread of $\calA$
	as no other nodes can possibly be ``converted''.

Since the claim in this theorem holds true for any arbitrary possible world $W$, it
also holds for all other possible worlds that belong to the same equivalence
class $\bf W$.
Since
\[
\sigma_\calA(S_\calA, S_\calB) = \sum_{\bf W} \Pr[\mathbf{W}] \sigma_\calA^{\mathbf{W}}(S_\calA, S_\calB),
\]
the theorem is proven.
\end{proof}

%%%%%%%%%%%%%%%%%%

\subsection{Proofs for Results in Section~\ref{sec:submod}}

%%%%%%%%%%%%%%%%%%

\textbf{Lemma~\ref{lemma:pw}} (re-stated).
{\em \lemmaPW}

\begin{proof}
The proof is based on establishing equivalence on all edge-level and node-level behaviors in \model and the PW model.
By the principle of deferred decisions and the fact that each edge is only tested once in one diffusion, edge transition processes are equivalent.
To generate a possible world, the live/blocked status of an edge is pre-determined and revealed when needed, while in a Com-IC process, the status is determined on-the-fly.

Tie-breaking is also equivalent.
Note that each node $v$ only needs to apply the random permutation $\pi_v$ for breaking ties at most once.
For \comic, we need to apply $\pi_v$ only when $v$ is transitioning out of state ($\calA$-idle, $\calB$-idle) after being informed of both $\calA$ and $\calB$.
Clearly, this transition occurs at most once for each node.
The same logic applies to the PW model.
Thus, the equivalence is obvious due to the principle of deferred decisions.
%each in-neighbor of $v$ can open the information
%	channel to $v$ in at most one time step 
%	and thus it only needs to be checked in the random
%	permutation order once.

%\note[Wei Chen]{I changed the wording of the tie-breaking rule in both the rule-based
%	description and the possible world model.
%My change in the rule-based model is a convenient way to ensure that
%	event "$v$ is tested with in-neighbor $u$ and item $\calA$ or $\calB$" is always
%	well defined (alternative descriptions are possible).
%My change for the possible world model is to make it precise on the iterative 
%	descriptio step by step and on when tie-breaking is applied.
%The proof of the above paragraph reflect my changes in the model.}

The equivalence of decision-making for adoption is straightforward as $\alpha_\calA^v$ and $\alpha_\calB^v$ are chosen uniformly at random from $[0,1]$.
Hence, $\Pr[\alpha_\calA^v \leq q] = q$, where $q \in \{ \qao, \qab, \qbo, \qba \}$. 

As to reconsideration, w.l.o.g.\ we consider $\calA$.
In \model,  the probability of reconsideration is $\rho_\calA = \max\{ (\qab - \qao), 0 \} / (1 - \qao)$.
In PW, when $\qab \geq \qao$, this amounts to the probability that $\alpha_\calA^v \leq \qab$ given $\alpha_\calA^v > \qao$, which is $(\qab - \qao) / (1 - \qao)$.
On the other hand, when $\qab < \qao$, 
	$\alpha_\calA^v > \qao$ implies $\alpha_\calA^v > \qab$, which means reconsideration
	is meaningless, and this corresponds to 
	$\rho_\calA=0$ in \comic.
Thus, the equivalence is established.

Finally, the seeding protocol is trivially the same.
%The rest of the rules: initialization, seeding, and global iterations are trivially equivalent.
Combining the equivalence for all edge-level and node-level activities, we can see that the two models are equivalent and yield the same distribution of $\calA$- and $\calB$-adopted nodes, for any given $S_\calA$ and $S_\calB$.
\end{proof}

%%%%%%%%%%%%%%%%%%
\noindent\smallskip\textbf{Theorem~\ref{thm:monotone}} (re-stated).
{\em \theoremMonotone}

\smallskip
For ease of exposition,  we also state a
	symmetric version of Theorem~\ref{thm:monotone}
	w.r.t.\ $\sigma_\calB$.
That is, {\em given any fixed $\calA$-seed set,
	$\sigma_\calB(S_\calA, S_\calB)$ is monotonically
	increasing in $S_\calB$ for any set of GAPs in
	$\bQ^+$ and $\bQ^-$.
Also, $\sigma_\calB(S_\calA, S_\calB)$ is monotonically
	increasing in $S_\calA$ for any GAPs in $\bQ^+$,
	and monotonically decreasing in $S_\calA$ for
	any $\bQ^-$.}
For technical reasons and notational convenience,
	in the proof of Theorem~\ref{thm:monotone}
	presented below, we ``concurrently'' prove both Theorem~\ref{thm:monotone}
	and this symmetric version, without loss of generality.

\begin{proof}[Proof of Theorem~\ref{thm:monotone}]
We first fix a $\calB$-seed set $S_\calB$.
Since $S_\calB$ is always fixed, in the remaining proof we ignore $S_\calB$ from the notations whenever
	it is clear from context.
It suffices to show that monotonicity holds in an arbitrary, fixed possible world, which implies
	monotonicity holds for the diffusion model. 
Let $W$ be an arbitrary possible world generated according to \textsection\ref{sec:pw}.

Define $\Phi_\calA^W(S_\calA)$ (resp.\ $\Phi_\calB^W(S_\calA)$) to be the set of $\calA$-adopted (resp.\ $\calB$-adopted) nodes in possible world $W$ 
	with $S_\calA$ being the $\calA$-seed set (and $S_\calB$ being the fixed $\calB$-seed set).
Furthermore, for any time step $t \ge 0$, define $\Phi_\calA^W(S_\calA, t)$ (resp. $\Phi_\calB^W(S_\calA, t)$) to be the set of $\calA$-adopted (resp. $\calB$-adopted)  nodes in $W$ by the end of step $t$, given $\calA$-seed set $S_\calA$. % and (and $\calB$-seed set $S_\calB$).
Clearly, $\Phi_\calA^W(S_\calA) = \cup_{t\ge 0} \Phi_\calA^W(S_\calA, t)$ and $\Phi_\calB^W(S_\calA) = \cup_{t\ge 0} \Phi_\calB^W(S_\calA, t)$.
Let $S$ and $T$ be two sets, with $S\subseteq T\subseteq V$. 

\spara{Mutual Competition $\bQ^-$}
Our goal is to prove that for any $v \in V$, (a) if $v \in \Phi_\calA^W(S)$, then $v \in \Phi_\calA^W(T)$; 
	and (b) if $v \in \Phi_\calB^W(T)$, then $v \in \Phi_\calB^W(S)$.
Item (a) implies self-monotonic increasing property while item (b) implies cross-monotonic decreasing property.
We use an inductive proof to combine the proof of above two results together, as follows.
For every $t\ge 0$, we inductively show that $(i)$ if $v \in \Phi_\calA^W(S, t)$, then $v \in \Phi_\calA^W(T, t)$; 
	and $(ii)$ if $v \in \Phi_\calB^W(T, t)$, then $v \in \Phi_\calB^W(S, t)$.

Consider the base case of $t=0$.
If $v \in \Phi_\calA^W(S, 0)$, then it means $v \in S$, and thus $v\in T = \Phi_\calA^W(T, 0)$.
If $v \in \Phi_\calB^W(T, 0)$, it means $v \in S_\calB$, and thus $v \in \Phi_\calB^W(S, 0) = S_\calB$.

For the induction step, suppose that for all $t < t$, $(i)$ and $(ii)$ hold, and we show
	$(i)$ and $(ii)$ also hold for $t=t'$.
For $(i)$, we only need to consider $v \in \Phi_\calA^W(S, t') \setminus \Phi_\calA^W(S, t'-1)$,
	i.e. $v$ adopts $\calA$ at step $t'$ when $S$ is the $\calA$-seed set.
Since $v$ adopts $\calA$, we know that $\alpha_\calA^v \le \qao$.
Let $U$ be the set of in-neighbors of $v$ in the possible world $W$.
Let $U_\calA(S_\calA) = U \cap \Phi_\calA^W(S_\calA, t'-1)$ and
	$U_\calB(S_\calA) = U \cap \Phi_\calB^W(S_\calA, t'-1)$,
	i.e. $U_\calA(S_\calA)$ (resp. $U_\calB(S_\calA)$) is the set of in-neighbors of
	$v$ in $W$ that adopted $\calA$ (resp. $\calB$) by time $t'-1$, when $S_\calA$
	is the $\calA$-seed set.
Since $v \in \Phi_\calA^W(S, t')$, we know that $U_\calA(S) \ne \emptyset$.
By induction hypothesis, we have $U_\calA(S) \subseteq U_\calA(T)$ and $U_\calB(T) \subseteq U_\calB(S)$.

Thus, $U_\calA(T)\ne \emptyset$, which implies that by step $t'$, $v$ must have been informed of $\calA$ when $T$ is the $\calA$-seed set.
If $\alpha_\calA^v \le \qab$, then no matter if $v$ adopted $\calB$ or not, $v$ would adopt $\calA$ by step $t'$ according to the possible world model.
That is, $v \in \Phi_\calA^W(T, t')$.

Now suppose $\qab < \alpha_\calA^v \le \qao$.
For a contradiction suppose $v \not \in \Phi_\calA^W(T, t')$, i.e.,  $v$ does not adopt $\calA$ by step $t'$ when $T$ is the $\calA$-seed set.
Since $v$ has been informed of $\calA$ by $t'$, the only possibility that $v$ does not adopt $\calA$ is because $v$ adopted $\calB$ earlier than $\calA$,
	which means $v \in \Phi_\calB^W(T, t')$.
Two cases arise:

First, if $v \in \Phi_\calB^W(T, t'-1)$, then by the induction hypothesis $v \in \Phi_\calB^W(S, t'-1)$.
Since $v \in \Phi_\calA^W(S, t') \setminus \Phi_\calA^W(S, t'-1)$, it means that when $S$ is the $\calA$-seed set, $v$ adopts $\calB$ first before adopting $\calA$, but this contradicts to the condition that $\qab < \alpha_\calA^v$.
Therefore, $v \not \in \Phi_\calB^W(T, t'-1)$.

Second, $v \in \Phi_\calB^W(T, t') \setminus \Phi_\calB^W(T, t'-1)$.
Since $v \not \in \Phi_\calA^W(T, t')$, it means that $v$ is informed of $\calA$ at step $t'$ when $T$ is the $\calA$-seed set, and thus
	the tie-breaking rule must have been applied at this step and $\calB$ is ordered first before $\calA$.
However, looking at the in-neighbors of $v$ in $W$, by the induction hypothesis, 
	 $U_\calA(S) \subseteq U_\calA(T)$ and $U_\calB(T) \subseteq U_\calB(S)$.
This implies that when $S$ is the $\calA$-seed set, the same tie-breaking rule at $\calA$ would still order $\calB$ first before $\calA$, but
	this would result in $v$ not adopting $\calA$ at step $t'$, a contradiction.
Therefore, we know that $v \in \Phi_\calA^W(T, t')$.

%We now prove the induction step for $(ii)$.
The statement of $(ii)$ is symmetric to $(i)$: if we exchange $\calA$ and $\calB$ and exchange $S$ and $T$, $(ii)$ becomes $(i)$.
In fact, one can check that we can literally translate the induction step proof for $(i)$ into the proof for $(ii)$ by exchanging
	pair $\calA$ and $\calB$ and pair $S$ and $T$ (except that (a) we keep the definitions of $U_\calA(S_\calA)$ and $U_\calB(S_\calA)$, and 
	(b) whenever we say some set is the $\calA$-seed set, we keep this $\calA$).
This concludes the proof of the mutual competition case.

%We now consider the mutual complementarity case, i.e. $\qao \le \qab$ and $\qbo \le \qba$.
\spara{Mutual Complementarity $\bQ^+$}
The proof structure is very similar to that of the mutual competition case.
Our goal is to prove that for any $v \in V$, (a) if $v \in \Phi_\calA^W(S)$, then $v \in \Phi_\calA^W(T)$; 
	and (b) if $v \in \Phi_\calB^W(S)$, then $v \in \Phi_\calB^W(T)$.
To show this, we inductively prove the following: 
For every $t\ge 0$, $(i)$ if $v \in \Phi_\calA^W(S, t)$, then $v \in \Phi_\calA^W(T, t)$; 
	and $(ii)$ if $v \in \Phi_\calB^W(S, t)$, then $v \in \Phi_\calB^W(T, t)$.
The base case is trivially true.

For the induction step, suppose $(i)$ and $(ii)$ hold for all $ t < t'$, and we 
	show that $(i)$ and $(ii)$ also hold for $t=t'$.

For $(i)$, we only need to consider $v \in \Phi_\calA^W(S, t') \setminus \Phi_\calA^W(S, t'-1)$, i.e. $v$ adopts $\calA$ at step $t'$ when $S$ is the $\calA$-seed set.
Since $v$ adopts $\calA$, we know that $\alpha_\calA^v \le \qab$.
Since $v \in \Phi_\calA^W(S, t')$, we know that $U_\calA(S) \ne \emptyset$.
By induction hypothesis we have $U_\calA(S) \subseteq U_\calA(T)$.
Thus we know that $U_\calA(T)\ne \emptyset$, which implies that by step $t'$, $v$ must have been informed of $\calA$ when $T$ is the $\calA$-seed set.
if $\alpha_\calA^v \le \qao$, then no matter $v$ adopted $\calB$ or not, $v$ would adopt $\calA$ by step $t'$ according to the possible world model.
Thus, $v \in \Phi_\calA^W(T, t')$.

Now suppose $\qao < \alpha_\calA^v \le \qab$.
Since $v \in \Phi_\calA^W(S, t')$, the only possibility is that $v$ adopts $\calB$ first by time $t'$ so that after reconsideration, $v$ adopts $\calA$
	due to condition $\alpha_\calA^v \le \qab$.
Thus we have $v \in \Phi_\calB^W(S, t')$, and $\alpha_\calB^v \le \qbo$.

If $v \in \Phi_\calB^W(S, t'-1)$, by induction hypothesis $v \in \Phi_\calB^W(T, t'-1)$, which means that $v$ adopts $\calB$ by time $t'-1$ when
	$T$ is the $\calA$-seed set.
Since $v$ has been informed of $\calA$ by time $t'$ when $T$ is the $\calA$-seed set, condition $\alpha_\calA^v \le \qab$ implies that
	$v$ adopts $\calA$ by time $t'$ when $T$ is the $\calA$-seed set, i.e. $v \in \Phi_\calA^W(T, t')$.

Finally we consider the case of $v \in \Phi_\calB^W(S, t') \setminus \Phi_\calB^W(S, t'-1)$.
Looking at the in-neighbors of $v$ in $W$, $v \in \Phi_\calB^W(S, t')$, implies that $U_\calB(S) \ne \emptyset$.
By the induction hypothesis, we have $U_\calB(S) \subseteq U_\calB(T)$, and thus $U_\calB(T) \ne \emptyset$.
This implies that when $T$ is the $\calA$-seed set, node $v$ must have been informed of $\calB$ by time $t'$.
Since $\alpha_\calB^v \le \qbo$, we have that $v$ adopts $\calB$ by time $t'$ when $T$ is the $\calA$-seed set.
Then the condition $\alpha_\calA^v \le \qab$ implies that $v$ adopts $\calA$ by time $t'$ when $T$ is the $\calA$-seed set,
	i.e. $v \in \Phi_\calA^W(T, t')$.

This concludes the inductive step for item $(i)$ in the mutual complementarity case.
The induction step for item $(ii)$ is completely symmetric to the inductive step for item $(i)$. 
Therefore, we complete the proof for the mutual complementarity case.
As a result, the whole theorem holds.
\end{proof}

%%%%%%%%%%%%%%%%%%
\noindent\smallskip\textbf{Lemma~\ref{lemma:tie}} (re-stated).
{\em \lemmaTie}

\begin{proof}
Without loss of generality, we only need to consider a node $v$ and two of its in-neighbours $u_\calA$ and $u_\calB$ which become $\calA$-adopted and $\calB$-adopted at $t-1$ respectively.
In a possible world, there are nine possible combinations of the values of $\alpha_\calA^v$ and $\alpha_\calB^v$.
We show that in all such combinations, the ordering $\pi_1  = \langle u_\calA, u_\calB \rangle$ and $\pi_2 = \langle u_\calB, u_\calA \rangle$ produce the same outcome for $v$.
\begin{enumerate}
\item {$\alpha_\calA^v \leq \qao \wedge \alpha_\calB^v \leq \qbo$.}
Both $\pi_1$ and $\pi_2$ make $v$ $\calA$-adopted and $\calB$-adopted. 

\item {$\alpha_\calA^v \leq \qao \wedge \qbo < \alpha_\calB^v \leq \qba$.}
Both $\pi_1$ and $\pi_2$ make $v$ $\calA$-adopted and $\calB$-adopted. 
With $\pi_2$, $v$ first becomes $\calB$-suspended, then $\calA$-adopted, and finally $\calB$-adopted due to re-consideration.

\item {$\alpha_\calA^v \leq \qao \wedge \alpha_\calB^v > \qba$.}
Both $\pi_1$ and $\pi_2$ makes $v$ $\calA$-adopted only. 

\item {$\qao < \alpha_\calA^v \leq \qab \wedge \alpha_\calB^v \leq \qbo$.}
%With both orderings, $v$ would adopt both $\calA$ and $\calB$, 
Symmetric to case 2 above.
%In particular, if $U_\calA$ is ordered first, $v$ first becomes $\calA$-suspended and then adopts $\calA$ by reconsideration.

\item {$\qao < \alpha_\calA^v \leq \qab \wedge \qbo < \alpha_\calB^v \leq \qba$.}
%It is easy to see that $v$ would adopt neither item.
In this case, $v$ does not adopt any item.

\item {$\qao < \alpha_\calA^v \leq \qab \wedge \alpha_\calB^v > \qba$.}
In this case, $v$ does not adopt any item.

\item {$\alpha_\calA^v > \qab \wedge \alpha_\calB^v \leq \qbo$.}
Symmetric to case 3 above $v$ is $\calB$-adopted only.
%$v$ only adopts $\calB$ and cannot adopt $\calA$ in either ordering.

\item {$\alpha_\calA^v > \qab \wedge \qbo < \alpha_\calB^v \leq \qba$.}
Symmetric to case 6 above.

\item {$\alpha_\calA^v > \qab \wedge \alpha_\calB^v > \qba$.}
In this case, $v$ does not adopt any item.
\end{enumerate}

Since the possible world model is equivalent to Com-IC (Theorem~\ref{lemma:pw}), the lemma holds as a result.
\end{proof}

%%%%%%%%%%%%%%%%%%

\noindent\smallskip\textbf{Lemma~\ref{lemma:indiff}} (re-stated).
{\em \lemmaInDiff}

\begin{proof}
Consider an arbitrary possible world $W$.
Let $q := \qbo = \qba$.
A node $v$ becomes $\calB$-adopted in $W$ as long as $\alpha_\calB^v \leq q$ and there is a live-edge path $P_\calB$ from $S_\calB$ to $v$ such that for all nodes $w$ on $P_\calB$ (excluding seeds), $\alpha_\calB^w \leq q$.
Since $\qbo = \qba$, this condition under which $v$ becomes $\calB$-adopted in $W$ is completely independent of any node's state w.r.t.\ $\calA$.
Thus, the propagation of $\calB$-adoption is completely independent of the actual $\calA$-seed set (even empty).
Due to the equivalence of the possible world model and Com-IC, the lemma holds.
\end{proof}

%%%%%%%%%%%%%%%%%%

\noindent\smallskip\textbf{Claim~\ref{clm:bjoina}} (re-stated).
{\em \lemmaClaimBA}

\begin{proof}[Proof of Claim~\ref{clm:bjoina}]
Since $P_\calA$ is an $\calA$-path (where all nodes are $\calA$-adopted), every node $w$ on $P_\calA$ (except the starting node)
	has $\alpha^w_\calA \le \qab$.
For every such node $w$, if $\qao < \alpha^w_\calA \le \qab$, then $w$ adopting $\calA$ implies that
	$w$ must have adopted $\calB$ first and $\alpha^w_\calB \le \qbo$.
	
Now suppose a node $w$ on the path adopts $\calB$ (under some $\calB$-seed set), and all nodes before $w$
	on path $P_\calA$ are $\calA$-ready.
Then all nodes before $w$ on this path adopts $\calA$ regardless of $\calB$-seed set.
Thus $w$ is informed of $\calA$.
Since $\alpha^w_\calA \le \qab$, $w$ adopts $\calA$.

Then consider the node $w'$ after $w$ on the path $P_\calA$.
Node $w'$ must be informed by both $\calA$ and $\calB$ since $w$ adopts
	both $\calA$ and $\calB$ and the edge $(w,w')$ is live.
If $\alpha^{w'}_\calA \le \qao$, $w'$ will adopt $\calA$, and then since $\qba =1$,
	$w'$ will then adopt $\calB$ -- this is where we use the key assumption that $\qba = 1$.
If $\qao < \alpha^{w'}_\calA \le \qab$, then we have argued that in this case $\alpha^{w'}_\calB \le \qbo$,
	so $w'$ would adopt $\calB$, followed by adopting $\calA$.
We can then inductively use the above argument along the path to show that every node after
	$w$ will adopt both $\calA$ and $\calB$, which proves the claim.
\end{proof}

%%%%%%%%%%%%%%%%%%

\noindent\smallskip\textbf{Claim~\ref{clm:bpathcrosssub}} (re-stated).
{\em \claimBPath}

\begin{proof}[Proof of Claim~\ref{clm:bpathcrosssub}]
We construct such a path $P_\calB$ backwards from $w$ as follows.

Since $w$ adopts $\calB$, if $w \in T \cup\{u\}$, then we are done; otherwise, 
	there must be some $\calB$-path $P'_{\calB}$ from 
	$\calB$-seed set $T\cup \{u\}$ to $w$.
If $P'_{\calB}$ is a $\calB$-ready path, then we have found $P_\calB$ to be $P'_{\calB}$,
 	and its starting node is $x_0$. 
This is because even if $x_0$ is the only $\calB$-seed, $\calB$ can still pass through all nodes on $P_\calB$
	to reach $w$, without the need of any node on the path to adopt $\calA$, because by definition,
	all nodes $v$ on a $\calB$-ready path satisfies that $\alpha_\calB^v \leq \qbo$. 

Now suppose $P'_{\calB}$ is not $\calB$-ready.
	It is clear that every non-$\calB$-ready node on path
	$P'_{\calB}$ must be $\calA$-ready (otherwise, it would not be possible
	for such a node to adopt $\calB$).
If every such non-$\calB$-ready node has an $\calA$-ready path from $S_\calA$, then
	we still find $P_\calB = P'_{\calB}$ with its starting node is $x_0$.
This is because whenever $\calB$ reaches a non-$\calB$-ready node on the path, the node has an $\calA$-ready path
	from $S_\calA$ and thus the node always adopt $\calA$, which means it will also adopt $\calB$ (for $\qba = 1$).

Now consider the case where 
	there exists some non-$\calB$-ready node on path $P'_{\calB}$ that does not have 
	any $\calA$-ready path from $S_\calA$.
Let $x_1$ be the first such node (we count backwards from $w$).	
By the definition of $x_1$, we know that as long as $x_1$ adopts $\calB$
	(regardless what is the actual $\calB$-seed set), then $x_1$ would
	pass $\calB$ along path $P'_{\calB}$ to $w$.
Thus our backward construction has found the last piece of the path $P_\calB$ as the
	path segment of $P'_{\calB}$ from $x_1$ to $w$, now we move the construction
	backward starting from $x_1$.

At $x_1$, find a $\calA$-path $P'_{\calA}$ from $S_\calA$ to $x_1$.
By the definition of $x_1$\footnote{Recall that it is $\calA$-ready, on a $\calB$-path, and there is no $\calA$-ready path from $S_\calA$ to $x_1$}, we know that $P'_{\calA}$ is not $\calA$-ready.
Let $x_2$ be the first non-$\calA$-ready node on path $P'_{\calA}$ counting forward from the starting $\calA$-seed.
Then $x_2$ must have adopted $\calB$.
By Claim~\ref{clm:bjoina}, we know that as long as $x_2$ adopts $\calB$,
	all nodes after $x_2$ on path $P'_{\calA}$ would adopt both $\calA$ and $\calB$, regardless of the actual
	$\calB$-seed set.
Applying this to $x_1$, we know that $x_1$ is both $\calA$- and $\calB$-adopted.
Then our backward construction has found the next piece of $P_\calB$, which 
	is the path segment of $P'_{\calA}$ from $x_2$ to $x_1$, which guarantees that if $x_2$
	adopts $\calB$, then $w$ must eventually adopt $\calB$.
	
Now if $x_2$ is a $\calB$-seed, we are done.
If not, we will repeat the above ``zig-zag constructions'': there must be a
	$\calB$-path $P''_{\calB}$ from $\calB$-seed set $T\cup \{u\}$ to $x_2$.
The argument on path $P''_{\calB}$ is exactly the same as the argument on $P'_{\calB}$, 
	and if the construction still cannot stop, we will find path $P''_{\calA}$
	similar to $P'_{\calA}$, and so on.
	
The construction keeps going backwards, and since the construction actually follows the strict adoption
	time line and going backward in time, in the diffusion process when $T\cup \{u\}$ is the $\calB$-seed set,
	the construction must eventually reach a $\calB$-seed $x_0$ and stops.
Then we have found the desired path $P_\calB$ and the desired starting node $x_0$.
\end{proof}

%\subsection{Proof of Theorem~\ref{thm:cimOpt}}

\subsection{Proofs for Results in Section~\ref{sec:rrset}}

%%%%%%%%%%%%%%%%%%
\noindent\smallskip\textbf{Lemma~\ref{lemma:generic-rr-2}} (re-stated).
{\em \lemmaGenericRRSubmod}

\begin{proof}
First consider ``if''.
Suppose both properties hold in $W$.
Monotonicity directly follows from Property~{\it (P1)}.
For submodularity, suppose $v$ can be activated by set $T \cup \{x\}$ but not by $T$, where $x\not\in T$.
By Property~{\it (P2)}, there exists some $u \in T \cup \{x\}$ such that $\{u\}$ can activate $v$ in $W$.
If $u \in T$, then $T$ can also activate $v$ by Property~{\it (P1)}, a contradiction.
Hence we have $u = x$.
Then, consider any subset $S \subset T$.
Note that by Property~{\it (P1)}, $S$ cannot activate $v$ (otherwise so could $T$), while $S \cup \{x\}$ can.
Thus, $f_{v,W}(\cdot)$ is submodular.

Next we consider ``only if''.
Suppose $f_{v,W}(\cdot)$ is monotone and submodular for every $v\in V$.
Property~{\it (P1)} directly follows from monotonicity.
For Property~{\it (P2)}, suppose for a contradiction that there exists a seed set $S$ that can activate $v$ in $W$, but there is no $u\in S$ so that $\{u\}$ activates $v$ alone.
We repeatedly remove elements from $S$ until the remaining set is the minimal set that can still activate $v$.
Let the remaining set be $S'$.
Note that $S'$ contains at least two elements.
Let $u\in S'$, and then we have
	$f_{v,W}(\emptyset)=f_{v,W}(\{u\})=f_{v,W}(S'\setminus \{u\})=0$,
	but $f_{v,W}(S') = 1$, which violates submodularity, a contradiction.
This completes the proof.
\end{proof}

%%%%%%%%%%%%%%%%%%
\noindent\smallskip\textbf{Lemma~\ref{lemma:generic-rr}} (re-stated).
{\em \lemmaGenericRR}

\begin{proof}
It is sufficient to prove that in every possible world $W \in \mathcal{W}$, $S$ activates $v$ if and only if $S$ intersects with $v$'s RR set in $W$, denoted by $R_W(v)$.

Suppose $R_W(v) \cap S \neq \emptyset$.
Without loss of generality, we assume a node $u$ is in the intersection.
By the definition of RR set, set $\{u\}$ can activate $v$ in $W$.
Per Property~{\it (P1)}, $S$ can also activate $v$ in $W$.

Now suppose $S$ activates $v$ in $W$.
Per Property~{\it (P2)}, there exists $u\in S$ such that $\{u\}$ can also activate $v$ in $W$.
Then by the RR-set definition, $u \in R_W(v)$.
Therefore, $S \cap R_W(v) \neq \emptyset$.
\end{proof}

%%%%%%%%%%%%%%%%%%
\noindent\smallskip\textbf{Theorem~\ref{thm:rrsim-correct}} (re-stated).
{\em \theoremRRSIMCorrect}

\begin{proof}
It suffices to show that, given a fixed possible world $W$, a fixed $\calB$-seed set $S_\calB$, and a certain node $u\in V$, for any node $v\not\in \Phi_\calA^W(\emptyset, S_\calB)$ with $\alpha_\calA^v \leq \qab$, we have: $v\in \Phi_\calA^W(\{u\}, S_\calB)$ if and only if there exists a live-edge path $P$ from $u$ to $v$ such that for all nodes $w\in P$, excluding $u$, $w$ satisfies $\alpha_\calA^w \leq \qab$, and in case $\alpha_\calA^w > \qao$,  then $w$ must be $\calB$-adopted.

The ``if'' direction is straightforward as $P$ will propagate the adoption information of $\calA$
	all the way to $v$.
If $\alpha_\calA^v \leq \qao$, it adopts $\calA$ without question.
If $\alpha_\calA^v \in (\qao, \qab]$, then $v$ must be $\calB$-adopted by the definition of $P$, which makes it $\calA$-adopted.

For the ``only if'' part,  suppose no such $P$ exists for $u$.
This leads to a direct contradiction since $u$ is the only $\calA$-seed, and $u$ lacks a live-edge path to $v$, it is impossible for $v$ to get informed of $\calA$, let alone adopting $\calA$.
Next, suppose there is a live-edge path $P$ from $u$ to $v$, but there is a certain node $w\in P$ which violates the conditions set out in the lemma.
First, $w$ could be have a ``bad'' threshold: $\alpha_\calA^w > \qab$.
In this case, $w$ will not adopt $\calA$ regardless of its status w.r.t.\ $\calB$, and hence the propagation of $\calA$ will not reach $v$.
Second, $w$ could have a threshold such that $\alpha_\calA^w \in (\qao, \qab]$ but it does not adopt $\calB$ under the influence of the  given $S_\calB$.
Similar to the previous case, $w$ will not adopt $\calA$ and the propagation of $\calA$ will not reach $v$.
This completes the ``only if'' part.

Then by Definition~\ref{def:generalRRset}, the theorem follows.
\end{proof}

%%%%%%%
\noindent\smallskip\textbf{Lemma~\ref{lemma:rr-sim-time}}
{\em \lemmaRRSIMRuntime } (re-stated).

\begin{proof}
Given a fixed RR-set $R\subseteq V$, let $\omega(R)$ be the number
	of edges in $G$ that point to nodes in $R$.
Since in \timsim, it is possible that we do not examine incoming
	edges to a node added to the RR-set
	({\em cf}.\ Cases 1$(ii)$ and 2$(ii)$ in the backward BFS),
	we have:
\begin{align*}
\EPT_B \leq \mathbb{E}[\omega(R)],
\end{align*}
where the expectation is taken over the random choices of $R$.
By Lemma 4 in \cite{tang14} (note that this lemma only relies on 
	the activation equivalence property of RR-sets, which holds true
	in our current one-way complementarity setting), 
\begin{align*}
\frac{|V|}{|E|}\cdot \mathbb{E}[\omega(R)] \leq \OPT_k.
\end{align*}
This gives
\[
\frac{|V|}{|E|} \cdot \EPT_B \leq \OPT_k.
\]
Following the same analysis as in~\cite{tang14} one can check that the lower bound
	$LB$ of $\OPT_k$ obtained by the estimation method in
		\cite{tang14} guarantees that $LB \ge \EPT_B\cdot |V|/|E|$.
Since in our algorithm we set $\theta = \lambda / LB$, where (following Eq.\eqref{eqn:opt})
\[
\lambda = \epsilon^{-2} \left((8+2\epsilon)|V| \left(\ell \log|V| + \log\binom{|V|}{k} + \log 2\right) \right),
\]
then we have that the expected running time of generating all RR-sets is:

\begin{align*}
O(\theta \cdot \EPT)
&= O(\frac{\lambda}{LB} \cdot (\EPT_F + \EPT_B)) \\
&= O\left(\frac{\lambda|E|}{|V| \EPT_B} (\EPT_F + \EPT_B) \right) \\
&= O\left(\frac{\lambda|E|}{|V|} \left(1+ \frac{\EPT_F}{\EPT_B} \right) \right) \\
& = O\left((k+\ell)(|V|+|E|)\log |V|\left(1+ \frac{\EPT_F}{\EPT_B}\right) \right).
\end{align*}

The time complexity for estimating $LB$ and for calculating the final seed set 
	given RR-sets are the same as in~\cite{tang14}, and thus the final complexity
	is $O\left((k+\ell)(|V|+|E|)\log |V|\left(1+ \frac{\EPT_F}{\EPT_B}\right) \right)$.
\end{proof}

%%%%%%%%%%%%%%%%%%
\noindent\smallskip\textbf{Theorem~\ref{thm:rrcim-correct}} (re-stated).
{\em \theoremRRCIMCorrect}

\begin{proof}
It suffices to show the following claim.

\begin{claim}\label{clm:rrcim-correct}
Consider any fixed possible world $W$, a fixed $\calA$-seed set $S_\calA$, and a certain node $u\in V$.
Then, for any node $v\not\in \Phi_\calA^W(S_\calA, \emptyset)$ with $\alpha_\calA^v \leq \qab$, we have: $v\in \Phi_\calA^W(S_\calA, \{u\})$ if and only if there exists a live-edge path $\mathcal{P}(u,v)$ from $u$ to $v$ such that one of the following holds :
\begin{itemize}
\item $(i)$.\ $\mathcal{P}(u,v)$ consists entirely of $\calA$-adopted or diffusible $\calA$-suspended/potential nodes, and $u$ must be $\calA$-suspended, or
\item $(ii)$.\ There exists an $\calA$-suspended node $u' \neq u$ on $\mathcal{P}(u,v)$, such that all nodes on the sub-path of $\mathcal{P}(u,u')$ (excluding $u$) are $\calB$-diffusible and the remaining ones (excluding $u'$ and $v$) are either $\calA$-adopted or diffusible $\calA$-suspended/potential.
\end{itemize}
\end{claim}

\begin{proof}[Proof of Claim~\ref{clm:rrcim-correct}]
$(\Longleftarrow).$
Suppose $(i)$ holds.
Clearly, $u$ will adopt $\calB$ first as a seed and then adopt $\calA$ by reconsideration.
Then all nodes on $\mathcal{P}(u,v)$ including $v$ will adopt both $\calA$ and $\calB$ since they are either $\calA$-adopted or diffusible $\calA$-suspended/potential.
Consider any node $w \neq v$ on  $\mathcal{P}(u,v)$: if $w$ is already $\calA$-adopted, it will adopt $\calB$ since $\qba = 1$; otherwise, $w$ adopts $\calB$ first since $\alpha_\calB^w \leq \qbo$ (diffusible) and then $\calA$.
Next, suppose $(ii)$ holds.
Since $S_\calB = \{u\}$, $u$ adopts $\calB$, so do all nodes on $\mathcal{P}(u,u')$.
Since $u'$ is $\calA$-suspended, it will then adopt $\calA$ by reconsideration.
The rest of the argument is exactly the same as the case above.

$(\Longrightarrow).$
Since $v$ is $\calA$-adopted, there must exist a live-edge path on $P_\calA$ from $S_\calA$ to $v$ such that all nodes on $P_\calA$ are $\calA$-adopted.
Suppose every node $w \in P_\calA$ have $\alpha_\calA^w \leq \qao$, then all nodes on $P_\calA$ would adopt $\calA$ when $S_\calB = \emptyset$, including $v$, which is a contradiction.
Thus, there exists at least one node $w\in P_\calA$ such that $\qao < \alpha_\calA^w \leq \qab$.
Let $w_1$ be the first of such nodes, counting from upstream.
In order for $v$ to adopt $\calA$, $w_1$ must adopt $\calA$, and in order for that to happen, $w_1$ must adopt $\calB$ first.

If the sub-path $\mathcal{P'}(w_1,v)$ does not have non-diffusible nodes, then it satisfies $(i)$ and $w_1 = u$ as in the theorem.
Now suppose there is one node $w_2 \in \mathcal{P'}(w_1,v)$ that is non-diffusible, i.e., $\alpha_\calB^{w_2} > \qbo$ and $\qao < \alpha_\calA^{w_2} \le \qab$.
This means in order for $w_2$ to adopt $\calA$ and $\calB$, it must be a seed for either or both items.
By construction, $w_2$ is not an $\calA$-seed, and thus it must be a $\calB$-seed.
Next, since $w_1$ adopts $\calB$ and $w_2$ is the only $\calB$-seed, there must be a live-edge path of $\calB$-diffusible nodes from $w_2$ to $w_1$.
In this way, $\mathcal{P'}(w_2,v)$ satisfies $(ii)$, with $w_2 = u$ and $w_1 = u'$ as in the theorem.
Also note that if there are more than one non-diffusible nodes on $\mathcal{P'}(w_1,v)$, all of them must be $\calB$-seeds in order to make $v$ adopt $\calA$.
Since the theorem only considers singleton $\calB$-seed sets, such paths will not lead $v$ to adopt $\calA$ with only one non-diffusible node as $\calB$-seed.
This completes the proof.
\end{proof}

By Definition~\ref{def:generalRRset}, the theorem thus holds true.
\qed
\end{proof}

%%%%%%%%%%
\noindent\smallskip\textbf{Lemma~\ref{lemma:rrcim-time}} 
{\em \lemmaRRCIMTime} (re-stated).

\begin{proof}
The following inequality continues to hold for \timcim,
	by applying Lemma 4 in \cite{tang14}.
\[
 |V| \cdot \EPT_{BS} / |E| \leq \OPT_k,
\]
%\[
% |V| \cdot \mathbb{E}[\omega(R)] / |E| \leq \OPT_k,
%\]
%where $\omega(R)$ is the number of edges pointing
%	to a node in $R$, given a fixed $R$.
%To estimate a lower bound of $\OPT_k$, for \timcim,
%	we do not include the eligible non-diffusible
%	$\calA$-potential nodes that are also in $R$, but
%	only focus on edges pointing to $\calA$-suspended nodes.
%Clearly,
%\[
%|V| \cdot EPT_{BS} / |E| \leq  |V| \cdot \mathbb{E}[\omega(R)] / |E| \leq \OPT_k.
%\]

%With this, in the execution of \generalTIM \& \timcim, we can 
%	estimate $LB$ of $\OPT_k$ by the method in \cite{tang14}, but
%	only counting the edges pointing to $\calA$-suspended nodes.
%We use $\EPT_{BO}$ to denote the expected number of all other
%	edges examined in Phase~\rom{2}, including primary
%	and all secondary searches.
Since we employ lazy sampling (of possible worlds), we conclude
	that $\EPT_F + \EPT_{BS} + \EPT_{BO} \leq |E|$.
That is, in the process of generating a single RR-set, the worst case
	is to flip a coin for every edge, and every edge becomes live, and
	the search goes on.
The rest of the analysis is similar to Lemma~\ref{lemma:rr-sim-time}, and finally,
	the expected time complexity for \generalTIM + \timcim is
\[
O\left((k+\ell)(|V|+|E|)\log |V|\left(1+ \frac{\EPT_F+\EPT_{BO}}{\EPT_{BS}}\right) \right).
\]
\end{proof}

%%%%%%%%%%%%%%%%%%
%\smallskip
%\textsc{Lemma~\ref{lemma:rr-sim-plus}} (re-stated).
%{\em \lemmaRRSIMPlus}

%%%%%%%%%%%%%%%%%%
\noindent\smallskip\textbf{Theorem~\ref{thm:sandwich}} (re-stated).
{\em 
Sandwich Approximation solution gives:
$$ %\begin{align} \label{eqn:sandapprox}
\sigma(S_{\mathit{sand}}) \ge \max \Big\{\frac{\sigma(S_\nu)}{\nu(S_\nu)}, \frac{\mu(S_\sigma^*)}{\sigma(S_\sigma^*)} \Big\}
  \cdot (1-1/e) \cdot  \sigma(S_\sigma^*),
$$ %\end{align}
where $S_\sigma^*$ is the optimal solution maximizing $\sigma$ (subject to  cardinality constraint $k$).
}

\begin{proof}
Let $S_\mu^*$ and $S_\nu^*$ be the optimal solution to maximizing $\mu$ and $\nu$ respectively.
We have
{\normalsize 
\begin{multline}\label{eqn:Snu}
 \sigma(S_\nu)  = \frac{ \sigma(S_\nu)}  {\nu(S_\nu)} \cdot \nu(S_\nu)  \geq \frac{\sigma(S_\nu)} {\nu(S_\nu)} \cdot (1-1/e) \cdot \nu(S_\nu^*) \\ 
 \geq \frac{\sigma(S_\nu)}{\nu(S_\nu)} \cdot (1-1/e) \cdot \nu(S_\sigma^*) \geq \frac{\sigma(S_\nu)}{\nu(S_\nu)} \cdot (1-1/e) \cdot \sigma(S_\sigma^*);
\end{multline}}
\vspace{-4mm}
{\normalsize
\begin{multline}\label{eqn:Smu}
\sigma(S_\mu) \geq \mu(S_\mu) \geq (1-1/e) \cdot \mu(S_\mu^*) \geq (1-1/e) \cdot \mu(S_\sigma^*) \\
 \geq \frac{\mu(S_\sigma^*)}{\sigma(S_\sigma^*)} \cdot (1-1/e) \cdot  \sigma(S_\sigma^*).
\end{multline}
}

The theorem follows by applying Eq. \eqref{eqn:sand}, the definition of $S_{\mathit{sand}}$.
\end{proof}

%%%%%%%%%%%%%%%%%%
\noindent\smallskip\textbf{Theorem~\ref{thm:mono-q}} (re-stated).
{\em \theoremMonoQ}

\begin{proof}[Proof (Sketch)]
The detailed proof would follow the similar induction proof structure for each possible world 
	as in the proof of Theorem~\ref{thm:monotone}.
Intuitively, 
%	for the mutual competition case, we would prove inductively that at every step increasing
%	$\qao$ or $\qab$ would increase $\calA$-adopted nodes and decrease $\calB$-adopted nodes;
%	for the mutual complementarity case, 
	we would prove inductively that at every step increasing
	$\qao$ or $\qab$  would increase both $\calA$-adopted and $\calB$-adopted nodes.
\end{proof}

%%%%%%%%%%%%%%%%%%%%%%%%%%%%%%%%%%%%
%%%%%%%%%%%%%%%%%%%%%%%%%%%%%%%%%%%%
%%%%%%%%%%%%%%%%%%%%%%%%%%%%%%%%%%%%

\subsection{Submodularity  Analysis for Competitive Cases of Com-IC}

%In competitive cases, there is no meaningful reasons
%	to study cross-submodularity as in the complementary settings.

First, we address cross-submodularity.
Note that by Theorem~\ref{thm:monotone}, $\sigma_\calA$ is monotonically decreasing w.r.t.\ $S_\calB$ (with any fixed $S_\calA$).
Intuitively, cross-submodularity for competitive products means adding an additional $\calB$-seed to a smaller $\calB$-seed set yields a larger decrease in the spread of $\calA$: for any $S\subseteq T\subseteq V$ and any $u\not\in T$, $\sigma_\calA(S_\calA, S) - \sigma_\calA(S_\calA, S\cup \{u\}) \geq \sigma_\calA(S_\calA, T) - \sigma_\calA(S_\calA, T\cup \{u\})$.
This notion is relevant and useful to the problem of influence blocking maximization, where one party wants to find the best seed set to block the spread of competitors~\cite{HeSCJ12,BudakAA11}.
Since influence blocking is not the focus of this work, from here on, we focus on self-submodularity and self-monotonicity only, and hereafter we
	drop ``self-'' in the terminologies.

When $q_{A|\emptyset} = q_{B|\emptyset} = 1$ and $q_{A|B} = q_{B|A} = 0$, the Com-IC model degenerates to 
	the homogeneous CIC model, for which submodularity and monotonicity both hold~\cite{infbook}.
In the general $\bQ^-$ setting, monotonicity holds (Theorem~\ref{thm:monotone}).
However, the following counter-example shows the opposite for submodularity.

\begin{example}%[Non-Self-Submodularity]
\label{exp:nonsub}
{\em
Consider the graph in Figure~\ref{fig:non-sub-compete}, where all edges have an influence probability of $1$.
Values of $\bQ$ are: $\qao = q \in (0,1)$, $\qab = \qba = 0$, $\qbo = 1$.
The $\calB$-seed set is $S_\calB = \{y\}$.
For $\calA$-seed set $S_\calA$, let $S = \{s_1\}$, $T = \{s_1, s_2\}$, and $u = s_3$.

We consider the probability $v$ becomes $\calA$-adopted with $S_\calA$, denoted by $ap_\calA(v, S_\calA)$ ($S_\calB$ is omitted since it is clear from context):
\begin{align*}
&ap_\calA(v, S) = 0, \\
&ap_\calA(v, S \cup \{u\}) = q^2, \\
&ap_\calA(v,T) = 0, \\
&ap_\calA(v,T \cup \{u\}) = q^2 + (1-q)\cdot q^6.
\end{align*}
%Thus, the marginal gain of $u$ w.r.t.\ $T$ is larger by $(1-q)\cdot q^6 > 0$, which is positive since $q\in (0,1)$.
Hence, we have
\begin{align*}
&\big(ap_\calA(v,T \cup \{u\}) - ap_\calA(v,T) \big) - \big(ap_\calA(v, S \cup \{u\}) - ap_\calA(v,S) \big)  \\
&= (1-q)\cdot q^6 > 0.
\end{align*}
Non-submodularity occurs for the entire graph if $v$ is replicated sufficiently many times.
\qed
}
\end{example}

The intuition of the above counter-example is that $\calA$-seeds $s_2$ and $s_3$
	together can block $\calB$ completely, and
	thus even if they cannot successfully activate $v$ to adopt $\calA$, $s_1$ could 
	later activate $v$ (with the additional probability $(1-q)\cdot q^6$), but
	when $s_2$ or $s_3$ acts alone, it cannot block $\calB$, so the influence of $\calA$
	from $s_1$ will not reach $v$.

\begin{figure}[t]
 \centering
   \includegraphics[width=0.5\textwidth]{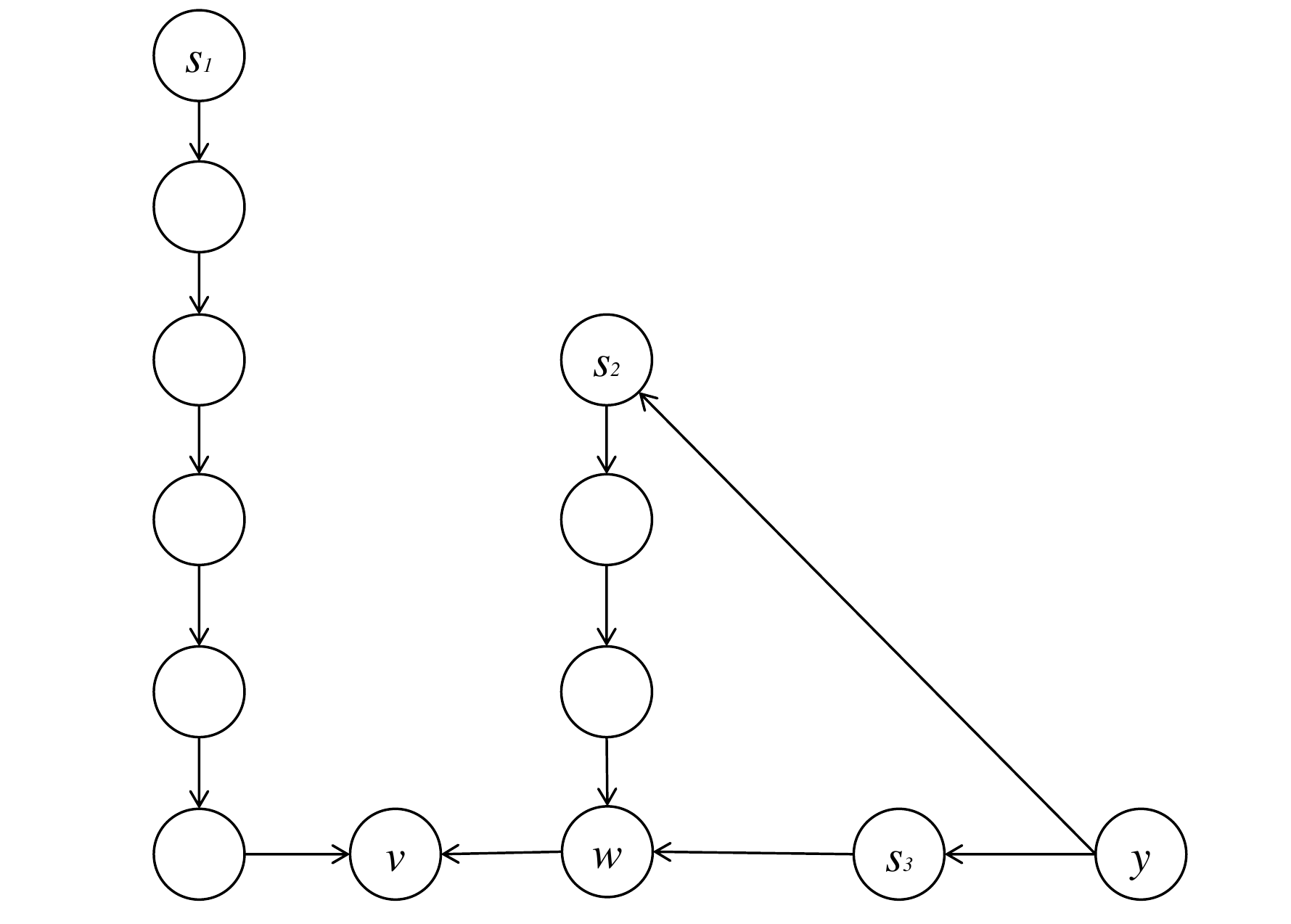}
   \caption{Graph for Example~\ref{exp:nonsub}}
   \label{fig:non-sub-compete}
\end{figure}

Next, we show a positive result which says that submodularity is satisfied
	as long as $q_{A|\emptyset} = q_{B|\emptyset} = 1$.

\begin{theorem}
In the Com-IC model, when $q_{A|\emptyset} = q_{B|\emptyset} = 1$, the influence spread function $\sigma_\calA(S_\calA, S_\calB)$ is \emph{submodular} w.r.t.\ $S_\calA$, for any given $S_\calB$.
\end{theorem}

\begin{proof}
We fix a $\calB$-seed set $S_\calB$ and consider an arbitrary possible world $W$ to show
	that submodularity is satisfied in $W$.
%Hence, by Theorem~\ref{lemma:pw} and the linearity of submodularity, this submodularity of $\sigma^X(S_\calA, S_\calB)$ follows.

In the possible world $W$, for an $\calA$-adopted node $v$, there must exist a live-edge path from $\calA$-seed set $S_\calA$ to
	$v$ with all nodes on the path adopting $\calA$, and we call this path an $\calA$-path
	($\calB$-path is defined symmetrically).
With mutual competition, the length of the shortest $\calA$-path from $S_\calA$ to $v$ is the
	same as the time step at which $v$ adopts $\calA$.
This is because in competitions when a node is informed about an item, it only has one chance at the
	same time step to decide about the adoption (reconsideration in later steps is not possible).

\begin{claim}\label{clm:aorb}
If a node $v$ is reachable from $S_\calA$ or $S_\calB$ in the possible world $W$, then $v$ must adopt at least one of $\calA$ and $\calB$.
\end{claim}

\begin{proof}[Proof of Claim~\ref{clm:aorb}]
This is due to the fact that $\qao = \qbo = 1$.
In this case, suppose a node $v$ is reachable from $S_\calA$, i.e. there is a path from $S_\calA$ to $v$ in the possible world $W$.
The claim can be proven by an induction on the path length.
Essentially, no matter which item an earlier node on the path adopts, it would inform the next node on the path, and since
	$\qao=\qbo=1$, the node would adopt the first item it get informed.
\end{proof}

Let $S\subseteq T \subseteq V$ be two sets of nodes, and $u\in V \setminus T$ be an additional $\calA$-seed.
For a fixed $\calB$-seed set $S_\calB$ (omitted from the following notations), let $\Phi_\calA^W(S_\calA)$ denote the set of $\calA$-adopted nodes in the possible world $W$ with $\calA$-seed set $S_\calA$ after the diffusion ends.
Consider a node $v \in \Phi_\calA^W(T \cup \{u\}) \setminus \Phi_\calA^W(T)$.
We need to show $v \in \Phi_\calA^W(S \cup \{u\}) \setminus \Phi_\calA^W(S)$.

We now construct a shortest $\calA$-path, denoted $P_\calA$, from some node 
	in $w_0 \in T \cup \{u\}$ to $v$ when $S_\calA = T \cup \{u\}$.
The path construction is done backwards.
Starting from $v$, if $v$ has only one in-neighbor $w$ attempting to inform $v$ through the 
	live edge (open information channel) from $w$ to $v$, 
	then in the previous time step, $w$ must have adopted $\calA$, and we select $w$ as the predecessor of $v$ on path $P_\calA$. 
If there are multiple in-neighbors that have live edges pointing to $v$, 
%let $w$ be the one ordered before all such in-neighbors which adopted $X$ or $XY$,
we examine the permutation $\pi_v$ and pick $w$ as $v$'s predecessor on this path, such that $w$ ranks first over all of $v$'s in-neighbors that adopted $\calA$ in the previous time step.
%and we select $w$ as the predecessor of $v$.
%Similarly, this node $w$ must have adopted $X$, or both with $X$ being adopted first.
We trace back and stop once we reach an $\calA$-seed in $T \cup \{u\}$.
It is clear that $P_\calA$ must be a shortest $\calA$-path from $T \cup \{u\}$ to $v$, since in every construction
	step we move back one time step.

\begin{claim}\label{clm:source}
Let $w_0 \in T \cup \{u\}$ be the starting point of $P_\calA$.
Then, even when $S_\calA = \{w_0\}$, all nodes on $P_\calA$ would still adopt $\calA$.
\end{claim}

\begin{proof}[Proof of Claim~\ref{clm:source}]
Suppose, for a contradiction, that some node on $P_\calA$ \textsl{does not adopt $\calA$} when $S_\calA = \{w_0\}$.
Let $w$ be the first of such nodes on $P_\calA$ (counting from $w_0$)\footnote{$w \neq w_0$, since as a seed, $w_0$ adopts $\calA$ automatically.}.
By Claim~\ref{clm:aorb}, $w$ must adopt $\calB$ as it does not adopt $\calA$.
This means that there must exist a $\calB$-path from $S_\calB$ to $w$.
We construct a shortest $\calB$-path $P_\calB$  from $S_\calB$ to $w$ in the same way as we constructed $P_\calA$:
	start from $w'=w$ backwards and always select the in-neighbor of the current node $w'$ that adopts $\calB$ and is ordered first in $\pi_w$.
Moreover, since $w$ is informed of $\calA$ (by its predecessor on path $P_\calA$) but does not adopt $\calA$,
	it holds that $\qab < \alpha_\calA^w \leq \qao$ in $W$.

%\note[Wei]{Need to consider: what if $w_\calA = w_\calB$?}

We now compare the length of $P_\calB$, denoted by $\ell_\calB$, to the length
	of the segment of $P_\calA$ from $w_0$ to $w$, denoted by $\ell_\calA$.
Suppose $\ell_\calB > \ell_\calA$. Then $w$ would still adopt $\calA$ when $S_\calA = \{w_0\}$
	(regardless of whether it will adopt $\calB$ or not), a contradiction.
%If $\ell(P_\calB) < \ell(P_\calA[w_0, w])$, then $w$ would adopt $\calB$ first and reject $\calA$ even when $S_\calA = T \cup \{u\}$. % (as $P_\calB$ would still be effective with $T \cup \{u\}$, ).
Next consider $\ell_\calB < \ell_\calA$.
When $T \cup \{u\}$ is the $\calA$-seed set, from $\calB$-seed set $S_\calB$ and through the path $P_\calB$, every time step one more node
	on path $P_\calB$ has to adopt either $\calA$ or $\calB$ (by Claim~\ref{clm:aorb}).
Thus by time step $\ell_\calB$, $w$ adopts either $\calA$ or $\calB$.
But since $P_\calA$ is the shortest $\calA$-path from $T \cup \{u\}$ to $w$ and $\ell_\calA > \ell_\calB$, $w$ cannot adopt $\calA$
	at step $\ell_\calB$, so $w$ must adopt $\calB$ at step $\ell_\calB$.
Since $\alpha_\calA^w > \qab$, it means that $w$ would not adopt $\calA$ after adopting $\calB$, also a contradiction.

We are left with the case $\ell_\calB = \ell_\calA$.
Let $w_\calA$ and $w_\calB$ be the predecessor of $w$ on $P_\calA$ and $P_\calB$, respectively.
Assume $w_\calA \neq w_\calB$ (we deal with $w_\calA = w_\calB$ later).
If $w_\calA$ is ordered ahead of $w_\calB$ in $\pi_w$, then according to the tie-breaking rule, 
	$w$ would be informed of $\calA$ first and adopt $\calA$ (since $\alpha_\calA^w \le \qao$) when $w_0$ is the only $\calA$-seed,
	contradicting the definition of $w$.

If $w_\calB$ is ordered ahead of $w_\calA$ in $\pi_w$, then consider again the scenario when $S_\calA = T \cup \{u\}$.
By Claim~\ref{clm:aorb},  through the path $P_\calB$, $w_\calB$ adopts either $\calA$ or $\calB$ by step $\ell_\calB-1$.
If $w_\calB$ adopts $\calA$, it contradicts the construction of $P_\calA$ since we would have chosen $w_\calB$ instead of $w_\calA$ in the
	backward construction.
If $w_\calB$ adopts $\calB$, then $w$ would be informed of $\calB$ from $w_\calB$ first, and then due to
	$\alpha_\calA^w > \qab$, $w$ would not adopt $\calA$, again a contradiction.

Finally, we also need to consider the case of $w_\calA=w_\calB$.
Since $w$ does not adopt $\calA$ when $w_0$ is the only $\calA$-seed, according to our tie-breaking rule, node
	$w_\calA$ must adopt $\calB$ first and then adopt $\calA$ at the same step.
We then trace back the predecessor of $w_\calA$ on path $P_\calA$  and $P_\calB$ respectively.
If these two paths never branch backward, i.e. $P_\calA = P_\calB$, 
	then we know that $w_0$ is both $\calA$ and $\calB$-seed.
According to the possible world model we would use $\tau_{w_0}$ to decide whether $w_0$ adopts $\calA$ or $\calB$ first, and
	then all nodes on the path $P_\calA = P_\calB$ would follow the same order.
Since $w_\calA$ adopts $\calB$ first before $\calA$, we know that $\tau_{w_0}$ is such that $\calB$ is ranked ahead of $\calA$.

Now we consider the scenario when $S_\calA = T \cup \{u\}$.
By the construction of path $P_\calA$, the successor of $w_0$ on this path, $w_1$, orders $w_0$ first among all
	$\calA$-adopted in-neighbors, so $w_1$ would take $w_0$ to get informed of $\calA$, but since $w_0$ adopts $\calB$ first,
	this implies that $w_1$ also adopts $\calB$ first.
We can apply the same argument along the path $P_\calA$ to see that $w$ adopts $\calB$ first before $\calA$.
However, since $\alpha_\calA^w > \qab$, we know that $w$ would not adopt $\calA$, a contradiction.

If paths $P_\calA$ and $P_\calB$ branch at some node $x$ when we trace backward, then let $x_\calA$
	and $x_\calB$ are the two predecessors of $x$ on paths $P_\calA$  and $P_\calB$ respectively, and 
	$x_\calA \ne x_\calB$.
When $w_0$ is the only $\calA$-seed, we know from above that $w_\calA$ adopts $\calB$ first before adopting $\calA$.
By the construction of $P_\calB$, $w_\calA$ must be informed of $\calB$ from its predecessor on $P_\calB$, so if this
	predecessor has not branched yet, it must also adopts $\calB$ before $\calA$.
Using the same argument, we know that $x$ must adopt $\calB$ first before adopting $\calA$.
Then we know that $x_\calB$ must be ordered before $x_\calA$ in $\pi_x$.
Now consider the scenario when $S_\calA = T \cup \{u\}$.
The argument is exactly the same as the previous case of $w_\calB$ ordered before $w_\calA$, and we also reach a contradiction.

We have exhausted all cases and Claim~\ref{clm:source} is proven.
\end{proof}

Claim~\ref{clm:source} immediately implies that $w_0 = u$, since otherwise, we have $w_0 \in T$, and
	by the proof of monotonicity (Theorem~\ref{thm:monotone}) we know that $v \in \Phi_\calA^W(\{w_0\}) \subseteq \Phi_\calA^W(T)$.
This contradicts with the definition of $v \in \Phi_\calA^W(T \cup \{u\}) \setminus \Phi_\calA^W(T)$.

Again by monotonicity, we have $v \in \Phi_\calA^W(\{u\}) \subseteq \Phi_\calA^W(S \cup \{u\})$.
Since $S \subseteq T$ and $v \not\in \Phi_\calA^W(T)$, then $v \not\in \Phi_\calA^W(S)$.
This gives $v \in \Phi_\calA^W(S \cup \{u\}) \setminus \Phi_\calA^W(S)$, which was to be shown.
\end{proof}

\end{document}